\documentclass{lmcs} 
\pdfoutput=1

\usepackage[utf8]{inputenc}

\usepackage{lastpage}
\lmcsdoi{18}{2}{22}
\lmcsheading{}{\pageref{LastPage}}{}{}%
{May~25,~2021}{Jun.~28,~2022}{}


\usepackage{hyperref}
\usepackage{amssymb}
\usepackage{subcaption}
\usepackage{tikz}
\usetikzlibrary{arrows,automata,positioning}
\usepackage{graphicx}
\usepackage{mathrsfs}
\usepackage{etoolbox}
\usepackage{xifthen}

\let\phi\varphi

\newcommand{\s}{\Sigma}

\newcommand{\alphloc}[1][\s]{(#1, \loc)}
\newcommand{\tracesloc}[1][]{%
	\ifthenelse{\isempty{#1}}{\mathrm{Tr}{\alphloc}}{\mathrm{Tr}{\alphloc[#1]}}%
}
\newcommand{\traces}[1][]{%
	\ifthenelse{\isempty{#1}}{\mathrm{Tr}(\s)}{\mathrm{Tr}(#1)}%
}
\newcommand{\loctimes}{\times_\ell}

\newcommand{\pset}{\mathscr{P}}
\newcommand{\olddalphabet}{\widetilde{\alphabet}}
\newcommand{\dalphabet}{\alphloc}
\newcommand{\Pialphabet}{\alphloc[\Pi]}
\newcommand{\loc}{\mathrm{loc}}
\newcommand{\ind}{I}
\newcommand{\dep}{D}
\newcommand{\alphabet}{\Sigma}
\newcommand{\isucc}{\lessdot}
\newcommand{\dcset}[1]{{\downarrow}#1}
\newcommand{\Pitraces}{\traces[\Pi]}

\newcommand{\trans}{\mathcal{F}}
\newcommand{\ls}{S}

\newcommand{\gs}{\ls}
\newcommand{\gf}{\ls_{\text{fin}}}
\newcommand{\lt}{\delta}
\newcommand{\gt}{\Delta}
\newcommand{\aaa}{A}

\newcommand{\atma}{(\{\ls_i\},M)}
\newcommand{\atmb}{(\{Q_i\},N)}
\newcommand{\wpis}{X \times Y}
\newcommand{\tralphabet}{\s\loctimes\gs}
\newcommand{\dtralphabet}{\alphloc[\tralphabet]}
\newcommand{\tracesdtr}{\traces[\tralphabet]}
\newcommand{\divides}{\prec}
\newcommand{\wt}{\widetilde}

\newcommand{\lc}{\circ_{\ell}}
\newcommand{\gtralphabet}{\alphabet\times\gs}
\newcommand{\dgtralphabet}{\alphloc[\gtralphabet]}
\newcommand{\tracesdgtr}{\traces[\gtralphabet]}
\newcommand{\gossip}{\mathcal{G}}

\newcommand{\view}[2]{{\downarrow}^{#1}(#2)}

\newcommand{\globalstate}{\mathsf{globalstate}}
\newcommand{\globalstates}[1]{\mathrm{gs}(#1)}

\newcommand{\Aseq}{A_{\mathrm{seq}}}
\newcommand{\fseq}{f_{\mathrm{seq}}}
\newcommand{\hatAseq}{\hat{A}_{\mathrm{seq}}}
\newcommand{\Bseq}{B_{\mathrm{seq}}}
\newcommand{\Cseq}{C_{\mathrm{seq}}}

\newcommand{\sinit}{s_\mathrm{in}} 
\newcommand{\qinit}{q_\mathrm{in}} 
\newcommand{\Da}{{\Downarrow}}
\def\inv{^{-1}}
\def\wras{\wr_{\text{as}}}

\newcommand{\Y}{\mathop{\mathsf{Y}\vphantom{a}}\nolimits}
\newcommand{\Si}{\mathbin{\mathsf{S}}}

\newcommand{\SSi}{\mathbin{\mathbb{S}}}
\newcommand{\Yleq}[2]{\Y_{#1}\leq\Y_{#2}}
\newcommand{\Yeq}[2]{(\Y_{#1}=\Y_{#2})}
\newcommand{\Ylt}[2]{(\Y_{#1}<\Y_{#2})}
\newcommand{\loctlf}{\mathsf{LocTL}[\Yleq{i}{j},\Y_i,\Si_i]}
\newcommand{\loctlv}{\mathsf{LocTL}[\Yleq{i}{j},\Si_i]}
\newcommand{\loctl}{\mathsf{LocTL}[\Y_i,\Si_i]}
\newcommand{\sprtl}{\mathsf{LocTL}[\Si_i]}

\newcommand{\fc}{\cdot}

\newcommand{\lab}{\Gamma}
\newcommand{\lalphabet}{\alphabet\times\lab}
\newcommand{\dlalphabet}{\alphloc[\lalphabet]}
\newcommand{\ltraces}{\traces[\lalphabet]}
\newcommand{\lcr}{\lc^{r}}
\newcommand{\thetaY}{\theta^{\Y}}
\newcommand{\muY}{\mu^{\Y}}
\newcommand{\nuY}{\nu^{\Y}}

\newcommand{\gslf}{\zeta}

\begin{document}

\title[Asynchronous wreath products for concurrent behaviours]{Asynchronous wreath product
and cascade decompositions for concurrent behaviours}

\author[B.~Adsul]{Bharat Adsul\lmcsorcid{0000-0002-0292-6670}}[a]	

\author[P.~Gastin]{Paul Gastin\lmcsorcid{0000-0002-1313-7722}}[b,d]	

\author[S.~Sarkar]{Saptarshi Sarkar\lmcsorcid{0000-0001-6989-6050}}[a]	

\author[P.~Weil]{Pascal Weil\lmcsorcid{0000-0003-2039-5460}}[c,d]	
\address{Indian Institute of Technology Bombay, Mumbai, India}	
\email{adsul@cse.iitb.ac.in, sapta@cse.iitb.ac.in}  

\address{Université Paris-Saclay, ENS Paris-Saclay, CNRS, LMF, 91190, Gif-sur-Yvette,
France}
\email{paul.gastin@ens-paris-saclay.fr}  
\address{Univ. Bordeaux, LaBRI, CNRS UMR 5800, F-33400 Talence, France}
\email{pascal.weil@labri.fr}

\address{CNRS, ReLaX, IRL 2000, Siruseri, India}	



\begin{abstract}
  We develop new algebraic tools to reason about concurrent behaviours modelled as
  languages of Mazurkiewicz traces and asynchronous automata.  These tools reflect the
  distributed nature of traces and the underlying causality and concurrency between
  events, and can be said to support true concurrency.  They generalize the tools that
  have been so efficient in understanding, classifying and reasoning about word languages.
  In particular, we introduce an asynchronous version of the wreath product operation and
  we describe the trace languages recognized by such products (the so-called asynchronous
  wreath product principle).  We then propose a decomposition result for recognizable
  trace languages, analogous to the Krohn-Rhodes theorem, and we prove this decomposition
  result in the special case of acyclic architectures.  Finally, we introduce and analyze
  two distributed automata-theoretic operations.  One, the local cascade product, is a
  direct implementation of the asynchronous wreath product operation. The
  other, global cascade sequences, although conceptually and operationally similar to the 
  local cascade product, translates to a more complex asynchronous implementation which
  uses the gossip automaton of Mukund and Sohoni. This leads to interesting applications 
  to the characterization of trace
  languages definable in first-order logic: they are accepted by a restricted local
  cascade product of the gossip automaton and 2-state asynchronous reset automata, and
  also by a global cascade sequence of 2-state asynchronous reset automata.  Over
  distributed alphabets for which the asynchronous Krohn-Rhodes theorem holds, a local
  cascade product of such automata is sufficient and this, in turn, leads to the
  identification of a simple temporal logic which is expressively complete for such
  alphabets.
\end{abstract}

\maketitle

\section{Introduction}\label{sec:intro}

Algebraic automata theory, that is, the use of algebraic tools and notions such as
monoids, morphisms and varieties, has been very successful in classifying recognizable
word languages, from both the theoretical and the algorithmic points of view, offering
structural descriptions of, say, logically defined classes of languages and decision
algorithms for membership in these classes
\cite{Eilenberg1976AutomataLA,str_cirBook,MPRI-notes}.  The purpose of this paper is to
extend some of this approach to the concurrent setting.  We are particularly interested in
decomposition results, in the spirit of the Krohn-Rhodes theorem, and their applications
to the study of first-order definable trace languages.

Let us first specify our model of concurrency.  Words represent sequential behaviours: a
sequence of letters models a sequence of events, occurring on a single process.  In a
concurrent setting involving multiple processes, we work with the well established
(Mazurkiewicz) traces \cite{Mazurkiewicz_1977, traces-book}: a trace represents a
concurrent behaviour as a labelled partial order which captures the distribution of events
across processes, as well as causality and concurrency between them.

The notion of a recognizable trace language is also very well established: a set of traces
is recognizable if the set of all the words representing these traces is a regular
language.  A key contribution, due to Zielonka, is the description of an
automata-theoretic model for the acceptance of recognizable trace languages, namely
asynchronous automata \cite{zielonka1987notes}.  These automata, with their local state
sets (one for each process), are natural distributed devices, which run on input traces in
a distributed fashion, respecting the underlying causality and concurrency between events.
More precisely, when working on an event during a run on an input trace, an asynchronous
automaton updates only the local states of the processes participating in that event; the
other processes remain oblivious to the occurrence of this event.  Zielonka's theorem
states that every recognizable trace language is accepted by an asynchronous automaton.

Early results seemed to indicate that the algebraic approach could be neatly transferred
to recognizable trace languages, with the monoid-theoretic definition of recognizability
matching the operational model of asynchronous automata (this is Zielonka's theorem
mentioned above \cite{zielonka1987notes}) and the characterization of star-free and
first-order definable trace languages (Guaiana \textit{et al.} \cite{guaiana1992star},
Ebinger and Muscholl \cite{ebinger1996logical}) in terms of aperiodic monoids.  Very few
significant results in this direction emerged since, and especially no strong Krohn-Rhodes
like decomposition results\footnote{An exception may be Guaiana et al.'s
~\cite{GUAIANA1998277}, which gives a wreath product principle for traces.  This work
however uses non-trace structures (structures that ignore the distributed nature of the
alphabet) to circumvent technical difficulties, thus limiting its relevance.}.

There are indeed deep, technical obstacles to this approach, which are discussed in some
detail in \cite{thesis, versionone}.
Our first contribution is the introduction of better suited notions of
\emph{asynchronous transformation monoids}, \emph{asynchronous morphisms} and
\emph{asynchronous wreath products}.  Because these notions closely adhere to the
distributed nature of traces, we obtain important results, such as a so-called
\emph{asynchronous wreath product principle}, describing the languages recognized by an
asynchronous wreath product of asynchronous transformation monoids.  Moreover, just as
(ordinary) transformation monoids model DFAs and their wreath product models the cascade
product of DFAs, our asynchronous wreath product of asynchronous transformation monoids
can be implemented as a \emph{local cascade product} of asynchronous automata.  The local
cascade product, in its purely automata-theoretic form, appeared in Adsul and Sohoni
\cite{DBLP:conf/fsttcs/AdsulS04, adsulthesis} in a characterization of first-order
definable trace languages in terms of asynchronous automata.

The next question is that of the possibility of a Krohn-Rhodes theorem in this setting.
The classical (sequential) Krohn-Rhodes theorem states that every morphism from the free
monoid $\Sigma^*$ to a transformation monoid is simulated (see
Definition~{\ref{wordsimulation}} for a formal definition of simulation) by a morphism to
the wreath product of particular simple transformation monoids: copies of the 2-state
reset transformation monoids and transformation groups based on the simple groups dividing
the given monoid of transformations.  This has many applications, in particular to the
description of simple automata (cascade products of reset automata) accepting all
first-order definable languages, see \cite{str_cirBook}.

The asynchronous analogue would state that any morphism from a trace monoid to a
transformation monoid is simulated by an asynchronous morphism to an asynchronous wreath
product of localized reset automata and transformation groups (where localization means
that all non-trivial actions take place on a single designated process).  We cannot
conclude, at this stage, that every distributed alphabet has this property, which we call
the asynchronous Krohn-Rhodes property (aKR), but we identify a large class of alphabets
for which it holds, namely those where the communication graph between processes is
acyclic.  The question is raised of which distributed alphabets have the aKR property.

Then we focus on first-order definable trace languages, which we know are the trace
languages recognized by aperiodic transformation monoids \cite{ebinger1996logical}.  Thus,
over aKR distributed alphabets (\textit{e.g.}, acyclic architectures), these languages are
exactly those that are recognized by an asynchronous wreath product of localized reset
asynchronous transformation monoids (or, equivalently, by a local cascade product of
localized reset asynchronous automata).
We show that, over arbitrary distributed alphabets, the trace languages recognized by such
a product are exactly those which are definable in a natural local temporal logic using
`process-based strict since' operators as its only modalities.  This logic, $\sprtl$, is
closely related to an expressively complete logic introduced by Diekert and Gastin
\cite{DBLP:journals/iandc/DiekertG06}.  Our result implies that $\sprtl$ has the same
expressive power as first-order logic over aKR distributed alphabets.

Over arbitrary alphabets again, we show that all first-order definable languages are accepted by
a local cascade product of the gossip automaton (introduced by Mukund and Sohoni
\cite{DBLP:journals/dc/MukundS97}), followed by localized reset asynchronous automata.

We know from Adsul and Sohoni \cite{DBLP:conf/fsttcs/AdsulS04} that gossip asynchronous
automata exhibit non-aperiodic behaviour, which is contrary to what we expect when
discussing first-order definable languages.  In order to avoid this situation, we
introduce two new operations.  One is the \emph{restricted local cascade product with the
gossip automaton}, which exploits the constants of an expressively complete logic called
$\loctlv$, and strongly encapsulates the non-aperiodic behaviour of the gossip
automaton.  We show that the first order definable trace languages are exactly the trace
languages accepted by a restricted cascade product of the gossip automaton with a local
cascade product of localized reset asynchronous automata, thus adding a new
characterization for this important class.

The second of these operations is the \emph{global cascade sequence} of asynchronous
automata.  Such a device is not properly an asynchronous automaton itself, but it
composes, in a novel way, the runs of the automata in the sequence to produce an
acceptance mechanism.  We exploit another expressively complete extension of $\sprtl$
called $\loctl$ to show that the first order definable trace languages are exactly the
trace languages accepted by global cascade sequences of localized reset asynchronous
automata, yet another characterization of this language class.

Independently of this characterization, we discuss how global cascade sequences can be
implemented by an asynchronous automaton, in the form of a restricted local cascade
product of the gossip automaton with the \emph{global-state detectors} for the factors of
the sequence.

\paragraph{Overview}
Before we dive into the body of the paper, and for the benefit of readers already fluent
in the vocabulary of traces and asynchronous automata, it is worth discussing in a little
more detail some of the inner workings of these results, especially around the notions of
accepting vs.\ computing (or relabelling) devices, and of sequential or asynchronous
transducer.

The classical \emph{sequential transducer} $\sigma_\aaa$ associated with a DFA $\aaa$ (or
with a morphism from a free monoid to a transformation monoid) maps each word $w$ to a
word of equal length, obtained by replacing letter $a$ in position $i$ by the pair
$(a,q)$, where $q$ is the state reached after reading the length $(i-1)$ prefix of $w$
(starting at the initial state of $\aaa$).  If we see the automaton $\aaa$ not as an
acceptor but as a computing device, outputting on every transition the full information it
has, \textit{i.e.}, the name of that transition (the pair composed of the label and the
source state of the transition), then $\sigma_\aaa(w)$ is the output of $A$ on
input $w$.  It is reasonable to view $\sigma_\aaa $ as the most general function computed
by $\aaa$.  Moreover, the cascade product of DFAs $\aaa$ and $B$ can be described as the
operation of $\aaa$ on an input word $w$, followed by the operation of $B$ on input
$\sigma_\aaa(w)$.  In particular, the sequential transducer of the cascade product of
$\aaa$ and $B$ is the composition $\sigma_B\circ \sigma_A$.

This idea of decorating each position in a word with information computed by an automaton
carries over nicely to the distributed setting: we can decorate every event in a trace $t$
with information computed by an asynchronous automaton $\aaa$, thus viewing $\aaa$ as an
asynchronous computing device.  However, the distinction between local and global states
in asynchronous automata leads to two distinct approaches, both discussed in this paper.

One relies on local information.  More precisely, we define the \emph{asynchronous
transducer} $\chi_\aaa$, which maps a trace $t$ to a trace with the same underlying poset,
where the label of an event $e$, say $a$, is replaced with the pair $(a,s_a)$ as follows.
If $t_e$ is the prefix of $t$ which consists of all events strictly less than $e$ (that
is, the strict causal past of $e$), $s_a$ is the collection of local states reached after
reading $t_e$ at all the processes which participate in $a$.  Again, one can see
$\chi_\aaa$ as the most general function computed by $\aaa$ when considering only the
local states.  With this notion in hand, stating and proving the asynchronous wreath
product principle, which describes the languages recognized by an asynchronous wreath
product of asynchronous transformation monoids, while technically more demanding, proceeds
along the same lines as in the sequential case.  Also, the local cascade product of
asynchronous automata $\aaa \lc B$ can be viewed as the operation of $\aaa$, running on
trace $t$ and outputting trace $\chi_\aaa(t)$, followed by automaton $B$ running on
$\chi_\aaa(t)$.  Here again, the asynchronous transducer of the local cascade product
$\aaa\lc B$ is $\chi_B \circ \chi_A$. 

Global cascade sequences are defined in the same spirit, using global rather than local
state information.  With the same notation as above, the \emph{global-state labelling
function} $\zeta_\aaa$ replaces the label $a$ of event $e$ with the pair $(a,s)$ where $s$
is the global state reached after reading $t_e$.  This is a different function that can be
said to be globally computed by $\aaa$, and the definition of global cascade sequences
corresponds exactly to the composition of these global-state labelling functions.

Both the asynchronous transducer and the global-state labelling function are asynchronous
computing devices but the latter carries more information.  This is reflected in the
different ways these functions can be implemented.  The composition of the asynchronous
transducers of (appropriately defined) asynchronous automata, as we explained, translates
directly to the local cascade product of these automata.  In contrast, an additional
ingredient is required to implement a global cascade sequence by means of an asynchronous
automaton, that is, to make the global information carried by $\zeta_\aaa$ known to local
states.  This ingredient is the restricted cascade product with the gossip automaton.

\paragraph{Organization}
We now briefly describe the organization of the paper.  Section~\ref{sec:prelim}
summarizes the necessary notions on distributed alphabets, traces, transformation monoids,
recognizable trace languages and asynchronous automata --- including the statement of
Zielonka's theorem.  

We begin Section~\ref{sec:alg} with a brief account of the wreath
product operation on transformation monoids, the notion of
simulation and the Krohn-Rhodes theorem. 
We introduce our notions of asynchronous transformation monoids
(Section~\ref{sec: atm}) --- directly in the spirit of Zielonka's asynchronous automata
---, asynchronous morphisms (Section~\ref{sec: asychm}) and asynchronous wreath products
(Section~\ref{section: asynchronous wreath}).  Asynchronous transducers are used in
Section~\ref{sec:wpp} to state and prove the asynchronous wreath product principle
(Theorem~\ref{thm:wpp2}).  The question of the implementation of the asynchronous wreath
product, in the form of the local cascade product of asynchronous automata, is discussed
in Section~\ref{sec: local cascade}.  Finally, in Section~\ref{sec: decomposition}, we
formulate and briefly discuss the asynchronous Krohn-Rhodes property of distributed
alphabets.

Section~\ref{sec: acyclic} is entirely dedicated to the proof that the asynchronous
Krohn-Rhodes property holds for distributed alphabets over an acyclic architecture
(Theorem~\ref{thm:acyclickr}).  The next section explores the applications of the notions
of asynchronous wreath product and local cascade product to the special case of
first-order definable languages: the characterization of the trace languages recognized by
asynchronous wreath products of localized resets by $\sprtl$ (Theorem~\ref{thm:lcascade}),
which therefore is expressively complete over aKR distributed alphabets
(Theorem~\ref{sprtl-complete}); and the characterization of the first-order definable
trace languages by means of the restricted local cascade product of the gossip automaton
with local cascade product of localized asynchronous reset automata (Theorem~\ref{thm:lcascade-loctlv}).
This characterization is obtained using a new local temporal logic $\loctlv$ which 
is proved to be expressively complete in Theorem~\ref{thm:loctlfo}.

The notions of global-state labelling function associated with an asynchronous automaton,
and of global cascade sequences as computing and accepting devices are introduced in
Section~\ref{sec:cascade}, where we prove the characterization of first-order definable
languages in terms of global cascade sequences of asynchronous localized reset automata
(Theorem~\ref{thm: gcs FO}).

Section~\ref{sec: implementation global cascade} completes the operational point of view
on the latter results by presenting an asynchronous implementation of global cascade
sequences (Theorems~\ref{thm:gat-correct} and~\ref{thm:gatseq-correct}).  Finally,
Section~\ref{sec: conclusion} outlines an intriguing question left open by this work.

The results in this paper are an elaboration and an extension of those presented at CONCUR
2020 \cite{2020:AdsulGastinSarkar}.  Proofs of several key results e.g.
Theorem~\ref{thm:acyclickr}, Theorem~\ref{thm:lcascade} and Theorem~\ref{thm: gcs FO} are
unavailable in the conference proceedings and are included here.  The complete
Section~\ref{sec: implementation global cascade} discussing implementation of a global
cascade sequence as an asynchronous automaton is a technical addition that is briefly
alluded to in our conference paper without details.  Finally, Section~\ref{sec:aperiodic}
includes significant new technical contributions that originate from introducing temporal
logic constants $\Yleq{i}{j}$ in the logic under consideration.  This leads to a new
characterization of first-order definable trace languages in
Theorem~\ref{thm:lcascade-loctlv}. 

\paragraph{Acknowledgement}
The collaboration between the authors was supported
by IRL 2000, ReLaX, CNRS.

\section{Preliminaries}\label{sec:prelim}

\subsection{Basic notions in trace theory}\label{section:traces} 

Let $\pset$ be a finite set of agents/processes.
If $\pset$ is clear from the
context, we use the simpler notation $\{X_i\}$ to denote a $\pset$-indexed family
${\{X_i\}}_{i \in \pset}$.

A \emph{distributed alphabet} over $\pset$ is a family $\olddalphabet = \{\alphabet_i\}$,
where the $\alphabet_i$ are non-empty finite sets that may overlap one another. Let
$\alphabet = \bigcup_{i \in \pset} \alphabet_i$. The \emph{location function} $\loc \colon
\alphabet \to 2^\pset$ is defined by setting $\loc(a) = \{i \in \pset \mid a \in
\alphabet_i\}$. Note that from $\alphabet$ and $\loc$, one can
reconstruct the distributed alphabet, and hence we also use the notation $\alphloc$ for
distributed alphabet. The corresponding \emph{trace alphabet} is the pair $(\alphabet,
\ind)$, where $\ind$ is the \emph{independence relation} $I = \{(a,b) \in \alphabet^2
\mid \loc(a) \cap \loc(b) = \emptyset \}$ induced by $\alphloc$. The corresponding
\emph{dependence relation} is $\dep = \alphabet^2 \setminus \ind$.

A $\alphabet$-labelled poset is a structure $t = (E, \leq, \lambda)$ where $E$ is a set,
$\leq$ is a partial order on $E$ and $\lambda \colon E \to \alphabet$ is a labelling
function.  For $e, e' \in E$, define $e \isucc e'$ if and only if $e < e'$ and for each
$e''$ with $e \leq e'' \leq e'$ either $e = e''$ or $e' = e''$. For $X \subseteq E$, let
$\dcset{X} = \{y \in E  \mid  y \leq x \text{ for some } x \in X \}$. For $e \in E$, we
abbreviate $\dcset{\{e\}}$ by simply $\dcset e$. We also write $\Da e=\dcset 
e\setminus\{e\}$ for the \emph{strict past} of $e$. 

A \emph{trace} over $\alphloc$ is a finite 
$\alphabet$-labelled poset $t = (E, \leq, \lambda)$ such that
\begin{itemize}
    	\item If $e, e' \in E$ with $e \isucc e'$ then $(\lambda(e), \lambda(e')) \in D$

	\item If $e, e' \in E$ with $(\lambda(e), \lambda(e')) \in D$, then $e \leq
        e'$ or $e' \leq e$
\end{itemize}

Let $\tracesloc$ denote the set of all traces over $\alphloc$; if $\loc$ is clear from the
context, we simply use the notation $\traces$. Henceforth a trace means a
trace over $\alphloc$ unless specified otherwise.  Now we turn our attention to the
important operation of concatenation of traces.  Let $t = (E, \leq, \lambda) \in \traces$
and $t' = (E', \leq' , \lambda') \in \traces$. Without loss of generality, we can assume
$E$ and $E'$ to be disjoint. We define $tt' \in \traces$ to be the trace $(E'', \leq'',
\lambda'')$ where 

\begin{itemize}
	\item $E'' = E \cup E'$,
	\item $\leq''$ is the transitive closure of ${\leq} \cup {\leq'} \cup \{(e,e') \in
			E \times E' \mid (\lambda(e), \lambda'(e')) \in \dep \}$,
	\item $\lambda''\colon E'' \to \alphabet$ where $\lambda''(e) = \lambda(e)$ if $e
		\in E$; otherwise, $\lambda''(e) = \lambda'(e)$.
\end{itemize}

This operation, henceforth referred to as \emph{trace concatenation}, gives $\traces$ a
monoid structure.  A trace $t'$ is said to be a \emph{prefix} of a trace $t$
if there exists $t''$ such that $t=t't''$.

Observe that, with $a$ (resp.\ $b$) denoting the singleton trace whose
only event is labelled $a$ (resp.  $b$), if $(a,b) \in I$ then $ab=ba$ in $\traces$.
A basic result in trace theory gives a presentation of the trace monoid as a quotient of
the \emph{free} word monoid $\Sigma^*$ by the congruence ${\sim_I} \subseteq \Sigma^*
\times \Sigma^*$ generated by $ab \sim_I ba$ for $(a,b) \in I$.
See~\cite{traces-book} for more details.

\begin{prop}\label{basicprop} The canonical morphism from
$\Sigma^* \to \traces$, sending a letter $a \in \Sigma$ to the trace $a$, factors through 
the quotient monoid $\Sigma^*/{\sim_I}$  and induces an isomorphism from 
$\Sigma^*/{\sim_I}$ to the trace monoid $\traces$.
\end{prop}

Let $t = (E, \leq, \lambda) \in
\traces$.  The elements of $E$ are referred to as \emph{events} in $t$ and for an event
$e$ in $t$, $\loc(e)$ abbreviates $\loc(\lambda(e))$. Further, let $i \in \pset$. The set
of $i$-events in $t$ is $E_i = \{e \in E  \mid i \in \loc(e)\}$. This is the set of events in
which process $i$ participates. It is clear that $E_i$ is totally ordered by $\leq$.

Note that, if we restrict the trace $t$ to a downward-closed subset of events
$c=\dcset{c}$, we get another trace $(c,\leq,\lambda)$ which is a prefix of $t$. In fact, every prefix of $t$ arises this way, and we often
identify prefixes with downward-closed sets of events.  Examples of prefixes are defined
by the empty set, $E$ itself and, more importantly, $\dcset{e}$ or $\Da e$, for every
event $e \in E$.

\subsection{Recognizable trace languages}\label{section: recog}

A map from a set $X$ to itself is called a \emph{transformation} of $X$.  The set
$\trans(X)$ of all such transformations forms a monoid under function composition: $fg=g\circ f$.
A \emph{transformation monoid} (or simply \emph{tm}) is a pair $T=(X,M)$ where $M$ is a
submonoid of $\trans(X)$. The tm $(X,M)$ is called \emph{finite} if $X$ is finite. 

\begin{exa}\label{ex:U2}
  Let $X = \{1,2\}$ and let $r_1, r_2$ be the constant maps on $X$, mapping each element
  to 1 and 2, respectively.  Then $M = \{\mathrm{id}_X, r_1, r_2\}$ is a monoid.  We
  denote the tm $(X, M)$ by $U_2$.
\end{exa}

\begin{exa}\label{ex:g}
  Let $M$ be a monoid.  For each $m\in M$, let $t_m$ be the transformation of $M$ defined
  by $t_m(x) = xm$ for all $x \in M$.  The map $m\mapsto t_m$ is an injective morphism
  from $M$ to $\trans(M)$, allowing us to view $M$ as a submonoid of $\trans(M)$.  We
  denote by $(M,M)$ the resulting tm.
\end{exa}

Let $N$ be a monoid and let $T = (X,M)$ be a tm.  By a morphism $\phi$ from $N$ to
$T$, we mean a (monoid) morphism $\phi \colon N \rightarrow M$.  We abuse the notation and
also write this as $\phi \colon N \rightarrow T$.

\begin{rem}\label{rem:tm-automata}
  Morphisms from free word monoids to transformation monoids almost correspond to
  deterministic and complete automata, the only difference being that an automaton also
  includes an initial state and a set of final states.  More precisely, let $\Sigma$ be a
  finite alphabet, $T=(X,M)$ be a tm and $\phi\colon\Sigma^{*}\to T$ be a morphism.  In
  the corresponding automaton, $X$ is the set of states and for each $a\in\Sigma$,
  $\phi(a)\in M\subseteq\trans(X)$ defines the deterministic and complete transition
  function for letter $a$.  Conversely, a deterministic and complete automaton $A$ over
  $\Sigma$, defines a tm $T=(X,M)$ where $X$ is the set of states of $A$, $M$ is the {\em
  transition monoid} of $A$ and a surjective morphism $\phi\colon\Sigma^{*}\to M$ as
  follows.  For each $a\in\Sigma$, let $\phi(a)\in\trans(X)$ be the transformation on $X$
  induced by the transition function for letter $a$ in $A$.  We obtain a morphism
  $\phi\colon\Sigma^{*}\to\trans(X)$ and we let $M=\phi(\Sigma^{*})$ be the submonoid of
  $\trans(X)$ generated by $\{\phi(a)\}_{a\in\Sigma}$.  For instance, the tm $U_2$ defined
  in Example~\ref{ex:U2} corresponds to the DFA in Figure~\ref{fig:u2 dfa} below, and the
  induced morphism is given by $\phi(a)=r_1$, $\phi(b)=r_2$ and $\phi(c)=\mathrm{id}_X$.
  \begin{figure}[hb]
    \centering
    \begin{tikzpicture}
[->,>=stealth',shorten >=1pt,auto,node distance=2cm,
                    semithick, initial text=]
\tikzstyle{every state}=[draw=black]

  \node[state] (1)                    {$1$};
  \node[state]         (2) [right=of 1] {$2$};

  \path (1) edge [bend left] node {$b$} (2)
	    edge [loop above] node {$a,c$}   (1)
        (2) edge [loop above] node {$b,c$} (2)
 	    edge [bend left]  node {$a$}   (1);
\end{tikzpicture}
    \caption{Automaton corresponding to the tm $U_2$}%
    \label{fig:u2 dfa}
  \end{figure}
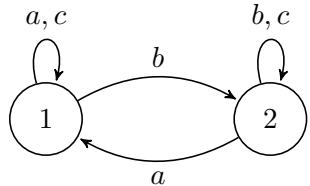
\end{rem}

We now fix a distributed alphabet $\alphloc$. 
Let $\phi \colon \traces \rightarrow M$ be a morphism to a monoid $M$.
We note that, if $(a,b) \in I$, then $ab=ba$ in
$\traces$ and hence $\phi(a)$ and $\phi(b)$ commute in $M$. In fact, in view of
Proposition~\ref{basicprop}, any function $\phi\colon  \Sigma \rightarrow M$ which has the
property that $\phi(a)$ and $\phi(b)$ commute for every $(a,b) \in I$, can be uniquely
extended to a morphism from $\traces$ to $M$.

Transformation monoids can be naturally used to recognize trace languages.  We say that a
trace language $L \subseteq \traces$ is \emph{recognized by} a tm $T=(X,M)$ if there
exists a morphism $\phi\colon \traces \rightarrow T$, an \emph{initial} element
$x_{\textrm{in}} \in X$ and a \emph{final} subset $X_{\textrm{fin}} \subseteq X$ such that
$L = \{t \in \traces \mid \phi(t)(x_{\textrm{in}}) \in X_{\textrm{fin}}\}$.  A trace
language is said to be \emph{recognizable} if it is recognized by a finite tm (see
\cite{traces-book,combinatorics-traces}).  

\begin{rem}\label{rem:diamond}
	A morphism $\varphi \colon \traces \to (X,M)$ from the trace monoid into a tm can
	also be viewed, by Proposition~\ref{basicprop}, as a morphism from the free monoid
	$\Sigma^*$ with the additional property that $\varphi(a)$ and $\varphi(b)$ commute
	for every $(a,b) \in I$. The automaton corresponding to $\varphi \colon \Sigma^*
	\to (X,M)$ (cf.\ Remark~\ref{rem:tm-automata}) thus has the property that $ab$ and
	$ba$ have the same state transitions for all $(a,b) \in I$.  An automaton with
	this property is called a \emph{diamond-automaton}.  It is a classical sequential automata
	model for recognizing trace languages since all linearizations of a trace reach the
	same destination state starting from any given state.  For instance, the DFA in
	Figure~\ref{fig:u2 dfa} is a diamond-automaton for the distributed alphabet $\{
	\Sigma_1 = \{a,c\}, \Sigma_2 = \{a,b\} \}$.  Here $(b,c) \in I$ and indeed, the
	words $bc$ and $cb$ have identical state transitions in the DFA. Note however that the
	runs of the DFA on the two linearizations are different, namely, $1 \to 1 \to 2$
	for $cb$ and $1 \to 2 \to 2$ for $bc$.
\end{rem}

\subsection{Asynchronous automata}\label{sec: asynchronous automata}

Recognizability of trace languages can also be seen as an automata-theoretic notion that
is concurrent in nature.  In
the upcoming definition of an asynchronous automaton, the set of states is structured as a
$\pset$-indexed family of finite non-empty sets $\{\ls_i\}_{i \in \pset}$.  The elements
of $\ls_i$ are called the \emph{local $i$-states}, or the \emph{local states} of process
$i$.  If $P$ is a non-empty subset of $\pset$,
a \emph{$P$-state} is a map $s\colon P \to \bigcup_{i\in\pset}\ls_i$ such that $s(j) \in \ls_j$ for every $j
\in P$.  We denote by $\ls_P$ the set of all $P$-states and we call $\gs=\ls_{\pset}$ 
the set of
\emph{global states}.\footnote{Note that we can naturally identify
$\ls_P$ with $\prod_{i \in P} S_i$  and $\gs$ with $\prod_{i \in \pset} S_i$.}

If $P' \subseteq P$ and $s \in \ls_P$ then $s_{P'}$ denotes the restriction of $s$ to
$P'$.  We use the shorthand ${-\!P}$ to indicate the complement of $P$ in $\pset$.  We
sometimes split a global state $s \in \gs$ as $(s_P, s_{-\!P}) \in \ls_P \times
\ls_{-\!P}$.  If $a\in \alphabet$, we talk about \emph{$a$-states} to mean
$\loc(a)$-states and we write $S_a$ for $S_{\loc(a)}$.  If $a\in\alphabet$, $\loc(a)
\subseteq P$ and $s$ is a $P$-state, we write $s_a$ for $s_{\loc(a)}$.

Finally, we use the following notion of extension.  If $P \subseteq \pset$ and $f$ is a
transformation of $S_{P}$, the \emph{extension} of $f$ to $\gs$ is the transformation
$g\in \trans(\gs)$ such that $(g(s))_P = f(s_P)$ and $(g(s))_{-\!P} = s_{-\!P}$.  In other
words, if $s=(s_P, s_{-\!P}) \in \gs$, then $g((s_P, s_{-\!P}))= (f(s_P), s_{-\!P})$.  We
observe that $f$ is entirely determined by $g$ and $P$.  Extensions of transformations of
$S_P$ are called \emph{$P$-maps} (see Section~\ref{sec: atm} for more details on
$P$-maps).

Asynchronous automata were introduced by Zielonka for concurrent computation on traces
\cite{zielonka1987notes}.  An \emph{asynchronous automaton} $\aaa$ over $\alphloc$ is a
tuple $({\{\ls_i\}}_{i \in \pset},  {\{\lt_a\}}_{a \in \alphabet}, \sinit)$ where
\begin{itemize}
  \item $\ls_i$ is a finite non-empty set of local $i$-states for each process $i$;

  \item For $a \in \alphabet$, $\lt_a \colon \ls_a \to \ls_a$ is a (complete) transition
  function on $a$-states;

  \item $\sinit \in \gs$ is an initial global state.
\end{itemize}

Similar to $\pset$-indexed families, we will follow the convention of writing $\{Y_a\}$ to
denote the $\alphabet$-indexed family ${\{Y_a\}}_{a \in \alphabet}$.

Observe that an $a$-transition of $\aaa$ reads and updates only the local states of the
agents which participate in $a$.  As a result, actions which involve disjoint sets of
agents are processed concurrently by $\aaa$.  For $a \in \alphabet$, let $\gt_a \colon
\gs \to \gs$ be the extension of $\lt_a \colon  \ls_a \to \ls_a$. Clearly, if $(a,b) \in
I$ then $\Delta_a$ and $\Delta_b$ commute.  Hence, the (global) transition functions 
$\{\gt_a\}$ induce a trace morphism $t\mapsto\Delta_t$ from $\traces$ to $\trans(\gs)$.
We denote by $A(t)$ the global state reached when running $\aaa$ on $t$, $\aaa(t)=\gt_t(\sinit)$.

Let $L \subseteq \traces$ be a trace language. We say that $L$ is \emph{accepted} by $\aaa$
if there exists a subset $\gf \subseteq \gs$ of final global states such that $L = \{t 
\in \traces \mid \aaa(t) \in \gf \}$.

The fundamental theorem of Zielonka~\cite{zielonka1987notes} states that a trace language
is recognizable if and only if it is accepted by some asynchronous automaton
(see~\cite{mukund2012automata} for another proof of the theorem).

\paragraph{Asynchronous labelling functions.}
We will also use asynchronous automata to decorate events of a trace with extra
information computed by the automaton.  This is similar to the notion of \emph{locally
computable functions} defined in~\cite{DBLP:journals/dc/MukundS97}.  It extends to traces
the notion of sequential letter-to-letter word transducers.  When dealing with a trace
$t=(E,\leq,\lambda)$, we wish to preserve the underlying poset $(E,\leq)$ and enrich the
labels with extra information.

Formally, let $\alphloc$ be a distributed alphabet and $\lab$ be a finite set.  The
alphabet $\lalphabet$ can be equipped with a distributed structure over $\pset$
by letting $(\lalphabet)_{i} = \alphabet_i\times\lab$. As a result, the
location of $(a,\gamma) \in \lalphabet$ is simply $\loc(a)$; thus we
unambiguously reuse the notation $\loc$ for the location function of $\lalphabet$, and use
the notation $\ltraces$ to denote the set of traces
over this distributed alphabet. A $\lab$-labelling function is a map
$\theta\colon\traces\to \ltraces$ such that, for
$t=(E,\leq,\lambda)\in\traces$, we have $\theta(t)=(E,\leq,(\lambda,\mu))$, i.e., the map
$\theta$ adds a new label $\mu(e)\in\lab$ to each event $e$ in $t$.

For instance, given $i\in\pset$, we may consider the $\{0,1\}$-labelling function
$\theta_i$ which decorates each event $e$ of a trace $t$ with $\mu_i(e)=1$ if the strict
causal past of $e$ contains some event on process $i$, i.e., if $E_i\cap\Da
e\neq\emptyset$, and $\mu_i(e)=0$ otherwise.

An \emph{asynchronous (letter-to-letter) $\lab$-transducer} over $\alphloc$ is a tuple
$\hat{\aaa}=(\aaa,\{\mu_a\})$ where
$\aaa=(\{\ls_i\},\{\lt_a\},\sinit)$ is an asynchronous automaton and each $\mu_a$
($a\in\Sigma$) is a map $\mu_a\colon\ls_a\to\lab$.  We associate with $\hat{\aaa}$, a
$\Gamma$-labelling function, also denoted by $\hat{\aaa}$, $\hat{\aaa}\colon\traces\to
\traces[\s \times \Gamma]$ as follows: for $t=(E,\leq,\lambda)\in\traces$, we let
$\hat{\aaa}(t)=(E,\leq,(\lambda,\mu))$ where for all events $e\in E$ with $\lambda(e)=a$
and $s=\aaa(\Da e)$, we have $\mu(e)=\mu_a(s_a)$.

Given a $\Gamma$-labelling function $\theta$, 
we say that $\hat{\aaa}$ \emph{computes} (or \emph{implements}) $\theta$ if for every $t
\in \traces$, $\hat{\aaa}(t) =\theta (t)$.
We also say that an asynchronous automaton $\aaa=(\{\ls_i\},\{\lt_a\},\sinit)$
\emph{computes} $\theta$ if there are maps $\mu_a\colon\ls_a\to\lab$ such that
$\hat{\aaa}=(\{\ls_i\},\{\lt_a\},\sinit,\{\mu_a\})$ computes $\theta$.

Notice that a $\lab$-labelling function $\theta$ is defined on all input traces,
hence an asynchronous automaton which computes $\theta$ must be complete in the sense that
it admits a run on all traces from $\traces$ and it does not use an acceptance condition:
when considering an asynchronous automaton $\aaa$ which computes a $\lab$-labelling
function, we always assume that all global states of $\aaa$ are accepting.

For instance, the above $\{0,1\}$-labelling function $\theta_i$ can be computed by an
asynchronous $\{0,1\}$-transducer.  We need two states $\{0,1\}$ for each process.
Initially, all processes start in state $0$.  When the first event $e$ occurs on process
$i$, all processes in $\loc(e)$ switch to state $1$: for all $a\in\Sigma_i$, $\lt_a$ is
the constant map sending all states in $\ls_a$ to $(1,\ldots,1)$.  Then, the information
is propagated via synchronizing events: for all $b\in\Sigma\setminus\Sigma_i$, the map
$\lt_b$ sends $(0,\ldots,0)$ to itself and all other states to $(1,\ldots,1)$.  It is easy
to add output functions $\{\mu_a\}$ in order to compute $\theta_i$.

There exists a canonical (most general) function computed by an asynchronous automaton
$\aaa=(\{\ls_i\},\{\lt_a\},\sinit)$, called the \emph{asynchronous transducer of $\aaa$}
and denoted $\chi_\aaa$.  This function, which was already defined in
\cite{DBLP:conf/fsttcs/AdsulS04}, simply adds to an event $e$ the local state information
of $\aaa$ before executing $e$.  Formally, letting $\lab_\aaa=\bigcup_{a}\ls_a$, it is
implemented by taking each output function $\mu_a$ to be the identity function from
$\ls_a$ to $\lab_\aaa$.  Notice that, traces in $\chi_\aaa(\traces)$ have labels from
$\{(a,s_a)\mid a\in\alphabet, s_a\in\ls_a\}\subseteq\alphabet\times\lab_\aaa$. We denote by
$\alphabet \loctimes S$ the set $\{(a,s_a)\mid a\in\alphabet, s_a\in\ls_a\}$, and consider
it a distributed alphabet as it naturally inherits the location function of
$\alphabet \times \lab_A$, that is, $\loc((a,s_a)) = \loc(a)$.  Clearly, all
$\lab$-labelling functions computed by $\aaa$ are abstractions of $\chi_\aaa$.

\section{Asynchronous structures and decomposition problems}\label{sec:alg}
This section is devoted to the development of new algebraic asynchronous structures.  Our
main interest is in transferring from the algebraic theory of word languages to trace
languages the results and methods which rely on the wreath product operation: the wreath
product principle (see {\cite{str_cirBook}}) and the Krohn-Rhodes theorem (see
{\cite{Eilenberg1976AutomataLA}}). We present a new algebraic framework for the setting of
traces, and introduce an appropriate asynchronous wreath product operation; an
asynchronous wreath product principle is also established that works in the realm of
traces. This allows posing the question of a meaningful distributed analogue of the
Krohn-Rhodes theorem which we then partially resolve in the remainder of this article.
\subsection{Wreath product in sequential setting}\label{subsec:sequential}
We first recall the definitions of division, simulation and wreath products in the context
of transformation monoids, see \cite{Eilenberg1976AutomataLA}.

\begin{defi}\label{wordsimulation}
\begin{figure}[ht]
		\centering
		\begin{tikzpicture}
			\node  at (0,0) (sw) {$X$};
			\node  at (1.4,0) (se) {$X$};
			\node  at (0,1.4) (nw) {$Y$};
			\node  at (1.4,1.4) (ne) {$Y$};
			\draw[->] (sw.east) to node[below] {\small $\pi(n)$} (se.west); 
			\draw[->] (nw.east) to node[above] {\small $n$}(ne.west); 
			\draw[->] (nw.south) to node[left] {\small $f$} (sw.north); 
			\draw[->] (ne.south) to node[right] {\small $f$} (se.north); 
		\end{tikzpicture}
    \hfil
		\begin{tikzpicture}
			\node  at (0,0) (sw) {$X$};
			\node  at (1.4,0) (se) {$X$};
			\node  at (0,1.4) (nw) {$Y$};
			\node  at (1.4,1.4) (ne) {$Y$};
			\draw[->] (sw.east) to node[below] {\small $\phi(a)$} (se.west); 
			\draw[->] (nw.east) to node[above] {\small $\psi(a)$}(ne.west); 
			\draw[->] (nw.south) to node[left] {\small $f$} (sw.north); 
			\draw[->] (ne.south) to node[right] {\small $f$} (se.north); 
		\end{tikzpicture}
        \caption{Conditions $\pi(n)(f(y)) = f(n(y))$ and $\phi(a)(f(y)) = f(\psi(a)(y))$ 
        in Definition~\ref{wordsimulation}}
		\label{fig:commute-word-simulation}
\end{figure}
We say that a monoid $M$ \emph{divides} a monoid $N$ (written $M \divides N$) if there exists a
surjective morphism $\phi\colon N'\to M$, defined on a submonoid $N'$ of $N$.

We say that a tm $(X,M)$ \emph{divides}
a tm $(Y,N)$ (written $(X,M) \divides (Y,N)$) if there exists a surjective map $f \colon Y
\rightarrow X$ and a surjective morphism $\pi\colon N'\to M$ defined on a submonoid $N'$
of $N$, such that $\pi(n)(f(y)) = f(n(y))$ for all $n \in N'$ and all $y \in Y$, see
Figure~\ref{fig:commute-word-simulation} (left).

Finally, given morphisms $\phi\colon \Sigma^* \to T=(X,M)$ and $\psi\colon \Sigma^* \to
T'=(Y,N)$, we say that $\psi$ \emph{simulates} $\phi$ if there exists a surjective map
$f\colon Y \rightarrow X$ such that $\phi(a)(f(y)) = f(\psi(a)(y))$ for all $a \in
\Sigma$ and all $y\in Y$, see
Figure~\ref{fig:commute-word-simulation} (right).
\end{defi}

\begin{exa}\label{ex:simulation and division}
  \begin{figure}[ht]
    \centering
    \begin{tikzpicture}
[->,>=stealth',shorten >=1pt,auto,node distance=2cm,
                    semithick, initial text=]
\tikzstyle{every state}=[draw=black]
  \node[initial,state] (1)                    {$q_1$};
  \node[state]         (a) [above right of=1] {$q_a$};
  \node[state]         (b) [below right of=1] {$q_b$};
  \node[state]         (ab) [below right of=a] {$q_{ab}$};
\begin{scope}[node distance=0.8cm]
\node [below of=b] {Automaton $A$};
\end{scope}

  \path (1) edge              node {a} (a)
            edge              node [swap] {b} (b)
        (a) edge              node {b} (ab)
        (b) edge              node [swap]  {a} (ab);

  \node[initial,state] (1') [right=3cm of ab]       {$q'_1$};
  \node[state]         (a') [above right of=1'] {$q'_a$};
  \node[state]         (b') [below right of=1'] {$q'_b$};
  \node[state]         (ab') [right of=a'] {$q'_{ab}$};
  \node[state]         (ba') [right of=b']       {$q'_{ba}$};
\begin{scope}[node distance=0.8cm]
\node [below of=b'] {Automaton $B$};
\end{scope}
  \path (1') edge              node {a} (a')
            edge              node [swap] {b} (b')
        (a') edge              node {b} (ab')
        (b') edge              node [swap]  {a} (ba');
\end{tikzpicture}
    \caption{Automata $A$ and $B$ on alphabet $\Sigma=\{a,b\}$}
    \label{fig:pb1}
  \end{figure}
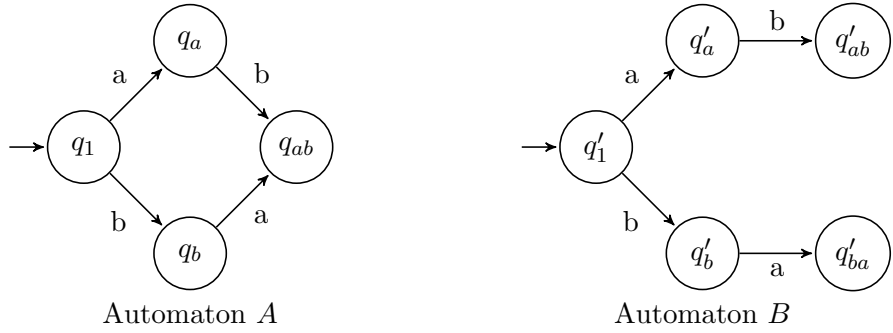
  Automata $A$ and $B$ in Figure~\ref{fig:pb1} define, respectively, morphisms
  $\phi\colon\Sigma^{*}\to T$ and $\psi\colon\Sigma^{*}\to T'$, to transformation monoids
  $T=(X,M)$ and $T'=(Y,N)$ (cf.\ Remark~\ref{rem:tm-automata}).  Note that to keep the
  pictures simple, the automata are not complete, all missing transitions go to an
  implicit sink state (state $s$ for automaton $A$ and $s'$ for automaton $B$).  In
  particular, $X=\{q_1,q_a,q_b,q_{ab},s\}$ and $M$ is generated by the transition
  functions $\phi(a)$ and $\phi(b)$ defined by $A$.  We can check that $\psi$ simulates
  $\phi$ using the surjective map $f\colon Y\to X$ defined by $q'_1\mapsto q_1,
  q'_a\mapsto q_a, q'_b\mapsto q_b, q'_{ab},q'_{ba}\mapsto q_{ab}, s'\mapsto s$.  The same
  map $f$, along with $\pi\colon N \to M$ as $\pi(\psi(w))=\phi(w)$, shows that $T$
  divides $T'$.
\end{exa}

\begin{rem} \label{rem: division-simulation-equiv}
  It is an elementary verification that, if a morphism $\psi\colon \Sigma^* \to T'=(Y,N)$
  simulates a morphism $\phi\colon \Sigma^* \to T=(X,M)$, then any language recognized by
  $\phi$ is also recognized by $\psi$.  Moreover, if $\phi$ is onto $M$, then $(X,M)
  \divides (Y,N)$: this can be checked using $N'=\psi(\Sigma^{*})$ and
  $\pi(\psi(w))=\phi(w)$.

  Similarly, if $(X,M) \divides (Y,N)$ and $\phi\colon \Sigma^* \to T=(X,M)$ is a
  morphism, then there exists a morphism $\psi\colon \Sigma^* \to T'=(Y,N)$ which
  simulates $\phi$: it suffices to choose an arbitrary element $\psi(a)$ in
  $\pi^{-1}(\phi(a))$ for every $a\in\Sigma$.
\end{rem}

For sets $U$ and $V$, we denote the set of all functions from $U$ to $V$ by 
${\mathcal F}(U, V)$.

\begin{defi} [Wreath Product]\label{def:wp}
  Let $T_1 = (X, M)$ and $T_2 = (Y, N)$ be two tm's.  We define the \emph{wreath product}
  of $T_1$ and $T_2$ to be the tm $T = T_1 \wr T_2 = (X \times Y, M \times {\mathcal F}(X,
  N))$ where, for $m \in M$ and $f \in {\mathcal F}(X, N)$, $(m,f)$ represents the
  following transformation on $X \times Y$:
  \[
  \text{~for~} (x,y) \in \wpis, \;\;\;(m,f)((x,y)) = (m(x), f(x)(y)) \,. 
  \]
  One verifies that the product of $(m_1, f_1), (m_2, f_2) \in M \times {\mathcal F}(X,
  N)$ is $(m_1, f_1)(m_2,f_2) = (m_1m_2,f)$, where, for each $x \in X$, $f(x) = f_1(x)
  \odot f_2(m_1(x))$ (here, $\odot$ denotes the operation in $N$).
\end{defi}

It is well known~\cite{Eilenberg1976AutomataLA} that the wreath product operation is
associative on transformation monoids.
The celebrated Krohn-Rhodes theorem \cite{KR} (see
\cite{str_cirBook,DBLP:journals/fuin/DiekertKS12} for different proofs), in its division
and its simulation formulations, is as follows.

\begin{thmC}[(Krohn-Rhodes Theorem)]\label{wordkr}
  Let $\phi\colon \Sigma^* \to T=(X,M)$ be a morphism into a finite tm.  Then $\phi$ is
  simulated by a morphism $\psi\colon \Sigma^* \to T'$, where the tm $T'$ is the wreath
  product of finitely many transformation monoids which are either copies of $U_2$ or of
  the form $(G,G)$ for some non-trivial simple group $G$ dividing $M$.

  In particular, every finite transformation monoid $(X,M)$ divides a wreath product of
  the form above.
\end{thmC}

Directly interpreting these results in terms of morphisms from trace monoids to ordinary
tm's leads to major technical difficulties (see \cite{thesis, versionone}). For instance,
division of transformation monoids does not imply simulation of morphisms from trace
monoid to the tm's due to the trace monoid not being a free monoid. Also the crucial
wreath product principle, that describes word languages recognized by a wreath product of
two transformation monoids in terms of word languages recognized by the individual tm's,
breaks down when working with trace languages and morphisms from trace monoid; this is
primarily due to the fact that the principle uses a sequential transducer that takes into
account the runs of an automaton on words. But in a diamond automaton different
linearizations of a trace may produce different runs, albeit ending in the same final
state (see Remark~{\ref{rem:diamond}}).

\subsection{Asynchronous transformation monoids}\label{sec: atm}
We now introduce a new algebraic framework to discuss recognizability for
trace languages, which is more consistent with the distributed nature of the alphabet
$\alphloc$ and with the notion of asynchronous automata. This point of view resolves the
issues mentioned in the last subsection. 

Asynchronous transformation monoids are defined as follows.

\begin{defi}\label{def: atm}
  An \emph{asynchronous transformation monoid} (in short, atm) $T$ (over $\pset$)
  is a pair $(\{S_{i}\}, M)$ where
	\begin{itemize}
		\item $\{S_{i}\}$ is a $\pset$-indexed family of finite non-empty sets.
		\item $M$ is a submonoid of $\trans(\gs)$. 
	\end{itemize}
\end{defi}

Note that an atm $T=(\{S_i\}, M)$ naturally induces the tm $(\gs, M)$, and that one can
view $T$ as the tm $(\gs, M)$, equipped with an additional structure which depends on
$\pset$.  We abuse notation and write $T$ also for this tm.

More precisely, a crucial feature of the definition of an atm is that it makes a clear
distinction between local and global states.  While the underlying transformations operate
on global states, we will be interested in global transformations that are essentially
``induced'' by a particular subset $P$ of processes, that is, $P$-maps in the sense of
Section~\ref{section: recog}.

It is worth pointing out at this stage, that a transformation $g\in \trans(\gs)$ such that
$g(s)$ is of the form $(s'_P,s_{-\!P})$ for every $s\in \gs$, is \emph{not} necessarily a
$P$-map.  This condition merely says that the $(-P)$-component of a global state is not
updated by $g$.  The update of the $P$-component might still depend on the
$(-P)$-component.

The following lemma provides a characterization of $P$-maps. We skip the easy proof.

\begin{lem}\label{easylemma} 
  Let $h\colon \gs \rightarrow \gs$.  Then $h$ is a $P$-map if and only if for every $s$
  in $\gs$, ${[h(s)]}_{-\!P} = s_{-\!P}$ and for every $s, s'$ in $\gs$, $s_P=s'_P$
  implies that ${[h(s)]}_P = {[h(s')]}_P$.
\end{lem}

\begin{exa}\label{ex:u2l} 
  Fix a process $p\in \pset$.  We define the atm $U_2[p]= (\{S_i\}, M)$ by letting
  $S_{p}=\{1,2\}$ and, for each $i \neq p$, $S_i$ is a singleton set.  Observe that $\gs$ has
  only two global states which are completely determined by their $p$-components.  We
  therefore identify a global state with its $p$-component.  Then we let $M =
  \{\mathrm{id}_{\gs}, r_1, r_2\}$, where $r_i$ maps every global state to the global
  state $i$.  Note that $r_1$ and $r_2$ are $\{p\}$-maps.

  Similarly, if $T = (X,M)$ is a tm, we let $T[p]$ be the atm $T[p]= (\{S_i\}, M)$ where
  $S_p = X$ and all other $S_i$ are singletons.  The extensions of the elements of $M$
  (transformations of $X = S_p$) are $p$-maps and form a monoid in natural bijection with
  $M$, which we again write $M$.
\end{exa}

A simple but crucial observation regarding $P$-maps is recorded in the following lemma.

\begin{lem}\label{simplelemma}
  Let $f, g \in \trans(\gs)$ and let $P, P' \subseteq \pset$.  If $f$ is a $P$-map, $g$ is
  a $P'$-map and $P \cap P' = \emptyset$, then $fg=gf$.
\end{lem}

\begin{proof} 
  Suppose that $f$ and $g$ are the extensions of some $f'\in\trans(S_P)$ and
  $g'\in\trans(S_{P'})$, respectively.  Let $Q=\pset\setminus(P\cup P')$.  We can denote,
  unambiguously, a global state $s \in \gs$ as $s=(s_P,s_{P'},s_Q)$.  We then have
  \begin{align*}
    fg ~((s_P, s_{P'}, s_Q))  =  g ~((f'(s_P), s_{P'}, s_Q))  =  (f'(s_P), g'(s_{P'}), s_Q) \\
    gf ~((s_P, s_{P'}, s_Q))  =  f ~((s_P, g'(s_{P'}), s_Q))  =  (f'(s_P), g'(s_{P'}), s_Q).
  \end{align*}
  This shows that $f$ and $g$ commute.
\end{proof}

\subsection{Asynchronous morphisms}\label{sec: asychm}

We now introduce particular morphisms from the trace monoid $\traces$ to asynchronous
transformation monoids.

\begin{defi}\label{def: asynch morphisms}
	Let $T = (\{S_i\}, M)$ be an atm. An \emph{asynchronous morphism} from
	$\traces$ to $T$ is a (monoid) morphism $\phi \colon \traces
	\rightarrow M$ such that $\phi(a)$ is an $a$-map for each $a \in \Sigma$. 
\end{defi}

It is important to observe that not every morphism $\phi \colon \traces \to M$ defines an
asynchronous morphism: indeed $\phi(a)$ and $\phi(b)$ may commute (say, if $(a,b) \in I$)
even if $\phi(a)$ (resp.  $\phi(b)$) is not an $a$-map (resp.  a $b$-map).

An elementary yet fundamental result about asynchronous morphisms is stated in
Lemma~\ref{lem:extends-asyn-morph} below.  

\begin{lem}\label{lem:extends-asyn-morph}
	Let $T = ( \{S_i\}, M)$ be an atm. Further, let $\phi \colon \Sigma \to M$ 
	be such that $\phi(a)$ is an $a$-map for every $a \in \Sigma$. Then $\phi$ can 
	be uniquely extended to an asynchronous morphism from
	$\traces$ to $T$.
\end{lem}

\begin{proof} 
  The map $\phi$ uniquely extends to a morphism from the free monoid $\Sigma^*$ to $M$.
  By Proposition~\ref{basicprop}, $\traces$ is the quotient of $\Sigma^*$ by the relations
  of the form $ab=ba$ where $(a,b) \in I$, so we only need to show that $\phi(a)$ and
  $\phi(b)$ commute.  Indeed, $(a,b) \in I$ means that $\loc(a) \cap \loc(b) = \emptyset$.
  As $\phi(a)$ and $\phi(b)$ are an $a$-map and a $b$-map, respectively, the result follows
  from Lemma~\ref{simplelemma}.
\end{proof}

\begin{exa}\label{ex:amorph}
	Let $\alphloc$ over $\pset=\{p_1,p_2,p_3\}$ be given by $\alphabet = \{a,b,c\}$ and $\loc(a) = \{p_1\}$, $\loc(b)
= \{p_1,p_2\}$, $\loc(c) = \{p_2,p_3\}$.  Letting $\phi(a) = r_1$, $\phi(b) = r_2$
and $\phi(c) = \mathrm{id}$, determines an asynchronous morphism from $\traces$ to
$U_2[p_1]$.

  If instead we let $\phi(c) = r_1$, $\phi$ determines a morphism from $\traces$ to 
  $U_2[p_1]$, which is not asynchronous.
\end{exa}

A very important example of an asynchronous morphism is given by the transition morphism
of an asynchronous automaton.  Let $\aaa = (\{\ls_i \}, \{\lt_a\}, \sinit)$ be an
asynchronous automaton over $\alphloc$.
For each $a\in \Sigma$, let $\gt_a \in \trans(\gs)$ be the extension of the local
transition function $\lt_a$, an $a$-map by definition.  Let also $M_\aaa$ be the submonoid
of $\trans(\gs)$ generated by the $\gt_a$ ($a\in \Sigma$).  By
Lemma~\ref{lem:extends-asyn-morph} the map $a \mapsto \gt_a$ extends to an asynchronous
morphism $\phi_\aaa$, from $\traces$ to the atm $T_\aaa = (\{\ls_i\}, M_\aaa)$.  We say
that $\phi_\aaa$ is the \emph{transition} (asynchronous) \emph{morphism} of $\aaa$ and
$T_\aaa = (\{\ls_i\}, M_\aaa)$ is the \emph{transition atm} of $\aaa$.

We record the following lemma, whose proof is immediate.

\begin{lem}\label{lem:atmza}
  Given an asynchronous automaton $\aaa = (\{\ls_i\},\{\lt_a\},\sinit)$ over $\alphloc$,
  its transition atm $T_\aaa$ and its transition asynchronous morphism $\phi_\aaa \colon
  \traces \to T_\aaa$ are effectively constructible.
  
  If $t \in \traces$, then $\phi_\aaa(t)(\sinit) = \aaa(t)$.
  
  A trace language is accepted by $\aaa$ if and only if it is recognized by $T_\aaa$ via
  $\phi_\aaa$, with $\sinit$ as the initial state.
\end{lem}

We also record the converse construction.  Let $T=\atma$ be an atm, let $\sinit \in \gs$
be a global state and let $\phi \colon \traces \to T$ be an asynchronous morphism.  For
each $a\in \Sigma$, $\phi(a)$ is an $a$-map: let $\lt_a$ be the (uniquely determined)
transformation of $S_a$ of which $\phi(a)$ is an extension.  Finally let $\aaa_\phi =
(\{\ls_i\}, \{\lt_a\}, \sinit)$.  Then $\aaa_\phi$ is an asynchronous automaton (over
$\alphloc$) and the following lemma is easily verified.

\begin{lem}\label{lem:zaatm}
  Given an atm $T=\atma$, a global state $\sinit \in \gs$ and an asynchronous morphism
  $\phi \colon \traces \to T$, the asynchronous automaton $\aaa_\phi$ is effectively
  constructible. Moreover, $\phi$ is the asynchronous transition morphism of $\aaa_\phi$.

A trace language $L \subseteq \traces$ is recognized by $T$ via $\phi$ (with
  initial state $\sinit$) if and only if it is accepted by $\aaa_\phi$.
\end{lem} 

Thus Zielonka's theorem can be rephrased to state that a trace language is recognizable if
and only if it is recognized by an asynchronous morphism to an atm.  We will see a more
precise rephrasing in Section~\ref{sec: decomposition} below (Theorem~\ref{zielonka}).

\subsection{Asynchronous wreath product}\label{section: asynchronous wreath}

We adapt the definition of wreath product (see Section~\ref{subsec:sequential}) 
to the setting of asynchronous transformation monoids.

\begin{defi}\label{def: wrp atm}
  Let $T_1 = (\{S_i\},M)$ and $T_2 = (\{Q_i\},N)$ be two asynchronous transformation
  monoids.  Their \emph{asynchronous wreath product}, written $T_1 \wras T_2$, is defined
  to be the atm $(\{S_i \times Q_i\}, M \times {\mathcal F}(\gs, N))$.  An element $(m,f)
  \in M \times {\mathcal F}(\gs, N)$ represents the following global\footnote{a global
  state (resp.\ $P$-state) of $T_1\wr T_2$ is canonically identified with an element of
  $S\times Q$ (resp.\ $S_P \times Q_P$)} transformation on $S \times Q$: 
  \[
  \text{~for~} (s,q) \in S \times Q , \;\;\;(m,f)((s,q)) = (m(s), f(s)(q)) \,.
  \]
\end{defi}

An important observation is that the tm associated with the atm $T_1 \wras T_2$ is the
wreath product of the transformation monoids $(\gs,M)$ and $(Q,N)$ associated with
$T_1$ and $T_2$ respectively.  In particular, the composition law on $M \times {\mathcal
F}(\gs, N)$ is the same as in Definition~\ref{def:wp}.  The associativity of the
asynchronous wreath product operation follows immediately.

\begin{lem}\label{lem:combinatorial}
  Let $T_1 = (\{S_i\},M)$ and $T_2 = (\{Q_i\},N)$ be asynchronous transformation monoids
  and let $P \subseteq \pset$.  If $(m,f) \in M \times \trans(\gs, N)$ is a $P$-map in
  $T_1 \wras T_2$, then
  \begin{itemize}
    \item $m$ is a $P$-map in $T_1$.

    \item For every $s \in \gs, ~f(s)$ is a $P$-map in $T_2$. Further,
    if $s, s' \in \gs$ are such that $s_P = s'_P$, then $f(s) = f(s')$.
  \end{itemize}
\end{lem}

\begin{proof}
  Recall that by Lemma~\ref{easylemma}, for a $\pset$-indexed family $\{X_i\}$, a
  transformation $h\in \trans(X)$ is a $P$-map if and only if (a) $(h(x))_{-\!P} = x_{-\!P}$
  for every $x\in X$ and (b) $x_P = x'_P$ implies $(h(x))_P = (h(x'))_P$ for all
  $x,x'\in X$.

  Let $s \in \gs$ and $q\in Q$.  Since $(m,f)$ is a $P$-map, we have
  $[(m,f)((s,q))]_{-\!P} =(s_{-\!P}, q_{-\!P})$, that is, $(m(s))_{-\!P} = s_{-\!P}$ and
  $(f(s)(q))_{-\!P} = q_{-\!P}$.

  In addition, if $s' \in \gs$ and $q'\in Q$ are such that $s'_P = s_P$ and $q'_P
  = q_P$, we have $[(m,f)((s,q))]_P = [(m,f)((s',q'))]_P$, that is, $m(s)_P = m(s')_P$ and
  $(f(s)(q))_P = (f(s')(q'))_P$.  This already establishes that $m$ is a $P$-map in $T_1$.

  Moreover, if we choose $s = s'$, we conclude that $f(s)$ is a $P$-map in $T_2$.  If
  instead we choose $q = q'$, we get
  \[
  f(s)(q) = ((f(s)(q))_P,q_{-\!P}) = ((f(s')(q))_P,q_{-\!P}) = f(s')(q),
  \]
  that is, $f(s) = f(s')$, which concludes the proof.
\end{proof}

\subsection{Asynchronous wreath product principle}\label{sec:wpp}

The classical wreath product principle \cite{str_cirBook} is a critical result that defines
the importance and utility of wreath product structures in formal language theory.
We give here analogous
results which exploit the distributed structure of asynchronous automata and asynchronous
transformation monoids.  In fact we recover the classical principle as a special case
when there is only one process.

Let $T = (\{\ls_i\}, M)$ be an atm and $\phi\colon \traces \to T$ be an asynchronous
morphism.  We associate with $T$ the alphabet $\tralphabet = \{ (a, s_a) \mid a \in
\alphabet, s \in \gs \}$ where each letter $a$ is decorated with local $a$-state
information of $T$.  Recall that this alphabet can be viewed as a distributed alphabet
by letting $(\tralphabet)_i$ ($i\in \pset$) be the set of
letters $(a, s_a) \in \tralphabet$ such that $a \in \alphabet_i$.  In other words,
$\loc((a,s_a)) = \loc(a)$.  The choice of an initial global state $\sinit \in \gs$ induces
the following transducer over traces.

\begin{defi}
  The \emph{asynchronous transducer} associated with $\phi$ and $\sinit$ is the map
  $\chi^{\sinit}_{\phi} \colon \traces \to \tracesdtr$ defined as follows.
  If $t = (E,\leq, \lambda) \in \traces$, then $\chi^{\sinit}_{\phi}(t)$ is the trace $(E,
  \leq, (\lambda,\mu)) \in \tracesdtr$ where the labelling function $\mu$ is
  defined as follows.  For each event $e$ of $t$, if $\lambda(e) = a$ and $s = \phi(\Da
  e)(\sinit)$, then $\mu(e) = s_a$.
\end{defi}

It is immediately verified that, if $\aaa=(\{\ls_i\}, \{\lt_a\}, \sinit)$ is an
asynchronous automaton and $\phi$ is its transition morphism, then $\chi^{\sinit}_{\phi}$
coincides with $\chi_\aaa$, the asynchronous transducer of $\aaa$ defined in
Section~\ref{sec: asynchronous automata}.

\begin{exa}\label{ex:lat}
  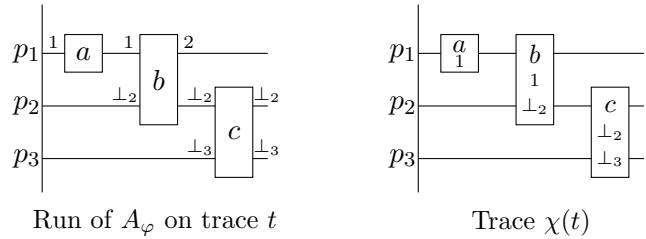
\begin{figure}[ht]
    \centering
    \begin{tikzpicture}
\draw (0,0) -- (0, 5*0.2 + 3*0.5);

\draw (0.3,1.6) rectangle (0.8,2.1);
\draw (1.3,0.9) rectangle (1.8,2.1);
\draw (2.3,0.2) rectangle (2.8,1.4);

\draw (0,0.2 + 0.5/2) -- (2.3, 0.2 +0.5/2); 
\draw (2.8,0.2 + 0.5/2) -- (3.0,0.2 + 0.5/2); 

\draw (0,2*0.2+0.5+0.5/2) -- (1.3,2*0.2+0.5+0.5/2); 
\draw (1.8,2*0.2+0.5+0.5/2) -- (2.3,2*0.2+0.5+0.5/2); 
\draw (2.8,2*0.2+0.5+0.5/2) -- (3.0,2*0.2+0.5+0.5/2); 

\draw (0, 3*0.2 + 2*0.5 + 0.5/2) -- (0.3, 3*0.2 + 2*0.5 + 0.5/2);
\draw (0.8, 3*0.2 + 2*0.5 + 0.5/2) -- (1.3, 3*0.2 + 2*0.5 + 0.5/2);
\draw (1.8, 3*0.2 + 2*0.5 + 0.5/2) -- (3.0, 3*0.2 + 2*0.5 + 0.5/2);

\node at (0.3 + 0.5/2, 1.6 + 0.5/2) {$a$};
\node at (1.3 + 0.5/2, 0.9 + 0.5 + 0.2/2) {$b$};
\node at (2.3 +  0.5/2, 0.2 + 0.5 + 0.2/2) {$c$};

\node at (-0.2, 0.2 + 0.5/2) {$p_3$};
\node at (-0.2, 2*0.2 + 0.5 + 0.5/2) {$p_2$};
\node at (-0.2, 3*0.2 + 2*0.5 + 0.5/2) {$p_1$};

\node at (0.15, 3*0.2 + 2*0.5 +0.5/2 + 0.15) {\tiny $1$};
\node at (1.15, 3*0.2 + 2*0.5 +0.5/2 + 0.15) {\tiny $1$};
\node at (1.95, 3*0.2 + 2*0.5 +0.5/2 + 0.15) {\tiny $2$};

\node at (1.1, 2*0.2 + 0.5 +0.5/2 + 0.15) {\tiny $\bot_2$};
\node at (2.1, 2*0.2 + 0.5 +0.5/2 + 0.15) {\tiny $\bot_2$};
\node at (2.99, 2*0.2 + 0.5 +0.5/2 + 0.15) {\tiny $\bot_2$};

\node at (2.1, 0.2 + 0.5/2 + 0.15) {\tiny $\bot_3$};
\node at (2.99, 0.2 + 0.5/2 + 0.15) {\tiny $\bot_3$};

\node at (1.5, -0.4) {\small Run of $A_\varphi$ on trace $t$};

\begin{scope}[shift={(5,0)}]
\draw (0,0) -- (0, 5*0.2 + 3*0.5);

\draw (0.3,1.6) rectangle (0.8,2.1);
\draw (1.3,0.9) rectangle (1.8,2.1);
\draw (2.3,0.2) rectangle (2.8,1.4);

\draw (0,0.2 + 0.5/2) -- (2.3, 0.2 +0.5/2); 
\draw (2.8,0.2 + 0.5/2) -- (3.0,0.2 + 0.5/2); 

\draw (0,2*0.2+0.5+0.5/2) -- (1.3,2*0.2+0.5+0.5/2); 
\draw (1.8,2*0.2+0.5+0.5/2) -- (2.3,2*0.2+0.5+0.5/2); 
\draw (2.8,2*0.2+0.5+0.5/2) -- (3.0,2*0.2+0.5+0.5/2); 

\draw (0, 3*0.2 + 2*0.5 + 0.5/2) -- (0.3, 3*0.2 + 2*0.5 + 0.5/2);
\draw (0.8, 3*0.2 + 2*0.5 + 0.5/2) -- (1.3, 3*0.2 + 2*0.5 + 0.5/2);
\draw (1.8, 3*0.2 + 2*0.5 + 0.5/2) -- (3.0, 3*0.2 + 2*0.5 + 0.5/2);

\node at (0.3 + 0.5/2, 1.6 + 0.5/2 + 0.1) {\small $a$}; 
\node at (0.3 + 0.5/2, 1.6 + 0.5/2 - 0.1) {\tiny $1$}; 
\node at (1.3 + 0.5/2, 1.6 + 0.5/2) {\small $b$};
\node at (1.3 + 0.5/2, 0.9 + 0.5 + 0.2/2) {\tiny $1$};
\node at (1.3 + 0.5/2, 0.9 + 0.5/2) {\tiny $\bot_2$};
\node at (2.3 +  0.5/2, 0.2 + 0.5 + 0.2 +0.5/2) {\small $c$};
\node at (2.3 +  0.5/2, 0.2 + 0.5 + 0.2/2) {\tiny $\bot_2$};
\node at (2.3 +  0.5/2, 0.2 + 0.5/2) {\tiny $\bot_3$};

\node at (-0.2, 0.2 + 0.5/2) {$p_3$};
\node at (-0.2, 2*0.2 + 0.5 + 0.5/2) {$p_2$};
\node at (-0.2, 3*0.2 + 2*0.5 + 0.5/2) {$p_1$};

\node at (1.5, -0.4) {\small Trace $\chi(t)$};
\end{scope}
\end{tikzpicture}
    \caption{Asynchronous transducer output on a trace.}%
    \label{fig:lat-example}
  \end{figure}
  Let $\phi$ be the first asynchronous morphism in Example~\ref{ex:amorph}, with $S_{p_1}
  = \{1,2\}$, $S_{p_2} = \{\bot_2\}$ and $S_{p_3}= \{\bot_3\}$.  Let
  $\sinit=(1,\bot_2,\bot_3)$ be the global initial state and let $\chi$ be the
  corresponding asynchronous transducer.  Figure~\ref{fig:lat-example} shows
  (automata-style) the computation of the asynchronous morphism $\phi$ on the trace $t =
  abc \in \traces$ (by showing local process states before and after each event), and the
  resulting trace $\chi(t) \in \tracesdtr$.
\end{exa}

The following lemma is a straightforward consequence of the definition of the 
asynchronous transducer.
 
\begin{lem}\label{lem:trfactor}
  Let $\phi\colon\traces \to T$ be an asynchronous morphism to an atm $T = (\{\ls_i\},
  M)$, let $\sinit\in \gs$ and let $\chi$ be the associated asynchronous transducer.  Let
  $t \in \traces$, $a \in \alphabet$ and $s = \phi(t)(\sinit)$.  Then the trace $\chi(ta)
  \in \tracesdtr$ factors as $\chi(ta) = \chi(t)(a,s_a)$.
\end{lem}

We now define a notion of asynchronous wreath product of asynchronous morphisms 
defined on trace monoids.

\begin{defi}\label{def:wp-asyn-morph}
  Let $\phi \colon \traces \to T$ and $\psi \colon \tracesdtr \to T'$ be
  asynchronous morphisms to asynchronous transformation monoids $T = (\{S_i\},M)$ and $T'
  = (\{Q_i\},N)$.  For each $a\in \Sigma$ and $s \in \gs$, let $g_a(s) = \psi(a,s_a)$.
  The map $\phi \wras \psi$ is defined on $\Sigma$ by letting $(\phi\wras\psi)(a)
  = (\phi(a),g_a)$.
\end{defi}

\begin{lem}
  If $\phi$ and $\psi$ are as in the above definition, then
  $\phi\wras\psi$ induces an asynchronous morphism to the atm $T\wras T'$.
\end{lem}

\begin{proof}
  For each $a\in \Sigma$, $(\phi\wras\psi)(a)$ is
  an $a$-map since $\loc(a) = \loc((a,s_a))$ and $\phi$ and $\psi$ are asynchronous.  As a result
  (Lemma~\ref{lem:extends-asyn-morph}), $\phi\wras\psi$ extends to an asynchronous
  morphism $\phi\wras\psi\colon \traces \to T \wras T'$.
\end{proof}

The following technical result establishes an important connection between
asynchronous transducers and
asynchronous wreath product of asynchronous morphisms. This is crucially utilized 
in the proof of the asynchronous wreath product principle. 

\begin{prop}\label{prop: asynchronous wpp}
  Let $\phi\colon\traces \to T$ be an asynchronous morphism to an atm $T = (\{\ls_i\},
  M)$, let $\sinit\in \gs$ and let $\chi$ be the associated asynchronous transducer.  Let
  $\psi\colon \tracesdtr \to T'$ be an asynchronous morphism to an atm $T' = (\{Q_i\},N)$.

  For each $t\in \traces$, let $\pi_1(t)$ and $\pi_2(t)$ be the first and second component
  projections of $(\phi\wras\psi)(t) \in M \times \trans(\gs,N)$.  Then $\pi_1(t) =
  \phi(t)$ and $\pi_2(t)(\sinit) = \psi(\chi(t))$.
\end{prop}

\begin{proof}
  Let $t = (E, \leq, \lambda) \in \traces$.  The fact that $\pi_1(t) = \phi(t)$ follows
  directly from the definition of wreath products.  The second equality is verified by
  induction on the cardinality of $E$.  It is trivial if $|E| = 0$ and we now suppose $t =
  t'a$, so that
  \[
  (\phi\wras\psi)(t) = (\phi\wras\psi)(t')\cdot(\phi\wras\psi)(a) \,.
  \]
  As in Definition~\ref{def:wp}, to visually distinguish the operations in the different
  monoids, we write $\odot$ for the operation of $N$.  We have $(\pi_1(t), \pi_2(t)) =
  (\pi_1(t'), \pi_2(t'))(\pi_1(a), \pi_2(a))$ and hence, for $x \in \gs$, $\pi_2(t)(x) =
  \pi_2(t')(x) \odot \pi_2(a)(\pi_1(t')(x))$.  This holds in particular for $x = \sinit$.
  Let also $s = \pi_1(t')(\sinit)=\phi(t')(\sinit)$.  Then we have
	\begin{align*}
    \pi_2(t)(\sinit) &= \pi_2(t')(\sinit) \odot \pi_2(a)(\pi_1(t')(\sinit)) \\
    &= \psi(\chi(t')) \odot \pi_2(a)(s) \quad\text{by induction} \\
    &= \psi(\chi(t')) \odot \psi((a,s_a)) \\
    &= \psi(\chi(t')(a,s_a)) \\
    &= \psi(\chi(t)) \quad\text{by Lemma~\ref{lem:trfactor},}
	\end{align*}
  and this concludes the proof.
\end{proof}

We can now state and prove both directions of what we term the \emph{asynchronous wreath
	product principle}. 

\begin{thm}\label{thm:wpp1}
  Let $\phi\colon \traces \to T$ be an asynchronous morphism to an atm $T = (\{\ls_i\},
  M)$ and let $\sinit \in \gs$.  Let $\chi\colon \traces \to \tracesdtr$ be the associated
  asynchronous transducer.  If $L \subseteq \tracesdtr$ is recognized by an atm $T'$, then
  $\chi^{-1}(L)$ is recognized by the atm $T \wras T'$.
\end{thm}

\begin{proof}
  Let $\psi \colon \tracesdtr \to T' = (\{Q_i\}, N)$ be an asynchronous morphism
  recognizing $L$, with $\qinit \in Q$ as the initial global state, and
  $Q_{\textrm{fin}} \subseteq Q$ as the set of final global states.  Then a trace $t
  \in \tracesdtr$ is in $L$ if and only if $\psi(t)(\qinit) \in Q_{\textrm{fin}}$.

  By Proposition~\ref{prop: asynchronous wpp}, for $t\in \traces$, we have
  $(\phi\wras\psi)(t)(\sinit, \qinit) = (\phi(t)(\sinit), \psi(\chi(t))(\qinit))$.
  Therefore $t \in \chi^{-1}(L)$ if and only if $(\phi\wras\psi)(t)(\sinit, \qinit) \in
  \gs \times Q_{\textrm{fin}}$, which concludes the proof that $\phi\wras\psi$ recognizes
  $\chi^{-1}(L)$.
\end{proof}
 
\begin{thm}\label{thm:wpp2}
  Let $T_1 = (\{ \ls_i \},M)$ and $T_2 = ( \{Q_i\},N)$ be atms and let $L \subseteq
  \traces$ be a trace language recognized by an asynchronous morphism $\eta\colon \traces
  \to T_1 \wras T_2$, with initial global state $(\sinit, \qinit)$.  For each $a\in
  \alphabet$, let $\eta(a) = (m_a,f_a)$.  Mapping each $a\in \Sigma$ to $m_a$ defines an
  asynchronous morphism $\phi\colon \traces \to T_1$.  Let $\chi$ be the local
  asynchronous transducer associated to $\phi$ and $s_\textrm{in}$.  Then $L$ is a finite
  union of languages of the form $U \cap \chi^{-1}(V)$, where $U \subseteq \traces$ is
  recognized by $T_1$ and $V \subseteq \tracesdtr$ is recognized by $T_2$.
\end{thm}

\begin{proof}
  For $a \in \Sigma$, $\eta(a) = (m_a, f_a) \in M \times \trans(\gs, N)$ is an $a$-map
  and, by Lemma~\ref{lem:combinatorial}, $m_a \in M$ is an $a$-map (of $T_1$) and
  $f_a\colon S \to N$ is such that, for each $s \in S$, $f_a(s) \in N$ is an
  $a$-map (of $T_2$) which depends only on $s_a$.  In particular, $f_a\colon S \to
  N$ may be viewed as a map $f_a\colon S_a \to N$.  Below we will use $f_a$ in this sense.

  It follows, by Lemma~\ref{lem:extends-asyn-morph}, that the map $a\mapsto m_a$ extends
  to an asynchronous morphism $\phi\colon \traces \to T_1$.  Similarly, mapping
  $(a,s_a)\in \s\loctimes\gs$ to $f_a(s_a)$ defines an asynchronous morphism
  $\psi\colon \tracesdtr \to T_2$ and we have $\eta = \phi\wras\psi$.

  By definition of recognizability, $L$ is the union of a finite family of languages
  recognized by $\eta$ with initial global state $(\sinit, \qinit)$ and a single final
  global state and we can, without loss of generality, assume that $L$ is recognized with
  a single final global state, say, $(s_{\textrm{fin}}, q_{\textrm{fin}})$.

  Let $\pi_1(t)$ and $\pi_2(t)$ be the first and second component projections of
  $\eta(t)$, for each $t\in \traces$.  Thus $t \in L$ if and only if $\eta(t)((\sinit,
  \qinit)) = (s_{\textrm{fin}}, q_{\textrm{fin}})$, that is,
  \[
  (\pi_1(t)(\sinit),\pi_2(t)(\sinit)(\qinit)) = (s_{\textrm{fin}}, q_{\textrm{fin}}) \,.
  \]
  Proposition~\ref{prop: asynchronous wpp}, applied to $\eta = \phi\wras\psi$, shows that
  this is equivalent to $\phi(t)(\sinit) = s_{\textrm{fin}}$ and $\psi(\chi(t))(\qinit) =
  q_{\textrm{fin}}$.

  Let now $U \subseteq \traces$ be recognized by the asynchronous morphism $\phi$ with
  initial and final states $\sinit$ and $s_{\textrm{fin}}$, and let $V \subseteq
  \tracesdtr$ be recognized by $\psi$ with initial and final states $\qinit$ and
  $q_{\textrm{fin}}$.  Then $t\in L$ if and only if $t\in U$ and $\chi(t)\in V$, that is,
  $L = U \cap \chi^{-1}(V)$, which completes the proof.
\end{proof}

\begin{rem}
  Note that the asynchronous wreath product principle, when restricted to a single
  process, corresponds exactly to the sequential wreath product principle.  
\end{rem}

\begin{exa}\label{ex: asyn wreath}
	Consider the distributed alphabet $\olddalphabet= (\{a,b\},\{b,c\},\{c\})$ over
	$\pset = \{p_1, p_2, p_3\}$ from Example~\ref{ex:amorph}. We define an
	asynchronous morphism $\eta$  from $\traces$ into the asynchronous wreath product
	$\atma \wras \atmb$ where $\atma = U_2[p_1]$ and $\atmb = U_2[p_3]$. Denoting the
	(isomorphic) monoids of $U_2[p_1]$ and $U_2[p_3]$ by $U_2$, by
	Definition~\ref{def: wrp atm}, we know that $\eta(a) = (m_a,f_a) \in U_2 \times
	\trans(\gs, U_2)$. Further, by Lemma~\ref{lem:combinatorial}, $f_a$ can be
	described as a function in $\trans(S_a, U_2)$. Let the local states of the first
	tm be $S_{p_1} = \{1,2\}, S_{p_2} = \{\bot_2\}, S_{p_3} = \{\bot_3\}$, and those
	of the second tm be $Q_{p_1} = \{\bot'_1\}, Q_{p_2} = \{\bot'_2\}, Q_{p_3} =
	\{1',2'\}$. It is clear from the description of $\eta$ below that $\eta(a)$ (resp.
	$\eta(b)$ and $\eta(c)$) is an $a$-map (resp. $b$-map and $c$-map). Therefore
	$\eta$ extends to an asynchronous morphism.
	\begin{align*}
		\eta(a) &=  (r_2, \{1 \mapsto \mathrm{id}, 2 \mapsto \mathrm{id}\})
			\\
		\eta(b) &= (\mathrm{id}, \{ 1\bot_2 \mapsto \mathrm{id}, 2\bot_2 \mapsto
		\mathrm{id}\}) \\
			\eta(c) &= (\mathrm{id}, \{ \bot_2 \bot_3 \mapsto r_{2'}\}) 
	\end{align*}
	The naturally derived asynchronous morphisms $\varphi \colon \traces \to U_2[p_1]$
	and $\psi \colon \tracesdtr \to U_2[p_3]$ to the individual atm's, as in the proof
	of Theorem~\ref{thm:wpp2}, are $\varphi = \{ a \mapsto r_2, b \mapsto
	\mathrm{id}\}$, and $\psi = \{c \bot_2 \bot_3 \mapsto r_{2'} \}$; since $U_2[p_1]$
	is localized, we only need to describe $\varphi$ on $\Sigma_{p_1}$ letters (the
	other letters being mapped to $\mathrm{id}$), similarly for $\psi$. Note that
	since there is no letter on which processes $p_1$ and $p_3$ synchronize, and the
	information passed by $U_2[p_1]$ via re-labelling of events is `local', it cannot
	be utilised by $U_2[p_3]$, and hence the asynchronous morphism $\eta$ to the
	asynchronous wreath product structure $U_2[p_1] \wras U_2[p_3]$ essentially
	reduces to an asynchronous morphism into the direct product $U_2[p_1] \times
	U_2[p_3]$.
\end{exa}

\subsection{Local cascade product}\label{sec: local cascade}

Now we present an automata-theoretic view of the asynchronous wreath product.

\begin{defi}\label{def:localcascade}
  Let $\aaa_1 = (\{\ls_i\}, \{\lt_a\}, \sinit)$ and $\aaa_2 = (\{Q_i\}, \{\lt_{(a,s_a)}\},
  \qinit)$ be asynchronous automata over the distributed alphabets $\alphloc$ and
  $\alphloc[\s\loctimes\gs]$, respectively.  The \emph{local cascade product} of $A_1$ and $A_2$,
  written $A_1 \lc A_2$, is the asynchronous automaton $ (\{R_i\}, \{\Delta_a\}, (\sinit,
  \qinit))$ over $\alphloc$ where, for $i\in \pset$, $R_i = \ls_i \times Q_i$ and where,
  for $a \in \Sigma$ and $(s_a, q_a) \in R_a$\footnote{We identify
  $R_a=\prod_{i\in\loc(a)}S_i\times Q_i$ with $S_a\times
  Q_a=\left(\left(\prod_{i\in\loc(a)}S_i\right)\times\left(\prod_{i\in\loc(a)}Q_i\right)\right)$.},
  $\Delta_a((s_a, q_a))= (\lt_a(s_a), \lt_{(a,s_a)}(q_a))$.
\end{defi}

This operation on asynchronous automata corresponds exactly to the asynchronous wreath
product defined in Section~\ref{section: asynchronous wreath}, in the sense of the
following statement.

\begin{prop}\label{prop: local cascade}
  Let $\phi \colon \traces \to (\{S_i\}, M)$ and $\psi \colon \tracesdtr \to (\{Q_i\}, N)$
  be the asynchronous transition morphisms of $A_1$ and $A_2$, respectively.  Then the
  asynchronous transition morphism of $A_1 \lc A_2$ is $\phi \wras \psi\colon \traces \to
  (\{S_i\}, M) \wras (\{Q_i\}, N)$.
\end{prop}

\begin{proof}
  Let $\bar\delta_a$, $\bar\delta_{(a,s_a)}$ and $\bar\Delta_a$ denote the extensions to
  global states ($\gs$, $Q$ and $R$, respectively) of the maps $\delta_a \in
  \trans(S_a)$, $\delta_{(a,s_a)}\in \trans(Q_a)$ and $\Delta_a \in \trans(R_a)$.  By
  definition of transition morphisms (Section~\ref{sec: asychm}), for each $a\in \Sigma$
  and $s\in\gs$, we have $\phi(a) = \bar\delta_a$ and $\psi(a,s_a) =
  \bar\delta_{(a,s_a)}$, while the transition morphism of $A_1 \lc A_2$ maps $a$ to
  $\bar\Delta_a$.

  Let $a\in \Sigma$, let $P = \loc(a) = \loc((a,s_a))$, and let $r \in R$, say, $r =
  (s,q)$.  Then $\bar\Delta_a(r) = (\Delta_a(r_P), r_{-\!P})$.  Now $\Delta_a(r_P) =
  (\delta_a(s_P),\delta_{(a,s_a)}(q_P))$, while $r_{-\!P} = (s_{-\!P}, q_{-\!P})$, so
  $\bar\Delta_a(r) = (\bar\delta_a(s), \bar\delta_{(a,s_a)}(q))$.

  Now compare with Definition~\ref{def:wp-asyn-morph}: the map $a\mapsto \bar\Delta_a$
  coincides with the restriction of $\phi\wras\psi$ to $\Sigma$, which concludes the
  proof.
\end{proof}

The correspondence between the asynchronous wreath product of asynchronous transformation
monoids and the local cascade product of asynchronous automata established in
Proposition~\ref{prop: local cascade}, induces an automata-theoretic version of the
asynchronous wreath product principle, using the asynchronous transducers of the
asynchronous automata involved.

More precisely, a run of the local cascade product $A_1 \lc A_2$ on a trace $t$ can be
understood as follows.  One first views $A_1$ as an asychronous computing device, namely
the asynchronous transducer $\hat A_1$ (see Section~\ref{sec: asynchronous automata})
which computes $\chi_{A_1}$.  Its run on $t$ outputs the trace $\chi_{A_1}(t) \in
\tracesdtr$ (see Figure~\ref{fig:lat-example}) and one then runs $A_2$ on that trace,
leading to state $q\in Q$.  Note that, due to the asynchronous nature of the
computing device, as soon as $A_1$ has finished working on some event (of $t$), $A_2$ can
start working on the `same' event (of $\chi_{A_1}(t)$).  Further, $A_1$ and $A_2$ work
asynchronously and can `simultaneously' process concurrent events.
\begin{figure}[ht]
  \centering
  \begin{tikzpicture}
\draw (0,0) rectangle (1.2,1.2);
\node at (0.6,0.6) {$A$};

\draw (1.2, 0.3) -- (1.5, 0.3);
\draw[->] (1.5, 0.3) -- (1.5,-0.2);
\node at (1.5,-0.35) {\small $(s,q)$};
\draw[-stealth'] (-0.5, 0.6) -- (0,0.6) node [midway, above] {$t$};

\draw (2.8,0) rectangle (4,1.2);
\node at (3.4,0.6) {$A_1$};

\draw (4, 0.3) -- (4.3, 0.3);
\draw[->] (4.3, 0.3) -- (4.3,-0.2);
\node at (4.3,-0.35) {\small $s$};
\draw[-stealth'] (2.3, 0.6) -- (2.8,0.6) node [midway, above] {$t$};

\draw (5.4,0) rectangle (6.6,1.2);
\node at (6,0.6) {$A_2$};

\draw (6.6, 0.3) -- (6.9, 0.3);
\draw[->] (6.9, 0.3) -- (6.9,-0.2);
\node at (6.9,-0.35) {\small $q$};
\draw[-stealth'] (4, 0.6) -- (5.4,0.6) node [midway, above] {\small $\chi_{A_1}(t)$};

\end{tikzpicture}
  \caption{Operational view of local cascade product}
  \label{fig:local-cascade}
\end{figure}
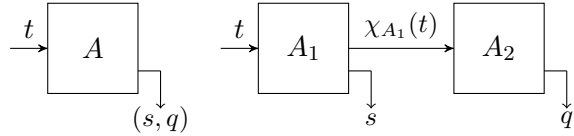
Finally, the run of $A_1 \lc A_2$ on $t$ takes the initial state $(\sinit,\qinit)$ to the
pair $(s,q)$, as in Figure~\ref{fig:local-cascade}.  A different take on this operational
point of view will be developed in Section~\ref{sec:cascade}.

Finally, note that the local cascade product is associative.  The following \emph{local
cascade product principle}, which relies on associativity, is the announced rephrasing of
the asynchronous wreath product principle in automata-theoretic terms.

\begin{thm}\label{thm:lcpp}
  Let $\aaa = \aaa_1 \lc \ldots \lc \aaa_n$ and $B = B_1 \lc \ldots \lc B_m$ be local
  cascade products such that $C = \aaa \lc B$ is defined in the sense that the distributed
  alphabets of $\aaa$ and $B$ (say, $\dalphabet$ and $\Pialphabet$ respectively) match as
  in Definition~\ref{def:localcascade}.  Further, let $\chi_{\aaa}: \traces \to
  \traces[\Pi]$
  be the asynchronous transducer of $\aaa$.  Then any language $L \subseteq \traces$
  accepted by $C$ is a finite union of languages of the form $U \cap \chi_\aaa^{-1}(V)$
  where $U \subseteq \traces$ is accepted by $\aaa$, and $V \subseteq \traces[\Pi]$ is
  accepted by $B$.
\end{thm}
\begin{exa}\label{ex: local cascade}
	Consider a distributed alphabet $\olddalphabet = ( \{a,b\}, \{b,c\}, \{c,d\})$ over
	processes $\pset = \{1,2,3\}$. We define a local cascade product $A \lc B$ where
	$A = A_1 \lc A_2 \lc A_3$ itself is also a local cascade product. $A_1, A_2, A_3$
	and $B$ are all localized asynchronous automata, having two local states in
	processes $1,2,3$, and $3$ respectively, and a single local state in each of the
	remaining processes. The automata are described in
	Figure~\ref{fig:lcascade-example}; note that due to its localized structure, $A_1$
	is completely described by transitions induced by the letters of $\Sigma_1$ on the
	local states of process $1$.  A similar statement holds for $A_2$, $A_3$, and $B$;
	also for simplicity, in the extended alphabet letters for $A_2$, $A_3$ and $B$, we
	have only displayed non-trivial state information.
	\begin{figure}
\begin{tikzpicture}
[->,>=stealth',shorten >=1pt,auto,node distance=2cm,
                    semithick, initial text=,
		    every edge/.append style={font=\scriptsize, align=center, auto}]
\tikzstyle{every state}=[draw=black]
  \node[initial,state] (p1)                    {$p_1$};
  \node[state, below of=p1]         (p2) {$p_2$};
\begin{scope}[node distance=0.8cm]
	\node [below of=p2] {$A_1$};
\end{scope}

\path (p1) edge                         node {$a$} (p2)
	  edge[loop above]              node {$b$} (p1)
	  (p2) edge[loop left]              node {$a,b$} (p2);

  \node[initial,state] (q1) [right=2cm of p1]       {$q_1$};
  \node[state]         (q2) [below of=q1] {$q_2$};
\begin{scope}[node distance=0.8cm]
	\node [below of=q2] {$A_2$};
\end{scope}
\path (q1) edge              node {$(b,p_2)$} (q2)
	edge[loop above]              node {$c, (b,p_1)$} (q1)
	(q2) edge[loop left]              node[align=center] {$c$ \\ $(b, p_1)$ \\ $(b,
	p_2)$ } (q2);

  \node[initial,state] (s1) [right=2cm of q1]       {$s_1$};
  \node[state]         (s2) [below of=s1] {$s_2$};
\begin{scope}[node distance=0.8cm]
	\node [below of=s2] {$A_3$};
\end{scope}
\path (s1) edge              node {$(c,q_2)$} (s2)
	edge[loop above]              node {$d, (c,q_1)$} (s1)
	(s2) edge[loop left]              node[align=center] { {\scriptsize $d$} \\
		{\scriptsize$(c,q_1)$} \\
	{\scriptsize $(c,q_2)$}} (s2);
  
    \node[initial,state] (t1) [right=3.6cm of s1]       {$t_1$};
  \node[state]         (t2) [below of=t1] {$t_2$};
\begin{scope}[node distance=0.8cm]
	\node [below of=t2] {$B$};
\end{scope}
\path (t1) edge              node {$(d,s_2)$} (t2)
	edge[loop above]              node {$(d,s_1), (c,q_1,s_1), (c,q_1,s_2)$ \\
	$(c,q_2,s_1), (c,q_2,s_2)$} (t1)
	(t2) edge[loop left]              node {$(d,s_1), (d,s_2)$ \\ $(c,q_1,s_1)$ \\ $(c,q_1,s_2)$ \\
	$(c,q_2,s_1)$ \\ $(c,q_2,s_2)$} (t2);
\end{tikzpicture}
\caption{Cascade product of localized reset asynchronous automata}
\label{fig:lcascade-example}
\end{figure}
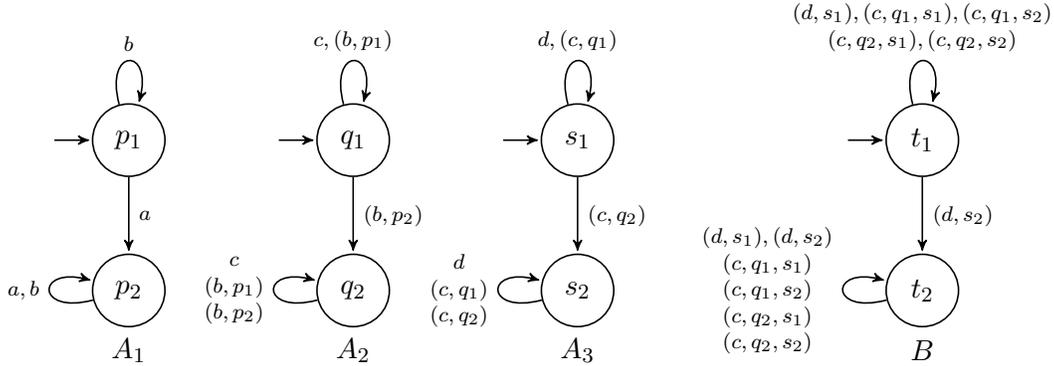

	It is not difficult to see
	that the asynchronous transducer $\hat A$ labels process $1$ events (resp. process
	$2$ events, and process $3$ events) by $p_2$ (resp. $q_2$ and $s_2$) if and only
	if that event has an $a$-labelled event in its past. Since $B$ detects the
	existence of letter $(d,s_2)$ by changing its state, $A \lc B$ can recognize
	for instance the language ``there exists a $d$ that has an $a$ in its past''.
\end{exa}

To conclude this section, we relate local cascade products with the $\Gamma$-labelling
functions introduced in Section~\ref{sec: asynchronous automata}.

\begin{prop}\label{prop: inverse labelling}
  Let $\aaa = (\{\ls_i\}, \{\lt_a\}, \sinit)$ be an asynchronous automaton on the
  distributed alphabet $\alphloc$, and let $\theta\colon \traces \to
  \ltraces$ be the $\lab$-labelling function computed by an asynchronous
  $\lab$-transducer of the form $\hat\aaa = (\{\ls_i\}, \{\lt_a\}, \sinit, \{\mu_a\})$.
  Let $B$ be an asynchronous automaton on the distributed alphabet $\dlalphabet$
  accepting a language $L \subseteq \ltraces$.  Then $\theta\inv(L)$ is
  accepted by $\aaa\lc B$.
\end{prop}

\begin{proof}
  Let $\gamma\colon \tralphabet\to \lalphabet$ be the map given by
  $\gamma(a,s_a) = (a,\mu_a(s_a))$.  We also denote by $\gamma$ the extension of this map
  to a morphism from $\tracesdtr$ to $\ltraces$.  By definition,
  $\theta = \gamma\circ \chi_\aaa$.  In particular $\theta\inv(L) =
  \chi_\aaa\inv(\gamma\inv(L))$.

  Let $B'$ be the automaton on alphabet $\tralphabet$ obtained from $B$ by keeping the
  same state sets and global initial and accepting states, and letting the transition
  $\delta_{(a,s_a)}$ be equal to the transition $\delta_{(a,\mu_a(s_a))}$ of $B$.  It is
  clear that $B'$ accepts $\gamma\inv(L)$ and, by Theorem~\ref{thm:wpp1}, $\aaa\lc B$
  accepts $\theta\inv(L) = \chi_\aaa\inv(\gamma\inv(L))$.
\end{proof}

\subsection{Questions of decomposition}\label{sec: decomposition}

Most of the known characterizations of interesting classes of recognizable trace languages
(\textit{e.g.}, star-free languages, languages definable in first-order logic or in a
global or local temporal logic, see
\cite{guaiana1992star,ebinger1996logical,DiGa02jcss,DBLP:journals/iandc/DiekertG06}) are
in terms of morphisms into ordinary transformation monoids or of syntactic monoids
(equivalently, of the transition monoids of certain canonical minimal diamond-automata).
In contrast, Zielonka's theorem simulates a morphism into a tm by an asynchronous morphism
into an atm.  We seek an asynchronous version of the Krohn-Rhodes theorem that would be
similar in spirit: starting with a morphism into a tm, we ask whether it can be simulated
by an asynchronous morphism into an asynchronous wreath product of `simpler' asynchronous
transformation monoids.  A positive resolution in this form would help us lift the
sequential or diamond automata-theoretic chacterizations of first-order definable trace
languages mentioned above to asynchronous automata-theoretic characterizations (cf.
Theorem~\ref{thm:lcascade-loctlv} below).

We would also like the statement of the proposed asynchronous version of the Krohn-Rhodes
theorem to coincide with the classical sequential Krohn-Rhodes Theorem
(Theorem~\ref{wordkr}) when the underlying distributed alphabet involves only one process.

We first extend the notion of simulation (Definition~\ref{wordsimulation}) to trace
morphisms as follows.

\begin{defi}
  Let $\phi\colon \traces \to T=(X,M)$ and $\psi\colon \traces \to T' = (Y,N)$ be
  morphisms to transformation monoids.  We say that $\psi$ \emph{simulates} $\phi$ if
  there exists a surjective function $f\colon Y \to X$ such that, for all $a \in \Sigma$
  and all $y\in Y$, $f(\psi(a)(y)) = \phi(a)(f(y))$.

  If $\psi$ happens to be an asynchronous morphism to an atm, we say that $\psi$ simulates
  $\phi$ if it is the case when $\psi$ is viewed as a morphism to the tm underlying $T'$.
\end{defi}

We note that Zielonka's fundamental theorem~\cite{zielonka1987notes}, already mentioned in
Sections~\ref{section: recog} and~\ref{sec: asychm}, can actually be rephrased as follows
(see \cite{mukund2012automata}).

\begin{thmC}[(Zielonka Theorem)]\label{zielonka}
  Every morphism $\phi\colon \traces \to T$ to a finite tm is simulated by an asynchronous
  morphism $\psi\colon \traces \to T'$ to an atm.\footnote{The proof in
  \cite{mukund2012automata} provides an automata-theoretic simulation of any
  diamond-automaton by an asynchronous automaton.  The morphism version stated here is an
  easy consequence of this result and the correspondence between morphisms (resp.\
  asynchronous morphisms) defined on trace monoids and diamond-automata (resp.\
  asynchronous automata).}
\end{thmC}

Recall the atm $U_2[p]$ defined in Example~\ref{ex:u2l}, an asynchronous analogue at
process $p\in\pset$ of the tm $U_2$.  If $G$ is a group, recall that $(G,G)$ is a tm and
let $G[p]$ be its asynchronous analogue at process $p$ (again, see Example~\ref{ex:u2l}).

\begin{defi}
  We say that a distributed alphabet $\dalphabet$ has the \emph{asynchronous Krohn-Rhodes
  property} ($\dalphabet$ is \emph{aKR}, for short) if every morphism $\phi\colon\traces
  \to T$ to a tm $T = (X,M)$ is simulated by an asynchronous morphism to an atm $T'$ which
  is the asynchronous wreath product of asynchronous transformation monoids of the form
  $U_2[p]$ or $G[p]$, where $p\in\pset$ and $G$ is a simple group dividing $M$.
\end{defi}

It is an interesting question to characterize which distributed alphabets are aKR, or even
whether all are.  While we are not able to answer this question, we show in
Section~\ref{sec: acyclic} that acyclic architectures are aKR. In
Section~\ref{sec:aperiodic}, we show that a weaker property holds when we restrict our
attention to morphisms from $\traces$ to aperiodic transformation monoids.  See also the
discussion in Section~\ref{sec: conclusion}.

In view of our discussion so far, it is clear that establishing that a distributed
alphabet is aKR amounts to a simultaneous generalization, for this particular distributed
alphabet, of the Krohn-Rhodes theorem (Theorem~\ref{wordkr}) and of Zielonka's theorem
(Theorem~\ref{zielonka}).

\section{The case of acyclic architectures}\label{sec: acyclic}
\begin{defi}
  The \emph{communication graph} of a distributed alphabet $ {\{\alphabet_i\}}_{i \in
  \pset}$ is the undirected graph $G = (\pset,E)$ where $E = \{ (i,j) \in \pset \times
  \pset \mid i \neq j \text{ and } \alphabet_i \cap \alphabet_j \neq \emptyset \}$.  An
  \emph{acyclic architecture} is a distributed alphabet whose communication graph is
  acyclic.
\end{defi}

Observe that if $\dalphabet$ is an acyclic architecture, then no action is shared by more
than two processes.  We note that Zielonka's theorem admits a simpler proof in this
case~\cite{acycliczielonka-anca}.  Our objective in this section is to establish the
following result.

\begin{thm}\label{thm:acyclickr}
  If $\dalphabet$ is an acyclic architecture, then $\dalphabet$ has the asynchronous
  Krohn-Rhodes property.
\end{thm}

To prove Theorem~\ref{thm:acyclickr}, we need several technical lemmas. 
We first define a notion of wreath product of trace morphisms into tm's.

\begin{defi}\label{def: asynch wrp}
	Let $\phi\colon\traces \to (X,M)$ and $\psi\colon \traces[\s \times X] \to
  (Y,N)$ be morphisms to transformation monoids.
  For each $a \in \Sigma$ and $x\in X$, let $f_a(x) = \psi(a,x)$.  The map
  $\phi\wr\psi$ is defined on $\Sigma$ by letting $(\phi\wr\psi)(a) =
  (\phi(a),f_a)$ for each $a\in \Sigma$.
\end{defi}
In general, the above map $\phi \wr \psi$ may not extend to a morphism from the trace
monoid (see \cite{thesis, versionone} for an example). However, it does so under a technical condition.

\begin{lem}\label{lem: asynch wrp}
  Let $\phi$ and $\psi$ be as in Definition~\ref{def: asynch wrp}.  If, for all
  independent letters $a,b$ and for all $x \in X$, we have $\psi(b, \phi(a)(x)) =
  \psi(b, x)$, then $\phi \wr \psi$ induces a morphism from $\traces$ to the tm
  $(X,M) \wr (Y,N)$.
\end{lem}
\begin{proof}
  We verify that, if $a$ and $b$ are independent
  letters, then $(\phi \wr \psi)(ab) = (\phi \wr \psi)(ba)$.  More precisely,
  let $(\phi \wr \psi)(ab)=(\eta_1,\eta_2)$.  By definition, we have
  $\eta_1=\phi(a)\phi(b)$ and for each $x\in X$, we have
  $\eta_2(x)=\psi(a,x)+\psi(b,\phi(a)(x))$.  Using our hypothesis, we get
  $\eta_2(x)=\psi(a,x)+\psi(b,x)$.  Since $\phi$ and $\psi$ are trace morphisms,
  we deduce that $(\phi \wr \psi)(ab) = (\phi \wr \psi)(ba)$.
\end{proof}

We now show how to simulate a morphism from a trace monoid to a tm, by the wreath
product of morphisms to transformation monoids that are simpler.

\begin{lem}\label{lem:split}
  For a process $p\in\pset$, let us define the set $\Sigma_0 = \{ a \in \Sigma \mid \loc(a)
  = \{p\} \}$.  Let $\phi \colon \traces \to (X,M)$ be a morphism to a tm.  There exist
  morphisms $\phi_1 \colon \traces \to (X_1, M_1)$ and $\phi_2 \colon \traces[\Sigma
  \times X_1] \to (X,M)$ such that
  \begin{itemize}
    \item $\phi_1$ ignores all non-process-$p$ letters (that is, $\phi_1$ maps letters
    outside $\Sigma_p$ to the identity transformation of $X_1$).  Moreover, any group
    dividing $M_1$ also divides $M$.
	
    \item $\phi_2$ ignores all letters local only to process $p$ (that is, $\phi_2$ maps
    letters in $\Sigma_0\times X_1$ to the identity transformation of $X$).  Also, for any
    $a \notin \Sigma_p$, $\phi_2(a,x) = \phi(a)$.

    \item The morphism $\phi_1 \wr \phi_2 \colon \traces \to (X_1, M_1) \wr (X, M)$
    simulates $\phi$.
  \end{itemize}
\end{lem}

\begin{proof}
  Observe that $\Sigma_0 \subseteq \Sigma_p$. The letters in $\Sigma_0$ are mutually
  dependent, that is, $\Sigma_0^*$ is in fact a submonoid of $\traces$. Let $N$ be the
  submonoid of $M$ generated by $\phi(\Sigma_0)$; we can view $N$ as $N =
  \phi(\Sigma_0^*)$.  For each $n \in N$, let $\bar{n}$ be the constant transformation of
  $N$ which maps every element to $n$.  The set $\bar N = N \cup\{\bar n \mid n\in N\}$ is
  easily verified to be a submonoid of $\trans(N)$.

  Let $\phi_1 \colon  \Sigma \to (N, \bar N)$ be defined by letting 
  \[
  \varphi_1(a) = 
	\begin{cases}
		\varphi(a) &\text{ if } a \in \Sigma_0, \\
		\overline{\mathrm{id}} &\text{ if } a \in \Sigma_p \setminus \Sigma_0, \\
		\mathrm{id} &\text{ if }a \not\in\Sigma_p \,. 
	\end{cases} 
  \]
  If $a$ and $b$ are independent letters, then they cannot be both in $\Sigma_p$, one of
  $\phi_1(a)$ and $\phi_1(b)$ at least is the identity, and hence $\phi_1(a)$ and
  $\phi_1(b)$ commute.  It follows that $\phi_1$ naturally extends to a morphism
  $\phi_1\colon \traces \to (N,\bar N)$, which is the identity on the submonoid generated
  by $\Sigma\setminus\Sigma_p$.

  We note that if a group $G$ divides $\bar N$, then it also divides $N$.  Indeed, suppose
  that $\tau$ is a morphism defined on a submonoid $N'\subseteq \bar N$ onto $G$.  Each
  element of the form $\bar n$ ($n\in N$) in $N'$ is idempotent, and hence $\tau(\bar n)$
  is the identity $1_G$ of $G$.  In particular, the restriction of $\tau$ to the submonoid
  $N'' = N' \cap N$ has the same range as $\tau$, that is, $G$ is a quotient of $N''$ and
  hence $G \divides N''$.

  Observe that, if $w = a_1\dots a_r$, $i$ is maximal such that $a_i\in
  \Sigma_p\setminus\Sigma_0$ ($i = 0$ if $w$ has no letter in $\Sigma_p\setminus\Sigma_0$)
  and $w'$ is the projection of $a_{i+1}\cdots a_r$ onto $\Sigma_0^*$, then $\phi_1(w) =
  \phi(w')$.  In other words, $\phi_1(w)$ is the evaluation under $\phi$ of the word read
  by process~$p$ since its last joint action with a neighbour.

  Now let us make $\Sigma\times N$ a distributed alphabet (over $\pset$) by letting
  $\loc((a,n)) = \loc(a)$ for each $(a,n)\in \Sigma\times N$.  Let also
  \[
  \varphi_2(a,n) = 
	\begin{cases}
		\mathrm{id} &\text{ if } a\in\Sigma_0, \\
		n\varphi(a) &\text{ if } a \in \Sigma_p\setminus \Sigma_0, \\
		\varphi(a)  &\text{ if } a\not\in\Sigma_p \,.
	\end{cases} 
  \]
  If $(a,n)$ and $(b,n')$ are independent letters, that is, if $a$ and $b$ are independent
  in $\Sigma$, then $\phi(a)$ and $\phi(b)$ commute because $\phi$ is defined on
  $\traces$.  Moreover, either $a,b \not\in\Sigma_p$, or $a\in \Sigma_p$ and
  $b\not\in\Sigma_p$ (or vice versa).  In the first case, $\phi_2(a,n) = \phi(a)$ and
  $\phi_2(b,n') = \phi(b)$ commute.  In the second case, $\phi(b)$ commutes with $n$ since
  the latter is (by definition of $N$) the $\phi$-image of a word in $\Sigma_0^*$.  It
  follows that $\phi_2(a,n)$ and $\phi_2(b,n')$ commute.

  Let $\eta\colon \Sigma \to \bar N \times \trans(N,M)$ be the map $\eta =
  \phi_1\wr\phi_2$ in Definition~\ref{def: asynch wrp}.  Assume that $a,b\in\Sigma$ are
  independent letters.  Let $n\in N$.  If $a\notin\Sigma_p$ then $\phi_1(a)(n)=n$ and if
  $a\in\Sigma_p$, then $b\notin\Sigma_p$ and $\phi_2(b,\phi_1(a)(n))=\phi(b)=\phi_2(b,n)$.
  In both cases, we get $\phi_2(b,\phi_1(a)(n))=\phi_2(b,n)$.  By Lemma~\ref{lem: asynch
  wrp}, it follows that $\eta$ extends to a morphism $\eta\colon\traces\to (N,\bar N) \wr
  (X,M)$.

  Finally, for each $(n,x)\in N\times X$, let $f(n,x) = n(x) \in X$ (recall that $N
  \subseteq M \subseteq \trans(X)$).  For any $a \in \Sigma$, $n \in N$ and $x \in X$, we
  note that
  \[
  f(\eta(a)(n,x)) = f(\phi_1(a)(n),\phi_2(a,n)(x)) \,.
  \]
A case analysis proves that $\eta$ simulates $\phi$ via $f$.
  \begin{enumerate}[i)]
    \item $a\in\Sigma_0$.  In this case, $\phi_2(a,n)=\mathrm{id}$ and the above becomes
    $f(\phi(a)(n), x) = f(n\phi(a), x) = (n\phi(a))(x) = \phi(a)(n(x)) = \phi(a)(f(n,x))$.
    The first equality is because the transformation by $\phi(a) \in N$ is the right
    multiplication on $N$.
    
    \item $a\in\Sigma_p\setminus\Sigma_0$.  In this case, $\phi_1(a)(n)=\mathrm{id}$ and
    the above becomes $f(\mathrm{id}, (n\phi(a))(x)) = (n\phi(a))(x) = \phi(a)(f(n,x))$.
    
    \item $a\notin\Sigma_p$.  In this case,
    \[
    f(n, \phi(a)(x)) = n(\phi(a)(x)) = (\phi(a)n)(x) = (n\phi(a))(x) = \phi(a)(f(n,x)) \,.
    \]
    The penultimate equality follows from the fact that the generators of $N$ commute with
    $\phi(a)$ since $p \notin \loc(a)$.
  \end{enumerate}
  In all cases, $f(\eta(a)(n,x)) = \varphi(a)(f(n,x))$.  This completes the proof.
\end{proof}

\begin{rem}\label{rem:empty}
	If, in Lemma~\ref{lem:split}, $\Sigma_0$ is empty, then $M_1$ is the trivial monoid (so
	$\varphi_1(a) = \mathrm{id}$ for every letter $a$) and $\varphi_2(a,x) = \varphi(a)$.
\end{rem}

Our next lemma shows how to combine simulations by asynchronous morphisms, under certain
technical assumptions.

\begin{lem}\label{lem:join}
Let $\phi_1 \colon \traces \to (X,M)$ and $\phi_2 \colon \traces[\Sigma \times X] \to
  (Y,N)$ be morphisms to transformation monoids.  Suppose that $\phi_1$ is simulated by an
  asynchronous morphism $\psi_1 \colon \traces \to (\{S_i\},P)$ via a map $f_1 \colon \gs
  \to X$.  Suppose also that, for all $a \in \Sigma$, if $s,s' \in \gs$ and $s_a = s'_a$,
  then $\phi_2(a, f_1(s)) = \phi_2(a, f_1(s'))$.  Then
  \begin{itemize}
    \item The map $\phi_1 \wr \phi_2$ extends to a morphism from $\traces$ to $(X,M) \wr
    (Y,N)$.

    \item The map $\phi'_2$, defined on $\tralphabet$ by letting $\phi'_2(a,s_a) =
    \phi_2(a,f_1(s))$ for all $a\in \Sigma$ and $s\in\gs$, extends to a morphism from
    $\traces[\s\loctimes\gs]$ to $(Y,N)$.

    \item If $\phi'_2$ is simulated by an asynchronous morphism $\psi_2 \colon \traces[\s\loctimes\gs]
    \to (\{Q_i\}, R)$, then $\phi_1 \wr \phi_2 \colon \traces \to (X,M) \wr (Y,N)$ is
    simulated by the asynchronous morphism $\psi_1 \wras \psi_2 \colon \traces \to
    (\{S_i\},P) \wras (\{Q_i\}, R)$.
  \end{itemize}
\end{lem}

\begin{proof}
  We rely on Lemma~\ref{lem: asynch wrp} to prove the first statement, that is, we want to
  show that, if $a,b \in \Sigma$ are independent letters and $x\in X$, then
  $\phi_2(a,\phi_1(b)(x))=\phi_2(a,x)$.  Since $f_1$ is surjective, we can choose
  $s\in\gs$ such that $f_1(s) = x$.  Then $\phi_2(a,x) = \phi_2(a,f_1(s))$.  Also,
  $\phi_2(a,\phi_1(b)(x)) = \phi_2(a,\phi_1(b)(f_1(s))) = \phi_2(a,f_1(\psi_1(b)(s)))$.
  Since $a$ and $b$ are independent and $\psi_1$ is an asynchronous morphism, $s_a =
  (\psi_1(b)(s))_a$ and our assumption on $f_1$ and $\phi_2$ shows that $\phi_2(a,x) =
  \phi_2(a, \phi_1(b)(x))$.  This concludes the proof of the first statement.

  The fact that $\phi'_2$ extends to a morphism defined on $\tracesdtr$ follows directly
  from the fact that $\phi_2$ is a morphism from the trace monoid $\traces[\Sigma \times
  X]$.

  Suppose $\phi'_2$ is simulated by $\psi_2$ via the map $f_2 \colon Q \to Y$.  For
  each $s\in \gs$, $q\in Q$, we let $f(s,q) = (f_1(s), f_2(q))$.  For any $a \in
  \Sigma$, we have
  \begin{align*}
    f((\psi_1 \wras \psi_2)(a)(s,q)) &= f(\psi_1(a)(s), \psi_2(a,s_a)(q)) \\
    &= (f_1(\psi_1(a)(s)), f_2(\psi_2(a,s_a)(q)))
    \\
    &= (\phi_1(a)(f_1(s)),
    \phi'_2(a,s_a)(f_2(q))) \\
    &= (\phi_1(a)(f_1(s)), \phi_2(a,
    f_1(s))(f_2(q))) \\
    &= (\phi_1 \wr \phi_2) (a) (f_1(s),
    f_2(q)) \\
    &= (\phi_1 \wr \phi_2) (a) (f(s,q))
  \end{align*}
  This completes the proof.
\end{proof}

%
\begin{proof}[Proof (of Theorem~{\ref{thm:acyclickr}})]
The proof proceeds via induction on the number of
processes.
The base case, with only one
process (and hence no distributed structure on the alphabet), is the Krohn-Rhodes theorem
(Theorem~\ref{wordkr}).
	
Let now $\pset = \{1, 2, \ldots, k \}$, with $k\ge 2$, and assume that the theorem holds
for acyclic architectures with less than $k$ processes.  Since the communication graph is
acyclic, there exists a `leaf' process which communicates with at most one other process.
Without loss of generality, let this leaf process be $1$, and its only neighbouring
process be $2$ (if process $1$ has no neighbour, then process $2$ can be any other
process).  We let $\Sigma_0$ be the set of letters $a$ such that $\loc(a) = \{1\}$.  By
our assumptions, $\loc(a) = \{1,2\}$ for every $a\in \Sigma_1\setminus\Sigma_0$.

Let $\phi \colon \traces \to (X,M)$ be a morphism into a tm.  Let $\varphi_1 \colon
\traces \to (X_1, M_1)$ and $\varphi_2 \colon \traces[\Sigma \times X_1] \to (X,M)$ be the
morphisms into transformation monoids given by Lemma~\ref{lem:split} for $p = 1$.

Note that no two letters of $\Sigma_1$ are independent, so that the submonoid of $\traces$
generated by $\Sigma_1$ is freely generated by $\Sigma_1$: we write it $\Sigma_1^*$.  Let
$\phi'_1\colon \Sigma_1^*\to (X_1,M_1)$ be the restriction of $\phi_1$ to $\Sigma_1^*$.
The Krohn-Rhodes theorem shows that $\phi'_1$ is simulated by a morphism $\psi'_1\colon
\Sigma_1^* \to T$, where $T$ is a wreath product of transformation monoids of the form
$U_2$ and $(G,G)$, where $G$ is a nontrivial simple group $G$ dividing $M_1$ --- and
hence, by Lemma~\ref{lem:split}, dividing $M$.

Next, we extend $\psi'_1$ to a morphism $\psi_1\colon \traces \to T$ by letting $\psi_1(a)
= \mathrm{id}$ for each letter $a\not\in\Sigma_1$.  Consider the atm $T[1]$ defined in
Example~\ref{ex:u2l}.  It is easily verified that $\psi_1 \colon \traces \to T[1]$ induces
an asynchronous morphism (Lemma~\ref{lem:extends-asyn-morph}), which simulates $\phi_1$,
and that $T[1]$ is an asynchronous wreath product of asynchronous transformation monoids
of the required form (that is, of the form $U_2[1]$ or $G[1]$ for a group $G$ dividing
$M$).

Suppose $\gs$ is the global state space of $T[1]$ and $\psi_1$ simulates $\phi_1$ via $f_1
\colon \gs \to X_1$. Let $s,s' \in \gs$ be global states such that $s_a = s'_a$.  If $1
\notin \loc(a)$, then $\phi_2(a, f_1(s)) = \phi(a) = \phi_2(a, f_1(s'))$ by
Lemma~\ref{lem:split}.  If $1 \in \loc(a)$, then $s_a = s'_a$ implies $s=s'$ by the
structure of $T[1]$. Hence $\phi_2(a, f_1(s)) = \phi_2(a, f_1(s'))$.

Let $\phi'_2 \colon \traces[\s\loctimes\gs] \to (X,M)$ be the morphism given in Lemma~\ref{lem:join}
(where $\phi'_2(a,s_a) = \phi_2(a, f_1(s))$).  Consider the distributed alphabet
$\wt{\Sigma'}$ over $\pset \setminus \{1\}$ where $\Sigma'_i = (\tralphabet)_i$ for every
$i\in \pset$, $i\ne 1$, where the location of a letter $(a,s_a)$ is $\loc(a) \setminus
\{1\}$.

Suppose that letters $(a,s_a)$ and $(b,s'_b)$ are independent in $\wt{\Sigma'}$.  We first
verify that they are independent in $\wt\Sigma$ as well.  If this is not the case, then
$\loc(a) \cap \loc(b) = \{1\}$.  However, as we noticed before, letters whose location
contains 1 and is not reduced to $\{1\}$ (as is the case for all letters $a$ such that
some $(a,s_a) \in \Sigma'$) have location $\{1,2\}$, and this means that $\loc(a) \cap
\loc(b) = \{1,2\}$, a contradiction.

It follows that $\traces[\Sigma']$ is a submonoid of $\traces[\s\loctimes\gs]$, and we now consider the
restriction $\phi''_2$ of $\phi'_2$ to $\traces[\Sigma']$.

By induction hypothesis, $\phi''_2$ is simulated by an asynchronous morphism $\psi_2
\colon \traces[\Sigma'] \to T'$, where $T'$ is an asynchronous wreath product of
asynchronous transformation monoids of the form $U_2[p]$ or $G[p]$ for some simple group
$G$ dividing $M$, and some $p \in \pset \setminus \{1\}$.  These asynchronous
transformation monoids can again be trivially extended over $\pset$ (by adding a singleton
state set for process 1).  The corresponding extension of $T'$, denoted
$T'[{\uparrow}\pset]$ is an asynchronous wreath product of the desired form.  Moreover
$\psi_2$ can also be extended over $\tralphabet$ by mapping any letter in $(\tralphabet)
\setminus \Sigma'$ to the identity.  It is easily verified that $\psi_2$ is an
asynchronous morphism, which simulates $\phi'_2$.

We can now apply Lemma~\ref{lem:join} to conclude the proof.
\end{proof}

\section{Temporal logics, first-order trace languages \& local cascade products}\label{sec:aperiodic}

Recall that star-free trace languages, aperiodic trace languages (those recognized by an
aperiodic monoid) and first-order definable trace languages coincide, by the combined
results of Guaiana, Restivo and Salemi \cite{guaiana1992star} and Ebinger and Muscholl
\cite{ebinger1996logical}.

A consequence of Theorem~\ref{thm:acyclickr} is that any \emph{aperiodic} trace language
over an acyclic architecture, and indeed over any aKR distributed alphabet, is recognized
by an asynchronous wreath product of 2-state asynchronous transformation monoids of the
form $U_2[p]$ (see Example~\ref{ex:u2l}).  Proposition~\ref{prop: local cascade} then
shows that it is accepted by a local cascade product of localized two-state reset
asynchronous automata.  For convenience, we denote these automata by $U_2[p]$ as well:
$U_2[p]$ has state set $\{S_i\}$, where $S_p=\{1,2\}$ and each $S_i$ ($i\neq p$) is a
singleton, and transitions as follows.  The alphabet $\alphabet_p$ contains two disjoint
subsets $R_1$ and $R_2$ that reset the states in $S_p$ to $1$ and $2$, respectively.  All
remaining letters act as the identity, in particular letters in
$\alphabet_p\setminus(R_1\cup R_2)$.  In this section, we aim at generalizing this result
to any distributed alphabet. Our route towards this utilizes yet another formalism used to
classify trace languages, that of temporal logics.

\subsection{Local temporal logics}\label{subsec:localtemporal}
We first introduce a process-based past-oriented local temporal logic over traces, called
$\loctlf$.  This logic, as well as some of its fragments, turns out be as expressive as 
first-order logic over traces, see Theorem~\ref{thm:loctlfo} below. Furthermore, we have
chosen the logic in a way that facilitates showing correspondence with local
cascade products (see the proofs of Theorem~{\ref{thm:lcascade}} and
Corollary~{\ref{cor:lcascade-loctlv}}).
The syntax of $\loctlf$ is as follows.
\begin{align*}
  \text{Event formula }\qquad\alpha &::= a \mid \neg\alpha \mid \alpha\vee\alpha \mid
  \Yleq{i}{j} \mid \Y_i{\alpha} \mid  \alpha \Si_i \alpha 
  \hspace{9mm}a \in \alphabet, i,j \in \pset 
  \\
  \text{Trace formula }\qquad\beta &::= \exists_i \alpha \mid \neg \beta \mid \beta \vee \beta 
\end{align*}
Turning to the semantics of $\loctlf$, each event formula is evaluated at an event of a
trace $t \in \traces$, say, $t = (E, \leq, \lambda)$.  For any event $e\in E$ and process
$i\in\pset$, we denote by $e_i$ the unique maximal event of $\Da e \cap E_i$, if it
exists, i.e., if $\Da e\cap E_i\neq\emptyset$.  Then
\begin{align*}
  t,e & \models a && \text{if } \lambda(e) = a \\
  t,e & \models \neg \alpha && \text{if } t,e \not\models \alpha \\
  t,e & \models \alpha \vee \alpha' && \text{if } t,e \models \alpha \mbox{~or~} t,e \models \alpha' \\
  t,e & \models \Yleq{i}{j}  && \text{if } e_i,e_j \text{ exist, and } e_i \leq e_j \\
  t,e & \models \Y_i{\alpha}  && \text{if } e_i \text{ exists, and } t, e_i \models \alpha \\
  t,e &\models \alpha \Si_i \alpha' && \text{if } e \in E_i, \text{ there exists }
  f \in E_i \text{ such that } f < e \text{ and } t,f \models \alpha'
  \\
  &&& \text{and, for all } g \in E_i, \enspace f < g < e \text{ implies } t,g \models \alpha.
\end{align*}
Note that, $\Y_i$ is a modality expecting an argument $\alpha$, whereas $a$ and
$\Yleq{i}{j}$ are basic constant formulas.  Also, the since operators $\Si_i$ are strict.
$\loctlf$ trace formulas are evaluated on traces by interpreting the Boolean connectives
in the natural way and by letting
\[
t \models \exists_i \alpha \qquad 
\text{if there exists a maximal $i$-event $e$ in $t$ such that } t,e \models \alpha \,.
\]
We will also consider various fragments of $\loctlf$ with inherited semantics.  $\loctl$
is the fragment where the constants $\Yleq{i}{j}$ are not allowed.  Similarly, $\loctlv$
is the fragment without the modalities $\Y_i$, and $\sprtl$ is the fragment where both
$\Yleq{i}{j}$ and $\Y_i$ are disallowed.

We say that a trace language $L\subseteq\traces$ is definable in $\loctlf$ (resp.\ one of
its fragments) if there is a trace formula $\beta$ in $\loctlf$ (resp.\ in the
corresponding fragment) such that $L=\{t\in\traces \mid t\models\beta\}$.

In Theorem~\ref{thm:loctlfo}, the three local temporal logics $\loctlf$, $\loctlv$ and
$\loctl$ are shown to have equal expressive power. We first prove in
Lemma~\ref{lem:yibyconstant} that $\loctlv$ is at least as expressive as $\loctl$; we only need to
show that the modality $\Y_i$ of $\loctl$ can be expressed in $\loctlv$. To this end, we
will use the following derived constants $\Yeq{i}{j}=(\Yleq{i}{j})\wedge(\Yleq{j}{i})$ and
$\Ylt{i}{j}=(\Yleq{i}{j})\wedge\neg(\Yleq{j}{i})$. Notice that $\Yeq{i}{i}$ is equivalent
to $\Y_i\top$ and simply means that there are $i$ events in the strict past of the current
event.
  
  For $i\in\pset$ and an event formula $\alpha$, we define
  \begin{align*}
    \Y_i^{1}\alpha & = \bigvee_{j}\Yeq{i}{j}\wedge(\bot\Si_j\alpha) 
    \\
    \Y_i^{m+1}\alpha & = \bigvee_{j}\Ylt{i}{j}\wedge
    [\Ylt{i}{j}\Si_j(\Y_i^{m}\alpha)] 
  \end{align*}
 \begin{rem}
  Notice that an event satisfying $\Yeq{i}{j}\wedge(\bot\Si_j\alpha)$ or
  $\Ylt{i}{j}\wedge[\Ylt{i}{j}\Si_j(\Y_i^{m}\alpha)]$ must be a $j$-event.  Hence, we may
  restrict the use of $\Yeq{i}{j}$ and $\Ylt{i}{j}$ to be at $j$-events only and still get
  an expressively complete logic.
 \end{rem}
  
  \begin{lem}\label{lem:yim1}
    Consider a trace $t \in \traces$.  For any event $e$ in the trace, any process $i$,
    and any natural number $m$, if $t,e\models\Y_i^m\alpha$ then $t,e\models\Y_i\alpha$.
  \end{lem}
  \begin{proof}
    The proof is by induction on $m$.  The base case is when $m=1$.  Suppose that
    $t,e\models\Yeq{i}{j}\wedge(\bot\Si_j\alpha)$ for some $j\in\pset$.  Due to
    $\Yeq{i}{j}$, we know that $e_i,e_j$ exist and $e_i=e_j$.  Now, from $\bot\Si_j\alpha$,
    we get that $e\in E_j$ and $t,e_j\models\alpha$.  We deduce that $t,e_i\models\alpha$
    and $t,e\models\Y_i\alpha$.
    
    For the inductive step, let $m\geq1$ and suppose that $t,e\models\Ylt{i}{j}\wedge
    [\Ylt{i}{j}\Si_j(\Y_i^{m}\alpha)]$ for some $j\in\pset$.  Due to $\Ylt{i}{j}$, we know
    that $e_i,e_j$ exist and $e_i<e_j$.  Now, from
    $t,e\models\Ylt{i}{j}\Si_j(\Y_i^{m}\alpha)$, we deduce that $e\in E_j$ and there exists
    an event $f \in E_j$ in the strict past of $e$ such that $t,f \models \Y_i^{m}\alpha$
    and at all $j$-events between $e$ and $f$, the formula $\Ylt{i}{j}$ is true.  By 
    induction, we get $t,f\models\Y_i\alpha$. Let
    $f=f^{k}<f^{k-1}<\cdots<f^{1}<f^{0}=e$ be the sequence of $j$-events between $f$ and
    $e$.  For all $0\leq\ell<k$, we have $f^{\ell}_j=f^{\ell+1}$ and
    $t,f^{\ell}\models\Ylt{i}{j}$.  We deduce by induction on $\ell$ that $e_i=f^{\ell}_i$
    for all $0\leq\ell\leq k$.  This is clear when $\ell=0$ since $e=f^{0}$ and for the
    inductive step it follows from $e_i=f^{\ell}_i<f^{\ell}_j=f^{\ell+1}$.
    Finally, $e_i=f_i$ and we get $t,e_i\models\alpha$ as 
    desired.
  \end{proof}
  
  \begin{lem}\label{lem:yim2}
    Consider a trace $t \in \traces$.  For any event $e$ in the trace, any process $i$,
    if $t,e\models\Y_i\alpha$ then $t,e\models\Y_i^{m}\alpha$ for some $m\leq|\pset|$.
  \end{lem}
  \begin{proof}
    Assume that $t,e\models\Y_i\alpha$, i.e., $e_i$ exists and $t, e_i\models\alpha$.  Since
    $e_i$ is in the strict past of $e$, there exists $f$ such that $e_i\isucc f \leq e$.
    It is well-known that there is $m\geq1$ and a sequence of events
    $f=f^{1}<f^{2}<\cdots<f^{m}=e$ and processes $j_\ell$ such that
    $j_\ell\in\loc(f^{\ell-1})\cap\loc(f^{\ell})$ for all $1<\ell\leq m$.  Moreover, we
    may assume that the processes $j_2,\ldots,j_m$ are pairwise distinct, and also
    different from $i$ since $e_i\isucc f$.  We get $m\leq|\pset|$.
    
    We show by induction on $\ell$ that $e_i=f^{\ell}_i$ and
    $t,f^{\ell}\models\Y_i^{\ell}\alpha$.  For the base case $\ell=1$, since $e_i\isucc
    f^{1}$ we have $e_i=f^{1}_i$ and we find $j_1\in\loc(e_i)\cap\loc(f^{1})$.
    We get $t,f^{1}\models\Yeq{i}{j_1}\wedge(\bot\Si_{j_1}\alpha)$.  Therefore,
    $t,f^{1}\models\Y_i^{1}\alpha$ and we are done.  Consider now $1\leq\ell<m$ and assume
    that $e_i=f^{\ell}_i$ and $t,f^{\ell}\models\Y_i^{\ell}\alpha$. Since 
    $e_i<f^{\ell}<f^{\ell+1}\leq e$, we deduce that $f^{\ell+1}_i=e_i$ and all $j_{\ell+1}$ 
    events $g$ with $f^{\ell}<g\leq f^{\ell+1}$ satisfy $\Ylt{i}{j_{\ell+1}}$.  Therefore,
    \[
    t,f^{\ell+1}\models\Ylt{i}{j_{\ell+1}}
    \wedge[\Ylt{i}{j_{\ell+1}}\Si_{j_{\ell+1}}(\Y_i^{\ell}\alpha)] \,.
    \]
    We obtain $t,f^{\ell+1}\models\Y_i^{\ell+1}\alpha$ which concludes this proof.
  \end{proof}

  \begin{lem}\label{lem:yibyconstant}
$\loctlv$ is at least as expressive as $\loctl$.
  \end{lem}
  \begin{proof} 
  From Lemmas~\ref{lem:yim1} and ~\ref{lem:yim2} we see that $\Y_i\alpha$ is 
  equivalent to the $\loctlv$ formula $\bigvee_{m\leq|\pset|}\Y_i^{m}\alpha$.
  This concludes the proof.
\end{proof}

We now establish the expressive completeness of our past-oriented local temporal logics
$\loctl$, $\loctlv$ and $\loctlf$.  This crucially depends on the
expressive completeness of a process-based pure future local temporal logic proved by
Diekert and Gastin in~\cite{DBLP:journals/iandc/DiekertG06}. It is unknown whether the
fragment $\sprtl$ is as expressive as $\loctl$ in general, see
Theorem~\ref{sprtl-complete} for a partial result.

\begin{thm}\label{thm:loctlfo}
  Let $\alphloc$ be a distributed alphabet over $\pset$. Over $\traces$, first-order 
  logic, $\loctlf$, $\loctl$ and $\loctlv$ have the same expressive power, i.e., a 
  language $L\subseteq\traces$ is definable by a first-order sentence if and only if it is definable 
  by a trace formula in $\loctl$ or in $\loctlv$.
\end{thm}

\begin{proof}
  First, from the semantics of $\loctlf$, we clearly see that each language definable in 
  $\loctlf$ is also first-order definable.
  
  Conversely, we first show that $\loctl$ is expressively complete.
  In~\cite{DBLP:journals/iandc/DiekertG06}, Diekert and Gastin give a process-based pure
  future local temporal logic which they show is expressively equivalent to first order
  logic over traces with a unique minimal event.  The event formulas
  of its past dual have the following syntax and semantics.
  \begin{align*}
    \text{Syntax:}\quad& 
    \alpha = \top \mid a \mid \neg\alpha \mid \alpha\vee\alpha \mid \Y_i{\alpha} 
    \mid  \alpha \SSi_i \alpha \qquad(a \in \alphabet, i \in \pset) 
    \\
    \text{Semantics:}\quad&t,e \models \alpha \SSi_i \alpha'\text{ if there exists } f \in E_i
    \text{ such that } f \leq e  \text{ and } t,f \models \alpha'
    \\
    & \quad \text{and, for all } g \in E_i, \enspace f < g \leq e \text{ implies } t,g \models \alpha.
  \end{align*}
  We show that $\loctl$ is expressive enough to define the $\SSi_i$ operator.  Consider
  the $\loctl$ event formulas $\gamma = \alpha' \vee (\alpha \wedge \alpha \Si_i \alpha')$ and
  $\underline{i} = \bigvee_{a \in \Sigma_i} a$.  Then $\alpha \SSi_i \alpha'$ is
  equivalent to the $\loctl$ event formula $(\underline{i} \wedge \gamma) \vee (\neg
  \underline{i} \wedge \Y_i \gamma)$. This shows $\loctl$ is expressively complete over
  \emph{prime traces} (traces with a unique maximal event). More precisely, for each 
  first order sentence $\varphi$, there is an event formula $\overline{\varphi}$ in 
  $\loctl$ such that for all prime traces $t$ we have $t\models\varphi$ iff 
  $t,\max(t)\models\overline{\varphi}$.
  
  Adsul and
  Sohoni~\cite{normalforms, adsulthesis} showed that any first order sentence over traces can be
  equivalently expressed in a \emph{first normal form} which is a boolean combination of
  first order sentences evaluated only in the process views of traces. More precisely, for
  a trace $t = (E, \leq, \lambda)$, let $t_i$ denote the trace induced by the restriction
  to $\dcset E_i$. Given any first order sentence $\varphi$, by \cite{normalforms} there exists a natural number
  $n$ and sentences $\varphi_{i,m}$ ($i \in \pset$, $1 \leq m \leq n$) such that $t
  \models \varphi$ iff for some $m$, for each $i\in\pset$, we have $t_i \models
  \varphi_{i,m}$. Consider the equivalent event formulas $\overline{\varphi_{i,m}}$ in 
  $\loctl$. Notice that each $t_{i}$ is a prime trace or is the empty trace $\varepsilon$.
  If $\varepsilon\not\models\varphi_{i,m}$ then we let 
  $\beta_{i,m}=\exists_{i}\overline{\varphi_{i,m}}$, otherwise we let
  $\beta_{i,m}=\exists_{i}\overline{\varphi_{i,m}}\vee\neg\exists_{i}\top$.
  We deduce that $t_{i}\models\varphi_{i,m}$ iff $t\models\beta_{i,m}$.
  The $\loctl$ sentence $\bigvee_{m}\bigwedge_{i}\beta_{i,m}$ is equivalent to $\varphi$.
  Therefore, $\loctl$ is expressively complete. 

  By Lemma~\ref{lem:yibyconstant}, $\loctlv$ is also expressively complete and the
  proof is complete.
\end{proof}  
\subsection{Cascade decomposition for $\sprtl$}\label{subsec:sprtl}
We now relate $\sprtl$ with local cascade products of localized reset automata $U_2[p]$.  

\begin{thm}\label{thm:lcascade}
  A trace language is defined by a $\sprtl$ formula if and only if it is accepted 
  by a local cascade product of asynchronous reset automata of the form $U_2[p]$.
\end{thm}

\begin{proof}
First consider a local cascade product $A=U_2[p]\circ_\ell B$ ($p\in\pset$) and suppose
that the languages accepted by $B$ are $\sprtl$-definable.  Let $\{\ls_i\}$ be the state
sets of $U_2[p]$ and let $\chi$ be the asynchronous transducer associated with $U_2[p]$
and its initial state, say $1$.  By the local cascade product principle
(Theorem~\ref{thm:lcpp}), any language accepted by $A$ is a union of languages of the form
$L_1 \cap \chi^{-1}(L_2)$ where $L_1 \subseteq \traces$ is accepted by $U_2[p]$ and $L_2
\subseteq \traces[\s\loctimes\gs]$ is accepted by $B$.

The languages accepted by $U_2[p]$ are
defined by the $\sprtl$-formulas
\begin{align*}
  &\text{With global accepting state 2:}&
  &\exists_p(R_2\vee(\neg R_1 \wedge ((\neg R_1)\Si_p R_2))) \,; 
  \\
  &\text{With global accepting state 1:}&
  \neg&\exists_p(R_2\vee(\neg R_1 \wedge ((\neg R_1)\Si_p R_2))) \,.
\end{align*}

To conclude that the languages accepted by $A$ are $\sprtl$-definable, we only need to
show that if $L_2$ is $\sprtl$-definable over alphabet $\s\loctimes\gs$, then
$\chi^{-1}(L_2)$ is $\sprtl$-definable over $\alphloc$.  This is done by structural
induction on $\sprtl$-formulas over $\s\loctimes\gs$.  For an event formula $\alpha$ of
$\sprtl$ over $\s\loctimes\gs$, we provide an event formula $\hat{\alpha}$ over $\alphloc$
such that for any trace $t \in \traces$, and any event $e$ in $t$, we have $t,e \models
\hat{\alpha}$ if and only if $\chi(t),e \models \alpha$.  The non-trivial case here is the
base case of letter formula $\alpha = (a,s_a)$.  If $p \notin \loc(a)$, then we let
$\hat{\alpha} = a$.  If instead $p\in\loc(a)$, we let
\begin{align*}
  \hat{\alpha} &= a \wedge (\neg R_1)\Si_p R_2 &&\text{if } {[s_a]}_p = 2 \\
  \hat{\alpha} &= a \wedge \neg( (\neg R_1)\Si_p R_2 ) &&\text{if } {[s_a]}_p = 1
\end{align*}

We now establish the converse implication.  For any $\sprtl$ event formula $\alpha$, 
We construct an asynchronous automaton $A_\alpha$, which is a local cascade product of
copies of $U_2[p]$, 
and which is such that for any trace $t$ and event $e$ of $t$, 
\emph{each} local state $[A_\alpha(\dcset{e})]_i$ ($i \in \loc(e)$) completely determines
whether $t,e\models\alpha$.
$A_\alpha$ is constructed by structural induction on the $\sprtl$ event formula $\alpha$.

\subparagraph*{Base case:} Suppose that $\alpha = a \in \alphabet$.  We let $A_\alpha =
(\{S_i\}, \{\lt_a\}, \sinit)$ where $S_i = \{\bot\}$ for all $i \notin \loc(a)$, and $S_i
= \{\top, \bot\}$ for all $i \in \loc(a)$.  For any $P$-state $s$, if for all $i \in P$ we
have $s_i = \bot$ (resp.\ $\top$), then we write $s = \bot$ (resp.\ $s = \top$).  We let
$\sinit = \bot$.  The local transition $\delta_a$ is a reset to $\top$ and the transitions
$\delta_b$ ($b\ne a$) are resets to $\bot$.  This construction ensures that, for all
$i\in\loc(a)$ we have ${[A_\alpha(\dcset e)]}_i = \top$ if and only if $t,e \models
\alpha$.  It is also easy to see that $A_\alpha$ is a local cascade product of $U_2[p]$
for $p\in\loc(a)$.

\subparagraph*{Inductive case:} The non-trivial case is $\alpha = \beta \Si_j \gamma$.  By
inductive hypothesis, we have constructed automata $A_\beta$ and $A_\gamma$ as local 
cascade products of copies of $U_2[p]$.  
Let $A=(\{S_i\}, \{\lt_a\}, \sinit)=U_2[j] \lc \hat{A}_\beta \lc \hat{A}_\gamma$ where the
first $U_2[j]$ with initial state $\bot$ is such that all letters from $\Sigma_j$ reset
the state to $\top$; and $\hat{A}_\beta$ (resp.\ $\hat{A}_\gamma$) simply lifts $A_\beta$
(resp.\ $A_\gamma$) to appropriate input alphabet by ignoring the local state information
provided in the local cascade product.
Hence, $A$ is a local cascade product of copies of $U_2[p]$ which simultaneously provides
the truth values of $\beta$ and $\gamma$ at any event and remembers whether some 
$j$-event already occured. Let $\chi$ be the associated local asynchronous transducer.  

We construct $B = (\{Q_i\}, \{\lt_{(a,s_a)}\}, \qinit)$ over $\s\loctimes\gs$ such that
$A\circ_\ell B$ is the required asynchronous automaton.  Let $Q_i = \{\top, \bot\}$ for
all $i \in \pset$.  Again, we denote a $P$-state $q$ as $\bot$ (resp.\ $\top$) if $q_i =
\bot$ (resp.\ $q_i = \top$) for all $i \in P$.  We let the initial state be $\qinit =
\bot$.  For any $a \notin \alphabet_j$, we let the local transition $\lt_{(a,s_a)}$ be the
reset to $\bot$.  By assumption, if a $j$-event $e$ of the trace $\chi(t)$ is labelled
$(a,s_a)$, then ${[s_a]}_j$ determines whether some $j$-event occured in the past of $e$,
and if this is the case, the truth values of $\beta$ and $\gamma$ at the previous
$j$-event $e_j$.
When $\xi$ is a boolean combination of $\beta,\gamma$, then we write
$[s_a]_j\vdash\Y_j\xi$ if according to the $j$ state of $s_a$, there is a previous
$j$-event $e_j$ and $\xi$ is true at $e_j$.
Then the transition for $a \in \alphabet_j$ is given by
\begin{align*}
  \lt_{(a,s_a)} &= \text{reset to } \bot &&\text{ if } [s_a]_j\not\vdash \Y_j\top 
  \text{ or } [s_a]_j\vdash \Y_j(\neg\beta\wedge\neg\gamma) \\
  \lt_{(a,s_a)} &= \text{reset to } \top &&\text{ if } {[s_a]}_j \vdash \Y_j\gamma \\
  \lt_{(a,s_a)}(q_a) &= \top &&\text{ if } [s_a]_j\vdash\Y_j(\beta\wedge\neg\gamma)
  \text{ and } {[q_a]_j} = \top \\
  \lt_{(a,s_a)}(q_a) &= \bot &&\text{ if } [s_a]_j\vdash\Y_j(\beta\wedge\neg\gamma)
  \text{ and } {[q_a]_j} = \bot 
\end{align*}
  The transitions make sense if we recall the identity $\beta \Si_j \gamma \equiv
  \bot \Si_j (\gamma \vee (\beta \wedge (\beta \Si_j \gamma)))$.
  Note that in the last two cases above, $\lt_{(a,s_a)}$ is the identity transformation on
  process $j$ states.  Hence, process $j$ update is realised by some $U_2[j]$.  The other
  processes of $\loc(a)$ can update their states mimicking process $j$ state update, once
  they also have the truth value of $\alpha=\beta\Si_j\gamma$ at the previous $j$-event 
  $e_j$, which is being made available at event $e$ by the above $U_2[j]$'s state.
  In view of this, it is easy to verify that $B$ is a local cascade product of $U_2[j]$
  followed by $U_2[p]$ for $p \neq j$.
  
  To conclude the proof, we have to handle trace formulas of $\sprtl$, which are the
  sentences defining trace languages.  So consider a trace formula $\exists_j\alpha$ where
  $\alpha$ is an event formula in $\sprtl$.  Let $\aaa_\alpha$ be the local cascade
  product of copies of $U_2[p]$ constructed above.  As in the inductive case above, we
  also use a copy of $U_2[j]$ which remembers whether some $j$-event already occurred in
  the past.  Let $t=(E,\leq,\lambda)\in\traces$ and let $s=(U_2[j]\lc\aaa_\alpha)(t)$.
  The local state $s_j$ allows to determine whether $E_j\neq\emptyset$ thanks to
  $U_2[j]$, and in this case whether the maximal event of $E_j$ satisfies $\alpha$ thanks
  to $\aaa_\alpha$.
\end{proof}

We are now ready to give new characterizations of first-order definable trace languages
over aKR distributed alphabets.

\begin{thm}\label{sprtl-complete} 
  Let $\alphloc$ be a distributed alphabet and $L \subseteq \traces$ be a trace
  language.  If $\alphloc$ is aKR
  (\emph{e.g.}, if $\alphloc$ is an acyclic architecture), then the following statements
  are equivalent.
  \begin{enumerate}
    \item $L$ is definable in first-order logic.
    
    \item $L$ is accepted by a counter-free diamond-automaton 
    (or, recognized by an aperiodic monoid).
    
    \item $L$ is accepted by a local cascade product of copies of $U_2[p]$ 
    (or, recognized by an asynchronous wreath product of atms of the form $U_2[p]$).
    
    \item $L$ is accepted by a counter-free asynchronous automaton
    (or, recognized by an aperiodic asynchronous transformation monoid). 
    
    \item $L$ is definable in  $\sprtl$.
  \end{enumerate}
\end{thm}

\begin{proof}
  By \cite{ebinger1996logical}, first-order definability coincides with recognizability by
  an aperiodic monoid.  By Theorem~\ref{thm:lcascade}, $\sprtl$-definability coincides
  with acceptability by local cascade product of asynchronous reset automata of the form
  $U_2[p]$, or equivalently asynchronous wreath product of atm's of the form $U_2[p]$.

  If $\alphloc$ is aKR, recognizability by an aperiodic monoid implies recognizability
  by an asynchronous wreath product of asynchronous transformation monoids of the form
  $U_2[p]$.  The converse implication follows from the easy verification that a wreath
  product of asynchronous transformation monoids of the form $U_2[p]$ is an aperiodic atm
  (that is, the associated global tm is aperiodic).  Putting these equivalences together
  and keeping in mind the correspondences between (asynchronous) automata and
  (asynchronous) morphisms into (asynchronous) transformation monoids, we get the desired
  result.
\end{proof}

\subsection{Cascade decomposition for $\loctlv$}\label{subsec:loctlv}
We turn now to the logic $\loctlv$ and its relation with local cascade products.  Here, we
seek a decomposition result which is valid for all distributed alphabets, and not only for
those that are known to be aKR. Unfortunately, we do not know whether all aperiodic trace
languages can be accepted by local cascade products of copies of $U_2[p]$.  It is
interesting to notice first that the argument in the proof of Theorem~\ref{thm:lcascade}
cannot be lifted to the logic $\loctlv$.  More precisely, when $\alpha=\Yleq{i}{j}$, the
automaton $A_\alpha$ specified in this proof cannot, in general, be obtained as a local
cascade product of copies of $U_2[p]$.  
To explain this formally, we use the notion of $\lab$-labelling functions computed by
asynchronous automata, defined in Section~\ref{sec: asynchronous automata}.

Given an event formula $\alpha\in\loctlf$, 
we let $\theta_\alpha$ be the $\{0,1\}$-labelling function which decorates each event of a
trace $t=(E,\leq,\lambda)\in\traces$ with the truth value of $\alpha$, i.e.,
$\theta_\alpha(t)=(E,\leq,(\lambda,\mu))$ where for all $e\in E$, we have
$\mu_\alpha(e)=1$ if $t,e\models\alpha$ and $\mu_\alpha(e)=0$ otherwise.

In the proof of Theorem~\ref{thm:lcascade}, we have constructed for each formula
$\alpha\in\sprtl$ an asynchronous automaton $\aaa_\alpha$ which computes $\theta_\alpha$,
and which is a local cascade product of copies of $U_2[p]$.  Lemma~\ref{lem:theta Yij not aperiodic} below shows that
this cannot be extended to $\loctlv$.  This is a slight modification of an example from
Adsul and Sohoni \cite{DBLP:conf/fsttcs/AdsulS04}.

\begin{lem}\label{lem:theta Yij not aperiodic}
  Let $\pset=\{1,2,3\}$, $\alphabet=\{a,b,c\}$ with the distribution $\Sigma_1=\{a,c\}$,
  $\Sigma_2=\{a,b\}$, $\Sigma_3=\{b,c\}$ and let $\alpha=\Yleq{1}{3}$.  As no two letters
  of $\Sigma$ are independent, $\traces$ is isomorphic to the free monoid $\Sigma^*$.
  
  Nevertheless, there is no aperiodic asynchronous automaton over $\alphloc$ which
  computes $\theta_\alpha$.  In particular, $\theta_\alpha$ cannot be computed by a local
  cascade product of copies of $U_2[p]$.
\end{lem}

\begin{proof}
  Suppose, for the sake of contradiction, that there exists an aperiodic asynchronous automaton $\aaa=(\{\ls_i\},\{\lt_a\},\sinit)$ and a
  $\{0,1\}$-transducer $\hat{\aaa}=(\{\ls_i\},\{\lt_a\},\sinit,\{\mu_a\})$ which computes
  $\theta_\alpha$.  The output function $\mu_c\colon\ls_c\to\{0,1\}$ computes the truth
  value of $\alpha=\Yleq{1}{3}$ at $c$-events.
  
  As $\aaa$ is aperiodic, there exists $n$ such that, starting at the initial global state
  $\sinit$, traces $(ab)^n$ and $(ab)^{n+1}$ reach the same global state, say
  $s=(s_1,s_2,s_3)=\aaa((ab)^n)=\aaa((ab)^{n+1})$.  It follows that the trace $ab$ fixes
  $s$.  As the transition function of $a$ (resp.\ $b$) does not change the local state of
  process $3$ (resp.\ process $1$), it must be that the $a$-transition at $s$ leads to a
  global state of the form $s'=(s_1,s'_2,s_3)$.  In particular, $s'$ is the global state
  reached on input $(ab)^na$.  Now, consider the traces $t=(ab)^nc$ and $t'=(ab)^nac$.
  The $c$-event in $t$ satisfies $\alpha=\Yleq{1}{3}$ whereas the $c$-event in $t'$ does
  not. Since $s_c=(s_1,s_3)=s'_c$, this contradicts the fact that $\mu_c$ computes the
  truth value of $\alpha$ at $c$-events.
\end{proof}

On the other hand, the \emph{gossip automaton}, one of the most important tools in the
theory of asynchronous automata, due to Mukund and Sohoni
\cite{DBLP:journals/dc/MukundS97}, computes all the constants of $\loctlv$.  Let
$\lab=\{0,1\}^{\pset\times\pset}$ and $\thetaY$ be the $\lab$-labelling function
which decorates each event $e$ of a trace $t$ with the truth values of all constants
$\Yleq{i}{j}$, i.e., for all $i,j\in\pset$, $\muY_{i,j}(e)=1$ if
$t,e\models\Yleq{i}{j}$ and $\muY_{i,j}(e)=0$ otherwise.
Since the events referred to by $\{\Y_i\}$ are called primary events, 
we call $\thetaY$ the \emph{primary order labelling function}.

\begin{thmC}[(Gossip Automaton~\cite{DBLP:journals/dc/MukundS97})]\label{thm:gossip for loctl}
	There exists an asynchronous automaton $\gossip = (\{\Upsilon_i\},\{\nabla_a\},v_{\textrm{in}})$, called the \emph{gossip automaton},
  which computes $\thetaY$.
\end{thmC}

Note that, as a consequence of Lemma~\ref{lem:theta Yij not aperiodic}, the gossip
automaton is not aperiodic in general, a fact that was already established in
\cite{DBLP:conf/fsttcs/AdsulS04}.

In view of Theorems~\ref{thm:lcascade} and~\ref{thm:gossip for loctl}, the following lemma
will help relate $\loctlv$ languages and local cascade products on arbitrary distributed
alphabets.

\begin{lem}\label{lem:loctlv and sprtl}
  Let $\alphloc$ be a distributed alphabet.
  \begin{enumerate}
    \item  If $L\subseteq\traces$ is $\loctlv$-definable over $\alphabet$ then
    $L'=\thetaY(L)$ is $\sprtl$-definable over $\lalphabet$.
    Notice that $L=\big(\thetaY\big)^{-1}(L')$.
  
    \item  If $L'\subseteq \traces[\lalphabet]$ is $\loctlv$-definable over
	    $\lalphabet$ then
    $L=\big(\thetaY\big)^{-1}(L')$ is $\loctlv$-definable over $\alphabet$.
  \end{enumerate}
\end{lem}

\begin{proof}
  We simply write $\theta$ for the primary order labelling function $\thetaY$.
  
  \noindent
  (1) For each event formula $\alpha\in\loctlv$ over $\alphabet$, we construct a
  formula $\widetilde{\alpha}\in\sprtl$ over $\lalphabet$ such that, for all 
  traces $t\in\traces$ and all events $e$ in $t$, we have $t,e\models\alpha$ if and only if 
  $\theta(t),e\models\widetilde{\alpha}$. The construction is by structural induction on 
  $\alpha$. The constants are the interesting cases. For $a\in\Sigma$, we define
  \[
  \tilde{a} = \bigvee_{\gamma\in\Gamma} (a,\gamma) \,.
  \]
  Now, for $i,j\in\pset$, we let $\Gamma_{i,j}=\{\gamma\in\Gamma\mid\gamma_{i,j}=1\}$ and
  we define
  \[
  \widetilde{\Yleq{i}{j}} = \bigvee_{a\in\alphabet,\gamma\in\Gamma_{i,j}} (a,\gamma) \,.
  \]
  The inductive cases are trivial, for instance
  $\widetilde{\alpha\Si_i\beta}=\widetilde{\alpha}\Si_i\widetilde{\beta}$.
  We deduce that, if $L$ is $\loctlv$-definable over $\alphabet$ then
  $\theta(L)$ is $\sprtl$-definable over $\lalphabet$.
  
  \medskip\noindent (2) Given an event formula $\alpha\in\loctlv$, we construct an event
  formula $\widehat{\alpha}\in\loctlv$ such that, for all traces $t\in\traces$ and all
  events $e$ in $t$, we have $\theta(t),e\models\alpha$ if and only if $t,e\models\widehat{\alpha}$.
  Again, the construction is by structural induction and the interesting cases are the
  constants.  For $(a,\gamma)\in\Sigma\times\lab$, we define
  \[
  \widehat{(a,\gamma)} = a \enspace
  \wedge \bigwedge_{i,j\in\pset\mid\gamma_{i,j}=1} \Yleq{i}{j} \enspace
  \wedge \bigwedge_{i,j\in\pset\mid\gamma_{i,j}=0} \neg\Yleq{i}{j} \,.
  \]
  The other cases are trivial, e.g., $\widehat{\Yleq{i}{j}}=\Yleq{i}{j}$ and
  $\widehat{\alpha\Si_i\beta}=\widehat{\alpha}\Si_i\widehat{\beta}$. 
\end{proof}

Since the gossip automaton $\gossip$ computes $\thetaY$ and all $\sprtl$-definable
languages over $\dlalphabet$ can be accepted by a local cascade product of copies of
asynchronous reset automata of the form $U_2[p]$ (Theorem~\ref{thm:lcascade}), we deduce
from Lemma~\ref{lem:loctlv and sprtl}~(1) and Proposition~\ref{prop: inverse labelling}
that all $\loctlv$-definable languages over $\dalphabet$ can be accepted by a local
cascade product of the gossip automaton $\gossip$ followed by copies of asynchronous reset
automata.

Now, as we saw, the gossip automaton exhibits a
non-aperiodic behaviour in general. In order to get a converse of the above
statement, we introduce a restricted version of the local cascade product.

Let $\aaa$ be an asynchronous automaton over $\dlalphabet$.  The
\emph{$\thetaY$-restricted} local cascade product $\gossip\lcr\aaa$ is an
asynchronous automaton which runs $\gossip$ on an input trace $t\in\traces$, and runs
$\aaa$ over $\thetaY(t)$.  Formally, if
$\widehat{\gossip}=(\{\Upsilon_i\},\{\nabla_a\},v_{\textrm{in}},\{\mu_a\})$ computes
$\thetaY$ and $\aaa=(\{\ls_i\},\{\lt_{(a,\gamma)}\},\sinit)$ then $\gossip\lcr\aaa$ is the asynchronous automaton
$\gossip\lcr\aaa = (\{R_i\},\{\Delta_a\},(v_{\textrm{in}},\sinit))$ over $\dalphabet$, where
$R_i=\Upsilon_i\times\ls_i$ for $i\in\pset$ and, for $a\in\Sigma$ and $(\upsilon_a,s_a)\in
R_a$, $\Delta_a((\upsilon_a,s_a))=(\nabla_a(\upsilon_a),\lt_{(a,\mu_a(\upsilon_a))}(s_a))$.
Notice that, in the definition of the transition relation of this \emph{restricted}
cascade product, the $a$-state $\upsilon_a$ of $\gossip$ has been abstracted to
$\mu_a(\upsilon_a)\in\lab$.
Notice that languages accepted by $\gossip\lcr\aaa$ are also accepted by $\gossip\lc\aaa'$
where $\aaa'$ is the extension of $\aaa$ to the suitable alphabet $\tralphabet$.

\begin{lem}\label{lem:lcr-principle}
  A trace language $L\subseteq\traces$ is accepted by $\gossip\lcr\aaa$ (with all global
  states of the gossip automaton accepting) if and only if $L=\big(\thetaY\big)^{-1}(L')$,
  where $L'$ is a trace language over $\dlalphabet$ accepted by $\aaa$.
\end{lem}

\begin{proof}
  The first component of $\gossip\lcr\aaa$ runs $\gossip$ on the input trace $t\in\traces$
  and accepts since all global states of $\gossip$ are accepting.  Now, the second
  components runs $\aaa$ on $\thetaY(t)$.  Therefore, $t$ is accepted by
  $\gossip\lcr\aaa$ if and only if $\thetaY(t)$ is accepted by $\aaa$.
\end{proof}

\begin{cor}\label{cor:lcascade-loctlv}
  Let $\dalphabet$ be a distributed alphabet and let $L\subseteq\traces$ be a trace
  language.  Then $L$ is $\loctlv$-definable if and only if $L$ is accepted by a
  \emph{$\thetaY$-restricted} local cascade product $\gossip\lcr\aaa$, where
  $\gossip$ is the gossip automaton and $\aaa$ is a local cascade product of
  asynchronous reset automata of the form $U_2[p]$.
\end{cor}

\begin{proof}
  If $L$ is $\loctlv$-definable then $L'=\thetaY(L)$ is $\sprtl$-definable by
  Lemma~\ref{lem:loctlv and sprtl}~(1).  By Theorem~\ref{thm:lcascade}, $L'$ is accepted
  by an asynchronous automaton $\aaa$ which is a local cascade product of
  asynchronous reset automata of the form $U_2[p]$.  Since
  $L=\big(\thetaY\big)^{-1}(L')$, the left-to-right implication follows from Lemma~\ref{lem:lcr-principle}.
  
  Conversely, assume that $L$ is accepted by $\gossip\lcr\aaa$ where $\aaa$ is a local
  cascade product of asynchronous reset automata of the form $U_2[p]$.  By
  Lemma~\ref{lem:lcr-principle}, we have $L=\big(\thetaY\big)^{-1}(L')$ where
  $L'$ is a trace language over $\dlalphabet$ accepted by $\aaa$.  By
  Theorem~\ref{thm:lcascade}, $L'$ is $\sprtl$-definable over $\dlalphabet$.  Finally,
  Lemma~\ref{lem:loctlv and sprtl}~(2) implies that $L$ is $\loctlv$-definable.
\end{proof}

\begin{rem}\label{rk: gossip-like}
  In Lemma~\ref{lem:lcr-principle} and Corollary~\ref{cor:lcascade-loctlv}, as well
  as in Theorem~\ref{thm:lcascade-loctlv} below, we may replace the gossip automaton
  $\gossip$ by any asynchronous automaton computing the primary order labelling function
  $\thetaY$.
\end{rem}

We close this section with the following theorem, summarizing our characterization of 
first-order definable trace languages using local cascade products.

\begin{thm}\label{thm:lcascade-loctlv}
  Let $\dalphabet$ be a distributed alphabet and let $L\subseteq\traces$ be a trace 
  language. The following are equivalent:
  \begin{enumerate}
    \item $L$ is recognized by an aperiodic monoid.
    
    \item $L$ is star-free.
    
    \item $L$ is definable in first-order logic.
    
    \item $L$ is definable in  $\loctl$ or in $\loctlv$ or in $\loctlf$.
    
    \item $L$ is accepted by a \emph{$\thetaY$-restricted} local cascade
    product $\gossip\lcr\aaa$ where $\gossip$ is the gossip automaton and $\aaa$ is a local
    cascade product of copies of asynchronous reset automata of the form $U_2[p]$.

    \item $L$ is accepted by a \emph{$\thetaY$-restricted} local cascade
    product $\gossip\lcr\aaa$ where $\gossip$ is the gossip automaton and $\aaa$ is an
    aperiodic asynchronous automaton.
  \end{enumerate}
\end{thm}

\begin{proof}
  As mentioned before, Guaiana, Restivo and Salemi \cite{guaiana1992star} established the
  equivalence of (1) and (2), and Ebinger and Muscholl \cite{ebinger1996logical} proved
  that these are equivalent to (3). The equivalence between (3) and (4) is 
  Theorem~\ref{thm:loctlfo} and Corollary~\ref{cor:lcascade-loctlv} gives the equivalence 
  with (5). 
  
  The equivalence with (6) was first proved by Adsul and Sohoni
  \cite{DBLP:conf/fsttcs/AdsulS04}.  We can obtain it as follows.  First, $(5)$ implies
  $(6)$ since a local cascade product of aperiodic asynchronous automata is again
  aperiodic.  Next, we show that (6) implies (4) as in the proof of
  Corollary~\ref{cor:lcascade-loctlv}.
  Assume that $L$ is accepted by $\gossip\lcr\aaa$ where $\aaa$ is an aperiodic
  asynchronous automaton.  By Lemma~\ref{lem:lcr-principle}, we have
  $L=\big(\thetaY\big)^{-1}(L')$ where $L'$ is a trace language over
  $\dlalphabet$ accepted by $\aaa$.  Since (1) implies (4), the language $L'$ accepted by
  $\aaa$ is $\loctlv$-definable over $\dlalphabet$.  Finally, Lemma~\ref{lem:loctlv and
  sprtl}~(2) implies that $L$ is $\loctlv$-definable over $\dalphabet$.
\end{proof}

\section{Global cascade sequences, the related principle and temporal logics}\label{sec:cascade}
We have established a direct correspondence between asynchronous wreath products of
asynchronous transformation monoids and local cascade products of asynchronous automata
(Proposition~\ref{prop: local cascade} and Theorem~\ref{thm:lcpp}).  In particular, over
aKR distributed alphabets, any asynchronous automaton is simulated by a local cascade
product of our proposed localised two-state reset and localised permutation asynchronous
automata.  This yields the same benefits as in the theory of word languages (see
Section~\ref{sec:aperiodic}, in particular Theorem~\ref{sprtl-complete}, for a concrete
example).  However, we do not know which distributed alphabets with non-acyclic
architecture are aKR. Despite this, we have been able to get decomposition results (see
Theorem~\ref{thm:lcascade-loctlv}) for first order definable trace languages over general
architectures.  This has been possible thanks to the expressive completeness of $\loctlv$
(Theorem~\ref{thm:loctlfo}) and the primary order labelling function $\thetaY$ which
allows to reason about $\loctlv$-definability over $\dalphabet$ in terms of
$\sprtl$-definability over $\dlalphabet$ (see Lemma~\ref{lem:loctlv and sprtl}).  We have
also ruled out, in Lemma~\ref{lem:theta Yij not aperiodic}, the possibility of an
aperiodic asynchronous automaton that computes $\thetaY$.  This finally leads to a
restricted version of the local cascade product characterization in
Theorem~\ref{thm:lcascade-loctlv}, which circumvents the non-aperiodic behaviour of the
gossip automaton which computes $\thetaY$.

In this section, we propose \emph{global cascade sequences} as a new model for accepting
trace languages.  This model is built using asynchronous automata and lets us pose
automata-theoretic and language-theoretic decomposition questions in the same spirit as
with the local cascade product.  Its definition and acceptance condition are inspired by
the operational point of view of the local cascade product, and are natural from an
automata-theoretic viewpoint.  Further, it supports a global cascade principle in the same
vein as the local cascade principle supported by local cascade products.

Later in the section, we show that global cascade sequences of localised two-state reset
asynchronous automata accept exactly the first order definable trace languages.  In fact,
we establish that $\loctl$-definability matches acceptability by a global cascade sequence
of copies of $U_2[p]$ and use the expressive completeness of $\loctl$ from
Theorem~\ref{thm:loctlfo}.  This allows a characterization of first order logic purely in
terms of an asynchronous cascade of copies of $U_2[p]$ albeit using {\em global cascade
sequences}.  The new characterization remains in the realm of aperiodic asynchronous
devices and can be considered {\em intrinsic} in the spirit of the `first-order =
aperiodic' slogan.  Of course, all this comes at a cost!  What we lose, bringing in global
cascade sequences instead of local cascade products in decomposition results, is the nice
correspondence to an algebraic operation.  In Section~\ref{sec: implementation global
cascade}, we show how to construct an asynchronous automaton that realizes a global
cascade sequence, making use, again, of the gossip automaton.

Let $\aaa=(\{\ls_i\},\{\lt_a\},\sinit)$ be an asynchronous automaton over $\dalphabet$ and
$\chi_\aaa$ be the asynchronous transducer computed by $\aaa$.  Recall that
$\chi_\aaa(t)$ preserves the underlying poset of events of $t$ and, at each event, records
the previous local states of the processes \emph{participating} in that event.
We now introduce a natural variant of $\chi_\aaa$ called the \emph{global-state labelling
function}, where we record at each event $e$ the \emph{best global state} that causally
precedes $e$.  This is the best global state that the processes participating in the
current event are (collectively) aware of.

To be more precise, we first set additional notation.
Recall that the alphabet $\alphabet \times \gs$ can be
equipped with a distributed structure (over $\pset$) by letting $\loc((a, s)) = \loc(a)$,
that is, $(\gtralphabet)_{i} = \alphabet_i\times \gs$.
We refer to $\dalphabet$ as the \emph{input alphabet} of $\aaa$ and
$\dgtralphabet$ as its \emph{output alphabet}.

\begin{defi}[Global-state Labelling Function]\label{def:grf}
	Let $\aaa=(\{\ls_i\},\{\lt_a\},\sinit)$ be an asynchronous automaton over $\dalphabet$.
	The global-state labelling function of $\aaa$ is the $\gs$-labelling function
	$\gslf_\aaa \colon \traces \to \tracesdgtr$ defined for $t = (E, \leq, \lambda) \in
	\traces$ by $\gslf_\aaa(t) = (E, \leq, (\lambda, \mu)) \in \tracesdgtr$ with the
	labelling $\mu \colon E \to \gs$ given by $\mu(e) = s \text{ where } s = A(\Da e)$
\end{defi}

\begin{rem}
  The global-state labelling function and asynchronous transducer associated with an
  automaton coincide in the sequential (single process) setting.
\end{rem}

\begin{rem}\label{rem:forgetfulmorphism} 
  Letting $f((a,s))=(a,s_a)$ for each $(a,s)\in \gtralphabet$ defines a morphism $f\colon \tracesdgtr \to \tracesdtr$.  We then have
  $\chi_\aaa(t) = f(\gslf_\aaa(t))$ for each $t\in\traces$.
\end{rem}

\begin{exa} 
  Figure~\ref{fig:gat-example} shows the image by the global-state labelling function
  $\gslf$, for the same asynchronous automaton $\aaa$ (or asynchronous morphism $\phi$)
  and trace $t$ as in Example~\ref{ex:lat}.
  Note the difference from Figure~\ref{fig:lat-example}.  For example, here the $p_3$-event
  has process $p_1$ state $2$ in its label (which is the best process $p_1$ state in its
  causal past) even though process $p_1$ and process $p_3$ never interact directly.
  \begin{figure}[ht]
	\centering
	\begin{tikzpicture}
\draw (0,0) -- (0, 5*0.2 + 3*0.5);

\draw (0.3,1.4) rectangle (0.8,2.3);
\draw (1.3,0.7) rectangle (1.8,2.3);
\draw (2.3,0) rectangle (2.8,1.6);

\draw (0,0.2 + 0.5/2) -- (2.3, 0.2 +0.5/2); 
\draw (2.8,0.2 + 0.5/2) -- (3.0,0.2 + 0.5/2); 

\draw (0,2*0.2+0.5+0.5/2) -- (1.3,2*0.2+0.5+0.5/2); 
\draw (1.8,2*0.2+0.5+0.5/2) -- (2.3,2*0.2+0.5+0.5/2); 
\draw (2.8,2*0.2+0.5+0.5/2) -- (3.0,2*0.2+0.5+0.5/2); 

\draw (0, 3*0.2 + 2*0.5 + 0.5/2) -- (0.3, 3*0.2 + 2*0.5 + 0.5/2);
\draw (0.8, 3*0.2 + 2*0.5 + 0.5/2) -- (1.3, 3*0.2 + 2*0.5 + 0.5/2);
\draw (1.8, 3*0.2 + 2*0.5 + 0.5/2) -- (3.0, 3*0.2 + 2*0.5 + 0.5/2);

\node at (0.3 + 0.5/2, 1.6 + 0.5/2) {$a$};
\node at (1.3 + 0.5/2, 0.9 + 0.5 + 0.2/2) {$b$};
\node at (2.3 +  0.5/2, 0.2 + 0.5 + 0.2/2) {$c$};

\node at (-0.2, 0.2 + 0.5/2) {$p_3$};
\node at (-0.2, 2*0.2 + 0.5 + 0.5/2) {$p_2$};
\node at (-0.2, 3*0.2 + 2*0.5 + 0.5/2) {$p_1$};

\node at (0.15, 3*0.2 + 2*0.5 +0.5/2 + 0.15) {\tiny $1$};
\node at (1.15, 3*0.2 + 2*0.5 +0.5/2 + 0.15) {\tiny $1$};
\node at (1.95, 3*0.2 + 2*0.5 +0.5/2 + 0.15) {\tiny $2$};

\node at (1.1, 2*0.2 + 0.5 +0.5/2 + 0.15) {\tiny $\bot_2$};
\node at (2.1, 2*0.2 + 0.5 +0.5/2 + 0.15) {\tiny $\bot_2$};
\node at (2.99, 2*0.2 + 0.5 +0.5/2 + 0.15) {\tiny $\bot_2$};

\node at (2.1, 0.2 + 0.5/2 + 0.15) {\tiny $\bot_3$};
\node at (2.99, 0.2 + 0.5/2 + 0.15) {\tiny $\bot_3$};

\node at (1.5, -0.4) {\small Run of $A_\varphi$ on trace $t$};

\begin{scope}[shift={(5,0)}]
\draw (0,0) -- (0, 5*0.2 + 3*0.5);

\draw (0.3,1.4) rectangle (0.8,2.3);
\draw (1.3,0.7) rectangle (1.8,2.3);
\draw (2.3,0) rectangle (2.8,1.6);

\draw (0,0.2 + 0.5/2) -- (2.3, 0.2 +0.5/2); 
\draw (2.8,0.2 + 0.5/2) -- (3.0,0.2 + 0.5/2); 

\draw (0,2*0.2+0.5+0.5/2) -- (1.3,2*0.2+0.5+0.5/2); 
\draw (1.8,2*0.2+0.5+0.5/2) -- (2.3,2*0.2+0.5+0.5/2); 
\draw (2.8,2*0.2+0.5+0.5/2) -- (3.0,2*0.2+0.5+0.5/2); 

\draw (0, 3*0.2 + 2*0.5 + 0.5/2) -- (0.3, 3*0.2 + 2*0.5 + 0.5/2);
\draw (0.8, 3*0.2 + 2*0.5 + 0.5/2) -- (1.3, 3*0.2 + 2*0.5 + 0.5/2);
\draw (1.8, 3*0.2 + 2*0.5 + 0.5/2) -- (3.0, 3*0.2 + 2*0.5 + 0.5/2);

\node at (0.3 + 0.5/2, 1.6 + 0.5/2 + 0.3) {\small $a$}; 
\node at (0.3 + 0.5/2, 1.6 + 0.5/2 + 0.1) {\tiny $1$}; 
\node at (0.3 + 0.5/2, 1.6 + 0.5/2 -0.1) {\tiny $\bot_2$};
\node at (0.3 + 0.5/2, 1.6 + 0.5/2 - 0.3) {\tiny $\bot_3$};

\node at (1.3 + 0.5/2, 1.6 + 0.5/2 +0.2) {\small $b$};
\node at (1.3 + 0.5/2, 0.9 + 0.5 + 0.2/2 +0.2) {\tiny $1$};
\node at (1.3 + 0.5/2, 0.9 + 0.5/2 +0.2) {\tiny $\bot_2$};
\node at (1.3 + 0.5/2, 0.9 + 0.5/2-0.1) {\tiny $\bot_3$};

\node at (2.3 +  0.5/2, 0.2 + 0.5 + 0.2 +0.5/2 +0.2) {\small $c$};
\node at (2.3 +  0.5/2, 0.2 + 0.5 + 0.2/2+0.2) {\tiny $2$};
\node at (2.3 +  0.5/2, 0.2 + 0.5/2 +0.2) {\tiny $\bot_2$};
\node at (2.3 +  0.5/2, 0.2 + 0.5/2 -0.15) {\tiny $\bot_3$};

\node at (-0.2, 0.2 + 0.5/2) {$p_3$};
\node at (-0.2, 2*0.2 + 0.5 + 0.5/2) {$p_2$};
\node at (-0.2, 3*0.2 + 2*0.5 + 0.5/2) {$p_1$};

\node at (1.5, -0.4) {\small Trace $\gslf(t)$};
\end{scope}
\end{tikzpicture}
	\caption{Global-state labelling function output on a trace}%
	\label{fig:gat-example}
\end{figure}
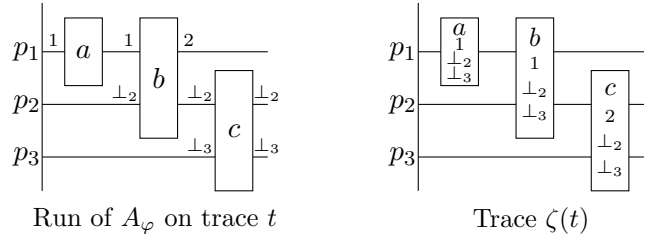
\end{exa}

As $\gslf_\aaa(t)$ carries more information than $\chi_\aaa(t)$
(Remark~\ref{rem:forgetfulmorphism}), one can view $\gslf_\aaa$
as an information-theoretic generalization of $\chi_\aaa$. However, unlike $\chi_\aaa$, 
it is not clear a priori whether it can be computed by an asychronous automaton. 
We will return to this important
issue in Section~\ref{sec: implementation global cascade}. At the moment, we simply extend
the operational point of view of local cascade products (see Figure~\ref{fig:local-cascade})
using global-state labelling functions instead of asynchronous transducers.

\begin{defi}
	A \emph{global cascade sequence} (in short, gcs) $\Aseq$ is a sequence $(\aaa_1, \aaa_2,
	\ldots, \aaa_n)$ of asynchronous automata such that, for $1 \leq i < n$, the input
	alphabet of $\aaa_{i+1}$ is the output alphabet of $\aaa_i$.  The input alphabet of
	$A_1$ is called the input alphabet of $\Aseq$ and the output alphabet of $A_n$ is called
	the output alphabet of $\Aseq$.

  We associate a global-state labelling function $\gslf_{\Aseq}$ from traces over the
  input alphabet of $\Aseq$ to traces over the output alphabet of $\Aseq$, namely the
  composition
  \[
  \gslf_{\Aseq} = \gslf_{A_1}  \gslf_{A_2}  \cdots  \gslf_{A_n}
  \]
  of the global-state labelling functions of the $A_i$.
  For instance, if $\Aseq=(A_1, A_2)$ then $\gslf_{\Aseq}(t) = \gslf_{A_2}(\gslf_{A_1}(t))$.
\end{defi}

It is important to observe that a gcs $\Aseq$ is {\em not} an asychronous automaton.  A
gcs is simply a cascade of a sequence of {\em compatible} automata which are {\em
connected} via global-state labelling mechanisms.  The following lemmas are immediate.

\begin{lem}\label{lem: composition of gcs}
  Let $\Aseq = (\aaa_1, \aaa_2, \ldots, \aaa_n)$ and $\Bseq = (B_1, B_2, \ldots, B_m)$ be
  two global cascade sequences such that the input alphabet of $\Bseq$ is the output
  alphabet of $\Aseq$.  Then $\Cseq = (\aaa_1, \ldots, \aaa_n, B_1, \ldots, B_m)$ is a
  valid global cascade sequence.  Moreover, $\gslf_{\Cseq} = \gslf_{\Aseq} \gslf_{\Bseq}$.
\end{lem}

The concatenation of global cascade sequences, as in Lemma~\ref{lem: composition of gcs},
is denoted by $\Cseq = \Aseq \fc \Bseq$.  This is an associative operation, as verified in
the following lemma.

\begin{lem} 
  Let $\Aseq$, $\Bseq$, $\Cseq$ be global cascade sequences such that the input alphabet
  of $\Bseq$ is the output alphabet of $\Aseq$, and the input alphabet of $\Cseq$ is the
  output alphabet of $\Bseq$.  Then $(\Aseq \fc \Bseq) \fc \Cseq = \Aseq \fc (\Bseq \fc
  \Cseq) \mbox{~and~} \gslf_{(\Aseq \fc \Bseq)} \gslf_{\Cseq} = \gslf_{\Aseq}
  \gslf_{(\Bseq \fc \Cseq)}$.
\end{lem}

We now show that, one can view a gcs as an acceptor of trace languages in a natural way. 
We begin with some notation.

For an automaton $\aaa$, the set of global states of $\aaa$ is denoted by
$\globalstates{\aaa}$.  Given a global cascade sequence $\Aseq=(\aaa_1, \ldots, \aaa_n)$,
we refer to the set $\globalstates{A_1} \times \ldots \times \globalstates{A_n}$ as the
\emph{global states} of $\Aseq$ and denote it as $\globalstates{\Aseq}$.  Similarly the
set of $P$-states of $\Aseq$ is the cartesian product of the sets of $P$-states of its
constituent asynchronous automata.  Given a $P$-state $s = (s_1, \ldots, s_n)$ of $\Aseq$,
and $P' \subseteq P$, we let $s_{P'} = ((s_1)_{P'}, \ldots, (s_n)_{P'})$ be the natural
restriction of $s$ to $P'$.

\begin{defi}[Language accepted by a global cascade sequence]
  Let $\Aseq=(\aaa_1, \ldots, \aaa_n)$ be a global cascade sequence with input alphabet
  $\dalphabet$.  Given a trace $t\in\traces$, we let $\Aseq(t)\in\globalstates{\Aseq}$ be
  the global state of $\Aseq$ reached after reading $t$:
  $$
  \Aseq(t) = \left( \aaa_1(t),\aaa_2(\gslf_{\aaa_1}(t)), \ldots, 
  \aaa_n( [\gslf_{\aaa_1}  \cdots  \gslf_{\aaa_{n{-1}}}](t) ) \right) \,.
  $$ 
  Given $F\subseteq\globalstates{\Aseq}$, we define the language $L(\Aseq, F)$, a language
  accepted by $\Aseq$, by
  $$
  L(\Aseq, F) = \{t \in \traces \mid \Aseq(t) \in F \} \,.
  $$
  A language $L\subseteq\traces$ is said to be {\em accepted by $\Aseq$} if there
  exists a subset $F \subseteq \globalstates{\Aseq}$ such that $L=L(\Aseq, F)$.
  See the Figure~\ref{fig:global-cascade-sequence}.
\end{defi}

\begin{figure}[t]
  \centering
  \begin{tikzpicture}

\draw (2.8,0) rectangle (4,1.2);
\node at (3.4,0.6) {$A_1$};

\draw (4, 0.3) -- (4.3, 0.3);
\draw[->] (4.3, 0.3) -- (4.3,-0.2);
\node at (4.3,-0.35) {\small $s = A_1(t)$};
\draw[-stealth'] (2.3, 0.6) -- (2.8,0.6) node [midway, above] {$t$};

\draw (5.4,0) rectangle (6.6,1.2);
\node at (6,0.6) {$A_2$};

\draw (6.6, 0.3) -- (6.9, 0.3);
\draw[->] (6.9, 0.3) -- (6.9,-0.2);
\node at (6.9,-0.35) {\small $q = A_2(\gslf_{A_1}(t))$};
\draw[-stealth'] (4, 0.6) -- (5.4,0.6) node [midway, above] {\small $\gslf_{A_1}(t)$};

\draw (9,0) rectangle (10.2, 1.2);
\node at (9.6,0.6) {$A_3$};

\draw (10.2, 0.3) -- (10.5, 0.3);
\draw[->] (10.5, 0.3) -- (10.5,-0.2);
\node at (10.5,-0.35) {\small $r = A_3((\gslf_{A_1}\gslf_{A_2})(t))$};
\draw[-stealth'] (6.6, 0.6) -- (9,0.6) node [midway, above] {\small $(\gslf_{A_1}\gslf_{A_2})(t)$};
\end{tikzpicture}%
  \caption{Global cascade product}
  \label{fig:global-cascade-sequence}
\end{figure}
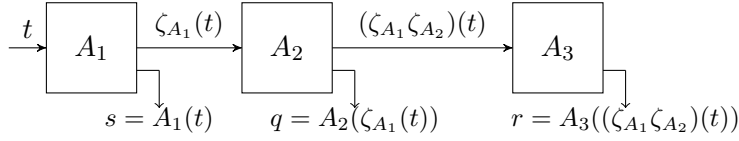

The following \emph{global cascade sequence principle} is an easy consequence of the
definitions.

\begin{thm}\label{thm:gcwpp2}
	Let $\Aseq$ and $\Bseq$ be global cascade sequences, let $\dalphabet$ be the input
	alphabet of $\Aseq$, and suppose that the output alphabet of $\Aseq$ is the input
	alphabet of $\Bseq$, say $\Pialphabet$.  Let $\Cseq = \Aseq \fc \Bseq$.  Then any
	language $L \subseteq \traces$ accepted by $\Cseq$ is a finite union of languages
	of the form $U \cap \gslf_{\Aseq}^{-1}(V)$ where $U \subseteq \traces$ is accepted
	by $\Aseq$, and $V \subseteq \Pitraces$ is accepted by $\Bseq$.
\end{thm}

Building on the simple observation in Remark~\ref{rem:forgetfulmorphism} that
global-state labelling functions are information-theoretic generalizations of local
asynchronous transducers, we now show that a local cascade 
product can be realized by an appropriate global cascade sequence. 

An asynchronous automaton $B$ over $\dtralphabet$ naturally gives rise to another
asynchronous automaton $\widehat{B}$ (with the same state sets, etc.)  operating over
$\dgtralphabet$ by defining the transition of $\widehat{B}$ on letter $(a,s) \in
\gtralphabet$ to be the transition of $B$ on letter $(a,s_a)$.  We abuse the notation
slightly in the following lemma and denote $\widehat{B}$ also by $B$.

\begin{lem}\label{lem:lcp-gcs}
	The action of a local cascade product $A = A_1 \lc \ldots \lc A_n$ can be simulated by
	the gcs $\Aseq = (A_1, \ldots, A_n)$ in the following sense: for any trace $t$, and any
	process $i$, $[A(t)]_i = [\Aseq(t)]_i$.
\end{lem}

\begin{proof}
  To be completely rigorous, the gcs in question is $\Aseq = (A_1,\hat A_2,\ldots,\hat A_n)$. 
  We prove the lemma by induction on $n$. If $n=1$, the statement is trivially true.

  For the inductive step, we assume that the lemma holds for $A' = A_1 \lc \ldots \lc
  A_{n-1}$, and $A'_{\mathrm{seq}} = (A_1,\hat{A}_2,\ldots,\hat{A}_{n-1})$.  Fix a trace
  $t$.  For any event $e$ in the trace and any process $i$, by induction hypothesis,
  $[A'(\Da e)]_i = [A'_{\mathrm{seq}}(\Da e)]_i$.  Hence if the label of the event $e$ in
  $\gslf_{A'_{\mathrm{seq}}}(t)$ is $(a,s)$ (for some $s \in
  \globalstates{A'_{\mathrm{seq}}}$), then the label of the \emph{same} event in
  $\chi_{A'}(t)$ corresponds to $(a,s_a)$ (see Remark~\ref{rem:forgetfulmorphism}).  By
  construction, transition on $(a,s_a)$ in the local cascade product component $A_n$ is
  same as that by $(a,s)$ in the corresponding global cascade sequence component
  $\hat{A}_n$.  Hence there is a natural correspondence between the run of $A_n$ (in $A$)
  over $\chi_{A'}(t)$ and the run of $\hat{A}_n$ (in $\Aseq$) over
  $\gslf_{A'_{\mathrm{seq}}}(t)$, and we can conclude that $[A(t)]_i = [\Aseq(t)]_i$.
\end{proof}

An important language-theoretic consequence of Lemma~\ref{lem:lcp-gcs} is that every language
accepted by $A = A_1 \lc \ldots \lc A_n$ is also accepted by the gcs 
$\Aseq = (A_1, \ldots, A_n)$.  

We now come to the main result of this section, that relates
the logic $\loctl$ with global cascade sequences of localized reset automata $U_2[p]$.
Recall that Theorem~\ref{sprtl-complete} gives an exact correspondence between first order
definable trace languages and local cascade products of $U_2[p]$ asynchronous automata for
aKR distributed alphabets. We generalize this language-theoretic decomposition result to
any distributed alphabet, using a global cascade sequence of the same distributed resets
instead of local cascade product.

\begin{thm}\label{thm:gcascade}
  A trace language is defined by a $\loctl$ formula if and only if it is accepted by a
  global cascade sequence of asynchronous reset automata of the form $U_2[p]$.
\end{thm}

The proof of Theorem~\ref{thm:gcascade} follows the same structure and re-uses elements of
the proof of Theorem~\ref{thm:lcascade}.

\begin{proof}
  First consider a global cascade sequence $A = U_2[p] \fc B$ where $B$ is a global
  cascade sequence.  Recall that $U_2[p]$ has two global states, say $\gs=\{1,2\}$,
  identified with the two local states of its $p$-component.  The input alphabet of the
  gcs $B$ is $\dgtralphabet$.  Suppose that the languages accepted by $B$ are
  $\loctl$-definable over $\dgtralphabet$.

  Let $\gslf\colon\traces\to\tracesdgtr$ be the global-state labelling function associated
  with $U_2[p]$ and its initial state, say, $1$.  By the global cascade principle
  (Theorem~\ref{thm:gcwpp2}), any language recognized by $A$ is a union of languages of
  the form $L_1 \cap \gslf^{-1}(L_2)$ where $L_1 \subseteq \traces$ is recognized by
  $U_2[p]$, and $L_2 \subseteq \tracesdgtr$ is recognized by $B$.  We have seen in the
  proof of Theorem~\ref{thm:lcascade} that $L_1$ is $\sprtl$ definable over alphabet
  $\dalphabet$.
  
  By assumption, we know that $L_2$ is $\loctl$ definable over alphabet $\dgtralphabet$
  and we need to prove that $\gslf^{-1}(L_2)$ is $\loctl$ definable over $\dalphabet$.
  This is done by structural induction on the $\loctl$-formula over $\dgtralphabet$
  defining $L_2$.  For a $\loctl$ event formula $\alpha$ over $\dgtralphabet$, we
  construct a $\loctl$ event formula $\hat{\alpha}$ over $\dalphabet$ such that for any
  trace $t \in \traces$ and any event $e$ in $t$, we have $t,e \models \hat{\alpha}$ if
  and only if $\gslf(t),e \models \alpha$.  The non-trivial case here is the base case of
  a letter formula from $\dgtralphabet$.  We let
  \begin{align*}
    \widehat{(a,2)} & = a \wedge \Y_p(R_2\vee(\neg R_1 \wedge ((\neg R_1)\Si_p R_2)))
    \\
    \widehat{(a,1)} & = a \wedge \neg\Y_p(R_2\vee(\neg R_1 \wedge ((\neg R_1)\Si_p R_2)))\,.
  \end{align*}
  The inductive cases are trivial, for instance $\widehat{\Y_i\alpha}=\Y_i\hat{\alpha}$.
  
  \bigskip

  Towards establishing the converse implication, we construct, for any $\loctl$ event
  formula $\alpha$, a gcs $A_\alpha$ from reset asynchronous automata of the form
  $U_2[p]$, which is such that for any trace $t$, event $e$ of $t$ and process
  $i\in\loc(e)$, the local state ${[A_\alpha(\dcset{e})]}_i$ determines whether
  $t,e\models\alpha$.  Again, this is done by structural induction on the $\loctl$ event
  formula $\alpha$.  The case of $\sprtl$-formulas, in view of Lemma~\ref{lem:lcp-gcs}, is
  already handled in (the proof of) Theorem~\ref{thm:lcascade} and we only need to deal
  with the inductive case where $\alpha$ is of the form $\alpha=\Y_j\beta$.
  
  By induction, a gcs of reset automata $A_\beta$ has been constructed, which provides the
  truth value of $\beta$ at any event.  We construct an asynchronous automaton $B =
  (\{Q_i\}, \{\lt_{(a,s_a)}\}, \qinit)$ over $\dgtralphabet$ such that $A_\alpha = A_\beta
  \fc U_2[j] \fc B$ as follows.  The middle $U_2[j]$ remembers whether some $j$-event
  already occured.  Concerning $B$, we let $Q_i = \{\top, \bot\}$ for all $i \in \pset$
  and, again, we denote a $P$-state $q$ as $\bot$ (resp.  $\top$) if $q_i = \bot$ (resp.
  $q_i = \top$) for all $i \in P$.  We let the initial state be $\qinit = \bot$.  Let
  $\gslf$ be the global-state labelling function associated with $A_\beta\fc U_2[j]$.  Let
  $t\in\traces$ be a trace, $e$ be an event in $t$, and $(a,s)$ be the label of $e$ in
  $\gslf(t)$.  Write $e_j$ the last $j$-event in $\Da e$ if it exists, i.e., if $\Da e\cap
  E_j\neq\emptyset$.  The local state $s_j$ determines whether $e_j$ exists, written
  $s_j\vdash\Y_j\top$, and in this case whether it satisfies $\beta$, written
  $s_j\vdash\Y_j\beta$. The transition functions of $B$ are:
  \begin{align*}
    \lt_{(a,s)} &= \text{reset to } \top &&\text{ if } s_j \vdash \Y_j\beta \\
    \lt_{(a,s)} &= \text{reset to } \bot &&\text{ if } s_j \not\vdash \Y_j\beta \,. 
  \end{align*}
  It is easy to see that $B$ is a gcs of copies of $U_2[p]$, one for each process
  $p\in\pset$.  These reset automata work independently of each other, each $U_2[p]$
  depends only on the global state information from $A_\beta\fc U_2[j]$ provided by
  $\gslf$.  Hence, $B$ is also a local cascade product of these $U_2[p]$.
  This completes the proof.
\end{proof}

Our subsequent result crucially uses the expressive completeness of $\loctl$ and its proof
is immediate from Theorems~\ref{thm:loctlfo} and~\ref{thm:gcascade}.  It is best seen as
an addition to the several characterizations of first-order definable trace languages
presented in Theorem~\ref{thm:lcascade-loctlv}.

\begin{thm}\label{thm: gcs FO}
  Let $\dalphabet$ be a distributed alphabet and let $L\subseteq\traces$ be a trace
  language.  Then $L$ is definable in first-order logic if and only if $L$ is accepted by a global
  cascade sequence of asynchronous reset automata of the form $U_2[p]$.
\end{thm}

\section{Asynchronous implementation of a global cascade sequence}\label{sec:
implementation global cascade}
We have already noted in Section~\ref{sec:cascade} that the global-state labelling
function $\gslf_\aaa$ associated with an asynchronous automaton $\aaa$ is not an
abstraction of the asynchronous transducer of $\aaa$, that is, it cannot be directly
computed by $\aaa$.  However, we will show how to construct an asynchronous automaton
$\aaa^{\gossip}$ which computes $\gslf_\aaa$ using the \emph{gossip automaton}.

Recall that the gossip automaton $\gossip=(\{\Upsilon_i\},\{\nabla_a\},v_{\textrm{in}})$
keeps track of the truth values all the constants of $\loctlv$: it computes the primary order
labelling function $\thetaY\colon\traces\to \traces[\lalphabet]$ which decorates each event
$e$ of a trace $t$ with the truth values $\gamma\in\lab=\{0,1\}^{\pset\times\pset}$ of all
constants $\Yleq{i}{j}$, i.e., for all $i,j\in\pset$, $\gamma_{i,j}=1$ if
$t,e\models\Yleq{i}{j}$ and $\gamma_{i,j}=0$ otherwise.

\paragraph*{Construction of $\aaa^{\gossip}$} 
Roughly speaking, in the automaton $\aaa^{\gossip}$, each process keeps track of its local
gossip state in the gossip automaton and the best global state of $A$ that it is aware of.
In fact, $\aaa^{\gossip}$ is realised as a $\thetaY$-restricted local cascade product of
the gossip automaton and an asynchronous automaton $\aaa^g$ -- the \emph{global-state
detector} -- derived from $A$ where each process in $\aaa^g$ keeps track of the best global
state of $A$ that it is aware of.  When processes synchronize in $\aaa^g$, they use the
$\lab$-labelling information to correctly update the best global state that they are aware
of at the synchronizing event.

So $\aaa^{\gossip} = \gossip \lcr \aaa^g$ where $\aaa^g$ is an asynchronous automaton over
$\dlalphabet$ derived from $\aaa$.  Therefore, for the construction, we only need to
describe $\aaa^g$.
Recall that $\aaa = (\{\ls_i\}, \{\lt_a\}, \sinit)$. Then $\aaa^g = (\{Q_i\},
\{\lt_{(a,\gamma)}\}, \qinit)$ where $Q_i = \gs$ for all $i \in \pset$, and $\qinit =
(\sinit,\ldots,\sinit)$. Before defining the transitions, we define for $(a,\gamma) 
\in \lalphabet$ the function $\globalstate_{(a,\gamma)}\colon Q_a \to \gs$ as follows:
$\globalstate_{(a,\gamma)}(q_a) = s \in \gs$ where, for each $i \in \pset$,
\[
s(i) = 
\begin{cases}
	q_a(j)(i) & \text{if there exists } j \in \loc(a) \text{ such that }	\gamma_{i,j} = 1 \\
  \sinit(i) & \text{otherwise.}
\end{cases}
\]
Note that, when $\gamma_{i,j}=1$ at some event $e$ of a trace $\thetaY(t)$, 
then $t,e\models\Yleq{i}{j}$ and process $j$ has the latest information about process 
$i$. Hence, the function $\globalstate_{(a,\gamma)}$ determines the best global-state that 
processes in $\loc(a)$ are collectively aware of.
We define the local transition functions of $\aaa^g$ 
by $\lt_{(a,\gamma)}(q_a)=q'_a$ where for all $i\in\loc(a)$ we set
\[
	q'_{a}(i)=\gt_a(\globalstate_{(a,\gamma)}(q_a)) \,.
\]
Recall that $\gt_a$ is the extension of the local transition function $\lt_a$ of $\aaa$ 
to global states in $\gs$.

In order to prove the correctness of the construction, we first introduce some notation.
Let $t = (E, \leq, \lambda) \in \traces$, and $i \in \pset$.  Then $\view{i}{t}$ is the
$i$-view of $t$ and it is defined by $\view{i}{t} = {\downarrow}E_i$.  It is easy to see
that if $\view{i}{t} \neq \emptyset$, then there exists $e \in E_i$ such that $\view{i}{t}
= \dcset{e}$.  We note that $\view{i}{t}$ is a trace prefix of $t$ and it represents
knowledge of the agent $i$ about $t$.

The next lemma shows that in $\aaa^{\gossip}$, each process keeps track of the best
global state of $A$ that it is aware of.

\begin{lem}\label{lem:GA-correctness}
  Let $t \in \traces$ with $\aaa^{\gossip}(t)=(\upsilon, q)$. Then 
  for every $i \in \pset$, $q_i = A(\view{i}{t})$.
\end{lem}

\begin{proof}
  Note that as $\aaa^{\gossip} = \gossip \lcr A^g$, $\aaa^{\gossip}(t)=(\upsilon, q)$
  implies that $\gossip(t)=v$ and $A^g(\thetaY(t)) = q$.  We prove the
  lemma by induction on the size of $t=(E, \leq, \lambda)$, that is, on $|E|$.  The base
  case of the empty trace is easy and skipped.

  Consider $t'=ta$, $\thetaY(t') = \thetaY(t) (a,\gamma)$ and let $(\upsilon, q) =
  \aaa^{\gossip}(t)$ and $(\upsilon', q') = \aaa^{\gossip}(t')$.  Clearly $A^g(\thetaY(t))
  = q$, $A^g(\thetaY(t')) = q'$ and $\delta_{(a,\gamma)}(q_a) = q'_a$.  By definition of
  the $(a,\gamma)$-transition function of $\aaa^g$, we have
  $q'_i=\Delta_a(\globalstate_{(a,\gamma)}(q_a))$ for each $i\in\loc(a)$.  Observe that,
  by the local nature of $a$-transition functions of asynchronous automata, we have
  $q'_i=q_i$ for $i \not\in \loc(a)$.

  By induction, for every $i \in \pset$, $q_i = A(\view{i}{t})$.  If $i \not\in \loc(a)$,
  $\view{i}{t'}=\view{i}{t}$.  As, in this case, we also have $q'_i = q_i$, we are done by
  the induction hypothesis.  Now we let $i \in \loc(a)$.  Let $e$ correspond to the last
  occurrence of $a$ in $t'$.  Then process $i$ participates in $e$.  As a result,
  $\view{i}{t'}=\dcset e$.

  We first study the global state $s=A(\Da e)\in\gs$.  For $i\in\pset$, let $e_i$ be the
  maximal $i$-event in $\Da e$ if it exists, i.e., if $\Da e\cap E_i\neq\emptyset$.  Fix a
  process $i \in \pset$.  If $e_i$ does not exist, then $s(i) = \sinit(i)$.  Otherwise
  there exists $j \in \loc(a)$ such that $e_i \leq e_j$.  Therefore $s(i)=[A(\Da
  e)]_i=[A(\view{j}{t})]_i=q_j(i)$ where the last equality is obtained using the induction
  hypothesis, $q_j = A(\view{j}{t})$.  From Theorem~\ref{thm:gossip for loctl}, we have
  $\gamma_{i,j} = 1$ (as $\thetaY(t') = \thetaY(t) (a,\gamma)$) and, using the definition
  of the function $\globalstate_{(a,\gamma)}$, it follows that
  $s=\globalstate_{(a,\gamma)}(q_a)$.

  Now, with $s' = A(\dcset e)$, we see that
  $s'=\gt_a(s)=\gt_a(\globalstate_{(a,\gamma)}(q_a))$.  As already observed,
  $q'_i=\gt_a(\globalstate_{(a,\gamma)}(q_a))$, for $i \in \loc(a)$.  The proof of the
  inductive step is now complete, as for $i \in \loc(a)$, $q'_i=s'=A(\dcset
  e)=A(\view{i}{t'})$.
\end{proof}

We now show that the asynchronous automaton $\aaa^{\gossip}$ simulates $\aaa$.  More
precisely, we define a simulation map $f\colon\Upsilon \times Q \to\gs$ by
$f(\upsilon, q)=s$ where, for $i \in \pset$, $s(i)=q_i(i)$.

\begin{lem}\label{lem:GA-sim}
  The action of $\aaa$ can be simulated by $\aaa^\gossip$ as, for $t \in \traces$,
  $\aaa(t) = f(\aaa^{\gossip}(t))$.
\end{lem}

\begin{proof}
  Let $t \in \traces$, $\aaa(t)=s$ and $\aaa^{\gossip}(t)=(\upsilon,q)$.  By
  Lemma~\ref{lem:GA-correctness}, for each $i \in \pset$, we have $q_i =
  \aaa(\view{i}{t})$ and hence, $q_i(i) = [\aaa(\view{i}{t})]_i=s(i)$.  This shows that
  $\aaa(t)=s=f(\aaa^{\gossip}(t))$.
\end{proof}

An important consequence of Lemma~\ref{lem:GA-sim} is that any language accepted by $\aaa$
is also accepted by $\aaa^\gossip$: if $L$ is accepted by $\aaa$ via $F \subseteq
\gs$ then it is accepted by $\aaa^\gossip$ via $f^{-1}(F)$.  Note that, as
$f^{-1}(F) = \Upsilon \times F'$ where $F' \subset \globalstates{A^g}$, the
accepting set $f^{-1}(F)$ is solely determined by $F'$ and does {\em not} depend on the
state reached by the gossip automaton.

Finally, we establish that $\aaa^{\gossip}$ computes the global-state labelling function
$\gslf_\aaa$. For this we fix the $\lab$-transducer $\widehat{\gossip} = 
(\{\Upsilon_i\},\{\nabla_a\},v_{\textrm{in}}, \{\nuY_a\colon \Upsilon_a \to \lab\})$ which
computes $\thetaY$.
Now we are ready to extend $\aaa^{\gossip}$ to a $\gs$-transducer
$\hat{\aaa}^{\gossip} = (\aaa^{\gossip}, \{\xi_a\})$. 
We define $\xi_a: \Upsilon_a \times Q_a \to \gs$ as follows:
$\xi_a (\upsilon_a, q_a) = \globalstate_{(a, \nuY_a(\upsilon_a))}(q_a)$.


\begin{thm}\label{thm:gat-correct}
  The transducer $\hat{\aaa}^{\gossip}$ computes the global-state labelling function
  $\gslf_\aaa$.
\end{thm}

\begin{proof}
  Let $t = (E,\leq,\lambda)\in \traces$ and $t' = \thetaY(t) = (E,\leq,(\lambda, \mu))\in
  TR(\dlalphabet)$.  Fix an event $e \in E$ with $\lambda(e)=a$ and $\mu(e)=\gamma$.  Note
  that, $\gslf_\aaa$ decorates the event $e$ with the additional information $s=A(\Da e)$.
  On the other hand, with ${\aaa}^{\gossip}(t)= (\upsilon,q)$, the transducer
  $\hat{\aaa}^{\gossip}$ decorates the event $e$ with the additional information
  $\xi_a(\upsilon_a, q_a) = \globalstate_{(a, \nuY_a(\upsilon_a))}(q_a)$.

  By Theorem~\ref{thm:gossip for loctl}, $\hat{\gossip}$ computes $\thetaY$ and we have
  $\nuY_a(\upsilon_a)=\gamma$.  Therefore it remains to show that $s=\globalstate_{(a,
  \gamma)}(q_a)$ to complete the proof.  But we have already seen in the proof of
  Lemma~\ref{lem:GA-correctness} (inductive step applied with $t=\Da e$ and
  $t'=\dcset{e}$) that $s=\globalstate_{(a,\gamma)}(q_a)$.
\end{proof}

Now we are ready to realize a global cascade sequence as an asychronous automaton.  We
associate with a gcs $\Aseq = (A_1, \ldots, A_n)$ an asychronous automaton
$\Aseq^{\gossip}$ and an asynchronous transducer $\hatAseq^{\gossip}$, whose constructions
extend the constructions of $A^{\gossip}$ and $\hat{A}^{\gossip}$ from $A$ in a natural
fashion.  In particular, $\Aseq^{\gossip}$ is a local cascade product $\gossip \lcr 
(A_1^g \lc \ldots \lc A_n^g)$.  The first {\em component} in this product is the gossip automaton
which computes $\thetaY$ and the first cascade product is $\thetaY$-restricted.  In the
subsequent components, each process keeps track of the best global state it is aware of
for the corresponding automaton in $\Aseq$.  In other words, we use the earlier
construction for each component automaton but keep a single copy of the gossip automaton.

Let $\dalphabet$ be the input (distributed) alphabet of $A_1$ (or that of $\Aseq$) with
total alphabet $\Sigma$.  Note that the input alphabet for $A_j$ is the output alphabet of
$A_{j-1}$.  We abuse the notation and use the total alphabet instead of the distributed
alphabet.  The distribution is anyway induced from the alphabet $\dalphabet$.  The input
alphabet for $A_j$ is $\Sigma \times \globalstates{A_1} \times \cdots \times
\globalstates{A_{j-1}}$.

\begin{rem}\label{rk: spirit of construction}
  It is important to keep in mind the special case $\Aseq=(A)$.  Our constructions of
  $\Aseq^{\gossip}$ and ${\hatAseq}^{\gossip}$ in this case are exactly identical to those
  of $A^{\gossip}$ and $\hat{A}^{\gossip}$.  The general case of a gcs is only
  notationally more involved.  It merely follows the constructions of $A^{\gossip}$ and
  $\hat{A}^{\gossip}$ in spirit and extends them naturally.
\end{rem}

In order to describe the complete construction, we need to define $A_1^g, \ldots, A_n^g$.
For brevity, we simply define $A_1^g$ and $A_2^g$ here and skip the remaining details.

Let $A_1 = (\{S_i\}, \{\delta_a\}, \sinit)$ and $A_2 = (\{S'_i\}, \{\delta'_{(a,s)}\},
\sinit')$. The definition of $A_1^g$ is verbatim to that of $A^g$ defined earlier.
In particular, $A_1^g$ is defined over $\dlalphabet$ and 
$A_1^g = (\{Q_i\}, \{\lt_{(a,\gamma)}\}, \qinit)$ where 
$Q_i = \gs$ for all $i \in \pset$, and $\qinit=(\sinit,\ldots,\sinit)$.
Further, we have defined the function $\globalstate_{(a, \gamma)}\colon Q_a \to \gs$
there which is used to construct the local transition functions of $A_1^g$ 
by $\lt_{(a,\gamma)}(q_a)=q'_a$ where for all $i\in\loc(a)$ we have
$q'_{a}(i)=\gt_a(\globalstate_{(a,\gamma)}(q_a))$.

Now we set $A_2^g = (\{Q'_i\}, \{\lt'_{(a,\gamma, q_a)}\}, \qinit')$ where, for all $i
\in \pset$, $Q'_i = \gs'$ is the set of global states of $A_2$ and 
$\qinit'=(\sinit',\ldots,\sinit')$.
As expected, we use the `same' function
$\globalstate'_{(a, \gamma)} \colon Q'_a\to \gs'$ defined by
$\globalstate'_{(a, \gamma)}(q'_a) = s' \in \gs'$ where, for each $i \in \pset$,
\[
s'(i) = 
\begin{cases}
	q'_{a}(j)(i) & \text{if } \exists j \in \loc(a) \text{ such that }	\gamma_{i,j} = 1 \\
  \sinit'(i) & \text{otherwise.}
\end{cases}
\]
Now we define the local transitions of $A_2^g$
by $\lt_{(a,\gamma,q_a)}(q'_a)=q''_a$ where for all $i\in\loc(a)$ we have
\[
q''_{a}(i)=\gt'_{(a,\globalstate_{(a,\gamma)}(q_a))} (\globalstate'_{(a,\gamma)}(q'_a))\,.
\]
Recall that $\gt'_{(a,s=\globalstate_{(a,\gamma)}(q_a))}$ is the extension of the 
local transition function $\lt'_{(a,s)}$ of $A_2$ to global states in $\gs'$.

In general, in $A_p^g$ -- the global-state detector of $A_p$, the local state-set of each
process is simply the set $\globalstates{A_p}$ of global states of $A_p$ and each
process starts in the initial global-state of $A_p$.  As mentioned earlier, each process
in $A_p^g$ simply records the best global state of $A_p$ that it is aware of.  At a
synchronizing event, the participating processes simply use the $\lab$-labelling
information $\gamma$ to correctly compute, using the $\globalstate_{(a,\gamma)}$ function,
the best {\em collective} global-state of $A_p$ that they are aware of just prior to the
synchronization.  These participating processes from the earlier components $A_q^g$ for $q
< p$ can also similarly compute best collective global-state of $A_q$ prior to the {\em
current} synchronization that they are aware of and make it available to $A_p^g$ via the
local-cascade mechanism.  This allows $A_p^g$ to correctly simulate the operational
global-cascade mechanism used by the component $A_p$ in $\Aseq$.

Now we simply describe the final form of $\Aseq^{\gossip} = 
\gossip \lcr (A_1^g \lc \ldots \lc A_n^g)$. It is easy to see that,
$\globalstates{\Aseq^{\gossip}} = \Upsilon_{\pset} \times \globalstates{A_1^g} \times
\ldots \times \globalstates{A_n^g}
= \Upsilon_{\pset} \times \globalstates{A_1}^{\pset} \times\ldots\times 
\globalstates{A_n}^{\pset}$.
As a result, we write a global state of $\Aseq^{\gossip}$ as 
$(\upsilon, q_1, \ldots, q_n)$
where $\upsilon \in \Upsilon_{\pset}$ and 
$q_j \in\globalstates{A_j}^{\pset}= \prod_{i \in \pset} \globalstates{A_j}$.

We next exhibit a simulation of the global cascade sequence $\Aseq$ by
the asynchronous automaton $\Aseq^{\gossip}$. We define 
$\fseq: \globalstates{\Aseq^{\gossip}} \to \globalstates{\Aseq}$ as follows:
$\fseq((\upsilon,q_1,q_2,\ldots,q_n))=(s_1, s_2, \ldots, s_n)$ where
$s_p(i)=q_p(i)(i)$ for each $1 \leq p \leq n$ and each $i\in \pset$.

Finally, we enrich $\Aseq^{\gossip}$ to get the asynchronous
$\globalstates{\Aseq}$-transducer ${\hatAseq}^{\gossip}=(\Aseq^{\gossip},\{\xi_a\})$ 
over $\dalphabet$ which computes the global-state
labelling function $\zeta_{\Aseq}$. For each
$a\in\alphabet$,  we define $\xi_a$ from $a$-states of $\Aseq^{\gossip}$ as
follows:  with $\gamma = \nuY_a(\upsilon_a)$,
\[
\xi_a(\upsilon_a, q^a_1, \ldots, q^a_n) = (\globalstate_{(a,\gamma)}(q^a_1), \ldots,
\globalstate_{(a,\gamma)}(q^a_n))
\]

In view of Remark~\ref{rk: spirit of construction}, our next set of results regarding the
correctness of $\Aseq^{\gossip}$ and ${\hatAseq}^{\gossip}$ is hardly surprising.  Their
proofs are almost verbatim duplicates of the proofs of the corresponding results about
$A^{\gossip}$ and $\hat{A}^{\gossip}$, differing only in the bookkeeping notation needed
to talk about a sequence.

\begin{lem}\label{lem:GAseq-correctness}
	Let $t \in \traces$ with $\Aseq^{\gossip}(t)=(\upsilon, q_1, \ldots, q_n)$.  Then, for
	every $i \in \pset$, $\Aseq(\view{i}{t})=(q_1(i), \ldots, q_n(i))$.
\end{lem}

\begin{lem}\label{lem:GAseq-sim}
  The action of $\Aseq = (A_1, \ldots, A_n)$ is simulated by $\Aseq^\gossip$ in the
  following sense: for $t \in \traces$, $\Aseq(t) = \fseq(\Aseq^\gossip(t))$.
\end{lem}

Lemma~\ref{lem:GAseq-sim} implies that  every language accepted by
$\Aseq$ is also accepted by $\Aseq^\gossip$.

\begin{thm}\label{thm:gatseq-correct}
  The transducer ${\hatAseq}^{\gossip}$ computes the global-state labelling function
  $\zeta_{\Aseq}$.
\end{thm}

\section{In lieu of a conclusion}\label{sec: conclusion}

An intriguing question, first asked by Adsul and Sohoni \cite{DBLP:conf/fsttcs/AdsulS04}
and revived by the results in this paper, is the following.  Zielonka's theorem
\cite{zielonka1987notes} states that every trace language that is accepted by a (diamond)
automaton, is accepted by an asynchronous automaton.  Equivalently, every trace language
that is recognized by a morphism to a tm, is recognized by an asynchronous morphism to an
atm.  Ebinger and Muscholl's theorem \cite{ebinger1996logical} also states that
first-order definable trace languages are exactly those that are recognized by an
aperiodic tm.  In view of the importance of first-order definability --- and of the vast
literature concerning that class of languages, it would be interesting to know whether an
\emph{aperiodic Zielonka} theorem holds, or for which distributed alphabets it does.  Such
a theorem would state that the first-order definable trace languages are exactly those
that are recognized by an asynchronous morphism to an aperiodic atm.

The asynchronous Krohn-Rhodes property introduced in Section~\ref{sec: decomposition} is
stronger: over aKR distributed alphabets, a trace language $L$ recognized by a tm $(X,M)$
is also recognized by a local cascade product of localized reset automata and localized
permutation automata of the form $G[p]$, where $G$ is a simple group dividing $M$.  If $L$
is first-order definable, the tm can be chosen such that $M$ is aperiodic, so $L$ is
accepted by a local cascade product of localized reset automata, and hence recognized by
an aperiodic atm.  It is conceivable that certain distributed alphabets would have the
aperiodic Zielonka property without being aKR. It is also conceivable that certain
distributed alphabets would fail to be aKR, yet would have that property when restricted
to first-order definable languages, that is, to languages recognized by an aperiodic tm (a
property that we term \emph{aperiodic aKR}).

We showed in Section~\ref{sec: acyclic} that acyclic architectures have the stronger aKR
property.  Apart from that, we do not know which distributed alphabets have the aperiodic
Zielonka or the asynchronous Krohn-Rhodes property, nor even whether all distributed
alphabets have these properties.  However, the results in Section~\ref{sec:aperiodic} have
the following consequence.  For a fixed distributed alphabet $\dalphabet$, recall (from
Section~\ref{sec:aperiodic}) that $\thetaY$ is the primary order labelling function which
decorates every event of a trace $t$ with the truth values of the constants $Y_i\le Y_j$.
It follows from Remark~\ref{rk: gossip-like} and Theorem~\ref{thm:lcascade-loctlv} that
if, for some distributed alphabet $\dalphabet$, the function $\thetaY$ can be computed by
an aperiodic asynchronous automaton (resp.\ by an asynchronous wreath product of
asynchronous transformation monoids of the form $U_2[p]$), then $\dalphabet$ has the
aperiodic Zielonka (resp.\ aperiodic aKR) property.  Notice that such an aperiodic
asynchronous automaton is easily constructed for acyclic architectures, thus providing an
alternate proof of Theorem~\ref{thm:acyclickr} in the aperiodic case.

\bibliographystyle{alphaurl}
\bibliography{main}

\end{document}